\documentclass[journal]{IEEEtran}
\usepackage{amsmath, amsthm, amsfonts, amssymb}
\usepackage{graphicx}
\usepackage{textcomp}
\usepackage[utf8]{inputenc} 
\usepackage[T1]{fontenc}    
\usepackage{hyperref}       
\usepackage{url}            
\usepackage{booktabs}       
\usepackage{amsfonts}       
\usepackage{nicefrac}       
\usepackage{microtype}      
\usepackage{algorithm}
\usepackage{algpseudocode}
\usepackage{booktabs}
\usepackage{comment}
\usepackage{tabularx}
\usepackage{paralist}
\usepackage{diagbox}
\usepackage{enumitem}
\setlist[itemize]{leftmargin=*, itemsep=0pt, parsep=0pt}

\usepackage[backend=biber, sorting=ynt]{biblatex}
\newtheorem{lemma}{Lemma}
\newtheorem{theorem}{Theorem}

\usepackage{array} 
\def\BibTeX{{\rm B\kern-.05em{\sc i\kern-.025em b}\kern-.08em
    T\kern-.1667em\lower.7ex\hbox{E}\kern-.125emX}}

\begin{document}

\title{Permissioned Blockchain-based Framework for Ranking Synthetic Data Generators}

\author{Narasimha Raghavan Veeraragavan, Mohammad Hossein Tabatabaei, Severin Elvatun, Vibeke Binz Vallevik, Siri Larønningen, Jan F Nygård
\thanks{Narasimha Raghavan Veeraragavan, Mohammad Hossein Tabatabaei, Severin Elvatun, Siri Larønningen, and Jan F Nygård are with Cancer Registry of Norway, Norwegian Institute of Public Health (e-mail: \{nara, moht, sela, sla, jfn\}@kreftregisteret.no)}
\thanks{Vibeke Binz Vallevik is with DNV Healthcare (e-mail: Vibeke.Binz@dnv.com)}
\thanks{Jan F Nygård is also with UiT - Artic University of Norway (email: jan.f.nygard@uit.no)}
}

\maketitle

\begin{abstract}
Synthetic data generation is increasingly recognized as a crucial solution to address data-related challenges such as scarcity, bias, and privacy concerns. As synthetic data proliferates, the need for a robust evaluation framework to select a synthetic data generator becomes more pressing given the variety of options available. In this research study, we investigate two primary questions: 1) How can we select the most suitable synthetic data generator from a set of options for a specific purpose? 2) How can we make the selection process more transparent, accountable, and auditable? To address these questions, we introduce a novel approach in which the proposed ranking algorithm is implemented as a smart contract within a permissioned blockchain framework called Sawtooth. Through comprehensive experiments and comparisons with state-of-the-art baseline ranking solutions, our framework demonstrates its effectiveness in providing nuanced rankings that consider both desirable and undesirable properties. Furthermore, our framework serves as a valuable tool for selecting the optimal synthetic data generators for specific needs while ensuring compliance with data protection principles.\end{abstract}

\section{Introduction}
\label{sec:introduction}
Synthetic data generation is increasingly acknowledged as a vital solution to overcome data-related challenges such as scarcity, bias, and privacy concerns. By generating artificial data sets that reflect real-world datasets, synthetic data offers a solution to replace the real data, balancing the need to replicate key data attributes while safeguarding sensitive information. This balance is pivotal, as it requires a careful blend of replicating essential real data aspects and adhering to ethical and legal standards.

The essence of synthetic data generation lies in its purpose; the 'why' that dictates the 'how'. As highlighted in various studies~\cite{gonzales2023}, including those in the healthcare domain~\cite{VALLEVIK2024105413}, different purposes for generating synthetic datasets require different levels of fidelity and characteristics. From software testing to complex data analysis, the purpose shapes the required properties of the synthetic data. This purpose-driven approach is critical to determine which real data properties must be replicated and which should be omitted, making the generation technique an important decision. 

Despite the existence of several studies on the generation of synthetic tabular data~\cite{fonseca2023tabular}~\cite{hernadez2023synthetic}, there is only a paucity of research focused on the systematic ranking of these techniques. Existing studies~\cite{arnold2020really,chundawat2022tabsyndex,yan2022multifaceted, pathare2023comparison} focus primarily on ranking generators based on their ability to replicate properties from real to synthetic data. However, these studies overlook the evaluation of undesirable properties in their rankings. This gap signifies a lack of assurance that principles such as purpose limitation, data minimization, and data protection by design and default~\cite{GDPR} are adhered to in the synthetic data. Consequently, the rankings might misguide decision-makers by highlighting generators that do not comprehensively evaluate all aspects of regulatory compliance and data protection principles. 

The selection process of synthetic data generators must be meticulously aligned with data protection regulations, such as the EU General Data Protection Regulation (GDPR), particularly its principles of data minimization, purpose limitation, and data protection by design and default~\cite{GDPR}. GDPR underscores the importance of protecting personal data and ensuring that data processing is purpose-specific, avoiding excessive data collection or usage. This regulatory framework necessitates a method in which synthetic data generation adheres to these principles, guaranteeing compliance and ethical integrity.

In navigating the complex landscape of synthetic data generators, an indexing metric is indispensable~\cite{hernandez2022synthetic,fonseca2023tabular,murtaza2023synthetic}.  This metric should be designed to rank different synthetic data generators based on their ability to meet specific purposes, evaluating and comparing various generators against a set of criteria. This evaluation should encompass both the technical efficacy of each generator and its alignment with legal and ethical standards, promoting a comprehensive approach to generator selection. 

Moreover, the European Union's Artificial Intelligence Act (EU AI Act) proposal~\cite{act2021proposal} requires transparency in AI systems, requiring clear disclosures about their capabilities and the governance of the data they use. Organizations employing synthetic data generation processes must ensure transparency, accountability, and auditability to comply with both the EU AI Act and GDPR. This is particularly crucial when personal data is utilized to train generators. 

In response, organizations can adopt permissioned blockchain frameworks~\cite{tabatabaei2023understanding} to integrate the ranking process, thereby adhering to regulatory requirements and ethical standards. This approach not only ensures compliance with the EU AI Act and GDPR but also enhances the transparency, accountability, and auditability of the selection process for synthetic data generators. This is particularly crucial when personal data is used to train generators or when synthetic data generators are applied in critical sectors such as healthcare and criminal justice.

Permissioned blockchain frameworks have been studied extensively in the literature~\cite{min2016permissioned, al2019blockchain, monrat2020performance, polge2021permissioned, qi2021bidl, peng2022neuchain, capocasale2023comparative, li2023fisco}. A permissioned blockchain is a type of blockchain where access is controlled by a consortium or organization, and only a limited number of participants have the permission to validate block transactions and append the transactions to the ledger, unlike permissionless blockchains. Furthermore, permissioned blockchains typically offer higher transaction throughput and low latency compared to permissionless blockchains~\cite{bakos2021permissioned}.

Permissioned blockchain frameworks offer several beneficial features for managing the process of ranking synthetic data generators. By controlling access, permissioned blockchains ensure that only authorized entities can contribute and access the blockchain. This feature can protect the process of ranking synthetic data generators from unauthorized access. Even though permissioned, the blockchain provides a transparent ledger of all transactions. This transparency means that the criteria used for ranking, the ranking outcome, and the ranking algorithm itself can be tracked and audited, building trust among users with respect to the quality and reliability of synthetic data. Furthermore, smart contracts can automate the ranking and validation process of synthetic data generators, executing predefined algorithms when the appropriate conditions are met.   

Using a permissioned blockchain technology, our aim is to build a framework that not only meets the technical and operational requirements for effective synthetic data generation, but also upholds the core principles of accountability, transparency, and auditability. This framework also supports data minimization, purpose limitation, and data protection by design and default in a secure, regulated, and efficient manner. 

The summary of our main contributions is listed as follows:
\begin{itemize}
\item We introduce a novel ranking algorithm with a well-grounded theory that advances the state-of-the-art in ranking synthetic data generators. This algorithm distinguishes itself by simultaneously considering desired and undesired properties while also taking into account the significance of each property relative to a specific purpose.

\item  We demonstrate the integration of this ranking algorithm within a permissioned blockchain framework, specifically Sawtooth~\cite{Sawtooth}, to enhance the accountability, auditability, and transparency of the synthetic data generation ranking process for an organisation. 

\item We conduct a comparative analysis of our proposed ranking algorithm against established baseline algorithms to highlight the advantages and improvements of our proposed ranking algorithm.

\item We conduct experiments to validate our ranking algorithm's practicality and performance when implemented as part of a permissioned blockchain framework.

\end{itemize}
The rest of the paper is organized as follows. Section~\ref{sec:systemModel} introduces our system model. Section~\ref{sec:proposed} proposes our mathematical framework, workflow and algorithms to rank synthetic data generators. Section~\ref{sec:conversion} shows how the ranking algorithm proposed in Section~\ref{sec:proposed} can be integrated into a permissioned blockchain (Sawtooth) ~\cite{Sawtooth}. Section~\ref{sec:experiments} presents our experiments that validate the proposed framework and a comparison with an existing baseline ranking solution. Section~\ref{sec:discussions} examines the correctness, computational complexity, scalability, security and privacy analysis, limitations, and social and ethical implications of our framework. The position of our work in advancing the state-of-the-art in synthetic data evaluation is discussed in Section~\ref{sec:relatedWork}. Section~\ref{sec:futureWork} explores different directions for the further development of the proposed framework. Finally, Section~\ref{sec:conlcusion} concludes the paper.

\section{System Model}
\label{sec:systemModel}
Our system model is designed to integrate synthetic data generation with GDPR principles and an indexing metric. The model is structured around key roles such as Product Manager, Data Scientist, multiple stakeholders, and External Auditor of an organization as shown in Figure~\ref{fig:system_model}.
\begin{figure}[htbp]
\caption{An Overview of a System Model for Ranking Synthetic Data Generators in Permissioned Blockchain, Covering Various User Roles.}
\centerline{\includegraphics[width=\columnwidth]{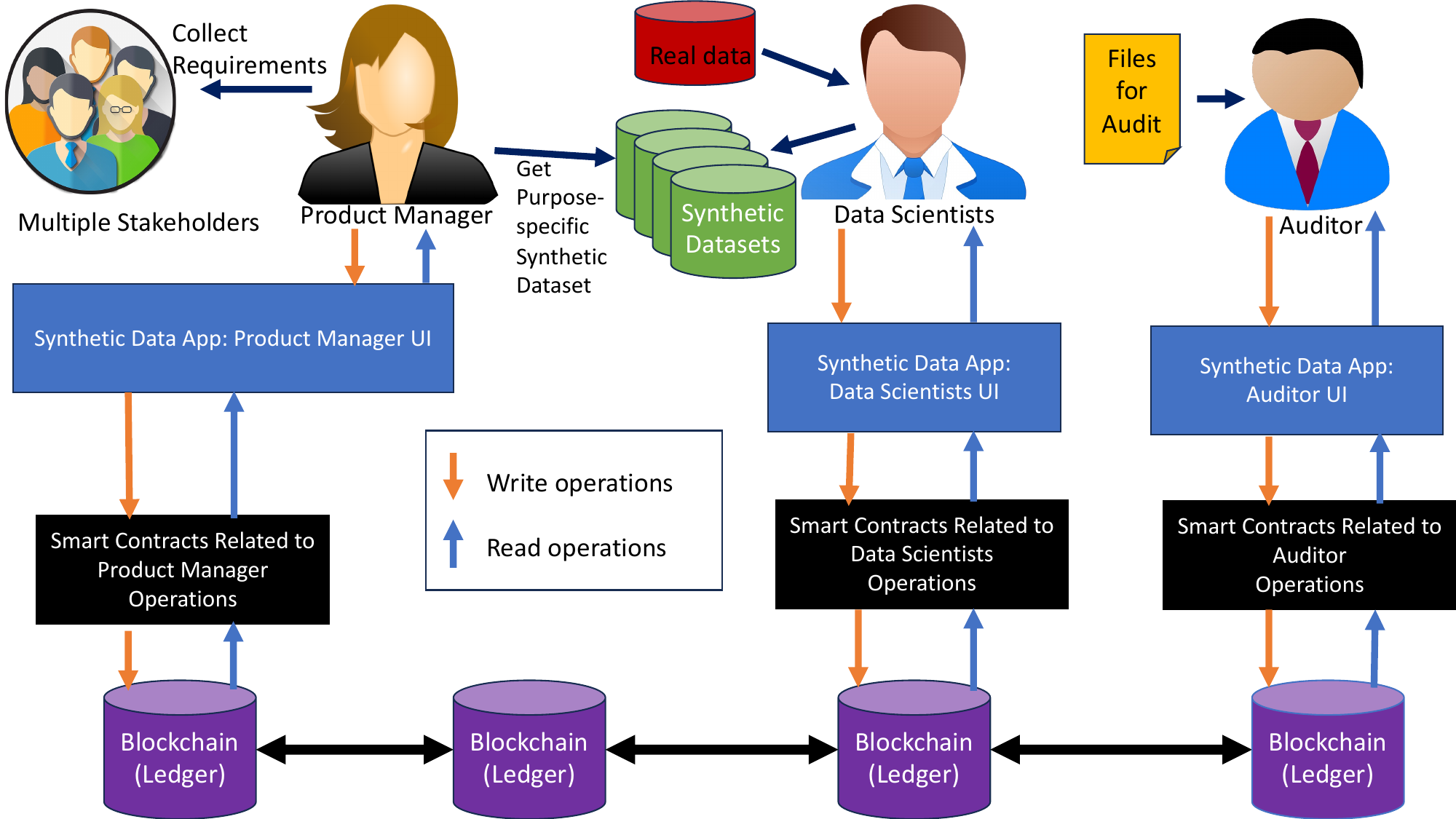}}
\label{fig:system_model}
\end{figure}

The Product Manager collaborates with all relevant stakeholders to define the specifications for generating synthetic data. These specifications encompass seven key elements: 
\begin{itemize}
\item The purpose for which the synthetic datasets are generated. 
\item The quality indicators that categorize the qualities to be assessed in the synthetic data relevant to the identified purpose. 
\item The desired properties of real data sets within each quality indicator that need to be replicated in synthetic data sets. 
\item The undesired properties within each quality indicator of real datasets that should not be replicated in the synthetic datasets. 
\item Weighted scoring applied to the quality indicators. The weighted score represents the relative importance of quality indicators. 
\item Weighted scoring applied to metrics that represent the desired properties within these indicators. The weighted scoring indicates the relative importance of the metrics that represent the desired properties. 
\item Weighted scoring applied to metrics that represent the undesirable properties within these indicators. 
\end{itemize}

These specifications are written in the blockchain. The ultimate goal of the Product Manager is to obtain synthetic data that aligns closely with these detailed specifications. 

The Data Scientist retrieves the Product Manager's specifications from the blockchain and conducts experiments with various synthetic data generators. Then, the Data Scientist evaluates the effectiveness of these generators against the metrics defined in the specifications and then records the results of these evaluations back to the blockchain.

The blockchain system utilizes a smart contract that accesses the specifications of the Product Manager and the evaluation results of the synthetic data generators from the Data Scientist. This contract computes the ranking of the synthetic data generators, aligning them with the Product Manager's specifications. The ranking results are subsequently written back to the blockchain.

Upon reading the ranking results from the blockchain, the Data Scientist identifies the most suitable synthetic data generated by the top-ranked generator according to the given specifications. This identified synthetic data set is then provided to the Product Manager.

An external or independent auditor accesses the Product Manager's specifications and retrieves the evaluation results and the ranking list from the Data Scientists. The auditor verifies whether these inputs from the Product Manager and Data Scientist correspond with the records on the blockchain. Finally, the auditor writes the verification results back into the blockchain.

Our system model incorporates a permissioned blockchain network executing a consensus protocol such as the Practical Byzantine Fault Tolerant (PBFT)~\cite{castro1999practical}, which offers robust guarantees against Byzantine faults. To adhere to the 3f + 1 rule, where f represents the number of Byzantine nodes, our network includes at least four nodes. This configuration ensures tolerance of at least one Byzantine fault node. Each node, which hosts the ledger and smart contracts, is strategically deployed in different departments within the organization, including the compliance department.

The deployment of nodes across various departments facilitates the assignment of specific roles to each node. For instance, if the product manager and the data scientists are in separate departments, one node will be deployed in the department of the product manager and another in that of the data scientist. Consequently, each department, based on its specific role, can utilize client applications to interact with the blockchain network.

\subsection{Threat Model}
\label{sec:threats}
\begin{itemize}
\item \textbf{Repudiation attack}: The product manager could deny setting incorrect criteria that are not aligned with the purpose, which can corrupt the ranking algorithm and violate the GDPR principles. Likewise, data scientists could deny providing incorrect evaluation metrics and associated values and not using the correct output of the ranking algorithm. The blockchain developer can write incorrect smart contracts and the auditor can deny his/her incorrect verification result. 
\item \textbf{Colluding attack}: A minority of departments can attempt to replace the correct nodes that host the ledger and smart contracts with incorrect nodes. 
\item \textbf{Poisoning attack}: Accidental or intentional manipulation of the stored inputs and the output of the ranking algorithm. 
\end{itemize}

\section{Proposed Scheme}
\label{sec:proposed}
In this section, we first introduce the definition and concepts required to understand the proposed algorithm. 
\subsection{Definitions and Concepts}
\subsubsection{Purposes and Quality Indicators}

The quality of synthetic data in synthetic data generation starts with two key elements: a) purposes and b) quality indicators (QIs). To formalize this concept, we define a purpose $P$ as a set of individual purposes $p$, each associated with a distinct set of QIs, as shown in Equations~\ref{eq:purpose} and~\ref{eq:QIs}. In this case, each purpose can be viewed as an independent scenario/use case.  

\begin{equation}
    P = \{p_1, p_2, \cdots, p_{|P|}\}
\label{eq:purpose}
\end{equation}

For each specific purpose $p_i$ within $P$, we define a corresponding collection of $QI^{p_i}$ shown in Equation~\ref{eq:QIs}. 
\begin{equation}
QI^{p_i}= \{QI^{p_i}_{1}, QI^{p_i}_{2}, \ldots, QI^{p_i}{_{|QI^{p_i}|}}\}
\label{eq:QIs}
\end{equation}

In Equation~\ref{eq:QIs}, each ${QI}^{p_i}_j$ represents a unique quality indicator $j$ of synthetic data, contributing to the general assessment of the quality of synthetic data $QI^{p_i}$, aligned with the defined purpose $p_i$.

\begin{equation}
    QI_{\text{total}} = \bigcup_{i=1}^{|P|} QI^{p_i}
    \label{eq:QItotal}
\end{equation}

Equation~\ref{eq:QItotal} defines $QI_{total}$ which represents the total set of unique Quality Indicators for all purposes within the organization.  

\subsubsection{Relative Importance of Quality Indicators}
Weights are assigned to each ${QI}^{p_i}_j$ in $QI^{p_i}$ to signify their relative importance in the overall quality assessment of synthetic data for the purpose $p_i$. These weights are structured in a matrix $W$ to reflect their significance for various purposes. 

$W$ is defined as a matrix with columns representing the elements from the Quality Indicators $QI_{total}$ and the rows represent the elements from the purposes $P$,
\begin{equation}
W = \begin{bmatrix}
w_{11} & w_{12} & \cdots & w_{1|QI_{total}|} \\
w_{21} & w_{22} & \cdots & w_{2|QI_{total}|}\\
\vdots & \vdots & \ddots & \vdots \\
w_{|P|1} & w_{|P|2} & \cdots & w_{|P||QI_{total}|}
\end{bmatrix}
\label{eq:QIWeightsPositive}
\end{equation}
In Equation~\ref{eq:QIWeightsPositive}, $w_{ij}$ denotes the weight of $j^{th}$ desirable quality indicator ${QI}^{p_i}_j$ for the $i^{th}$ purpose ($p_i$). The size of matrix $W$  is $|P| \times |QI_{total}|$ where $|P|$ is the total number of purposes, $|QI_{total}|$ is the total count of different quality indicators for all purposes. 

\begin{equation}
W = w_{ij} \in [0, 1] \quad  \forall i \in \{1, 2, \ldots, |P|\},  \forall j \in \{1, \ldots, |QI_{total}|\}
\label{eq:QIweights}
\end{equation}

In equation~\ref{eq:QIweights}, $W$ is a matrix where each element $w_{ij}$ represents the weight assigned to the quality indicator $j^{th}$ for the purpose $i^{th}$. The value of each $w_{ij}$ is within the range $[0,1]$, which indicates the importance of each quality indicator relative to a specific purpose $p_i$. The index $j$ corresponds to each quality indicator within the total set of quality indicators $QI_{total}$. 

The weights in each row of $W$ should be normalized to ensure that they sum up to 1 providing a balanced assessment of Quality Indicators $QI^{p_i}$ within each purpose $p_i$ as shown in Equation~\ref{eq:NormalizationConstraintPositiveQIs}.

\begin{equation}
\sum_{j=1}^{|QI^{p_i}|} w_{ij} = 1 \quad \forall i \in \{1, 2, \ldots, |P|\}
\label{eq:NormalizationConstraintPositiveQIs}
\end{equation}

The normalization constraint defined in Equation~\ref{eq:NormalizationConstraintPositiveQIs} ensures that, for each purpose, the collective weight of all Quality Indicators is balanced, providing a comprehensive and equitable approach to evaluating each Quality Indicator.
\subsubsection{Metrics, Desired and Undesired properties} 
While Quality Indicators (QIs) are higher-level abstractions representing different dimensions or aspects of quality in synthetic data, metrics offer a more granular view. They are practical, quantifiable measures that collectively assess the performance of a quality indicator (QI). 

Metrics represent both desired and undesired properties in synthetic data. Desired properties are those aspects of real data that should be replicated in synthetic data, while undesired properties are aspects of real data that should not be replicated.

For each Quality Indicator $QI^{p_i}_{j}$ of purpose $p_i$, a corresponding set of metrics $M^{QI^{p^{+}_i}_{j}}$ and $M^{QI^{p^{-}_i}_{j}}$ are assigned to quantitatively measure the characteristics of the synthetic data. Formally, the set of metrics for $M^{QI^{p^{+}_i}_{j}}$ and $M^{QI^{p^{-}_i}_{j}}$ is defined as in Equations~\ref{eq:M+} and~\ref{eq:M-}.

\begin{equation}
M^{QI^{p^{+}_i}_{j}} = \{m_1^{QI^{p^{+}_i}_{j}}, m_2^{QI^{p^{+}_i}_{j}}, \ldots, m^{QI^{p^{+}_i}_{j}}_{|M^{QI^{p^{+}_i}_{j}}|}\}
\label{eq:M+}
\end{equation}

\begin{equation}
M^{QI^{p^{-}_i}_{j}} = \{m_1^{QI^{p^{-}_i}_{j}}, m_2^{QI^{p^{-}_i}_{j}}, \ldots, m^{QI^{p^{-}_i}_{j}}_{|M^{QI^{p^{-}_i}_{j}}|}\}
\label{eq:M-}
\end{equation}

$M^{QI^{p^{+}_i}_{j}}$ denotes the set of metrics that measure desirable properties for the purpose $p_i$. In contrast, $M^{QI^{p^{-}_i}_{j}}$ represents the set of metrics that evaluate undesirable properties for the same purpose $p_i$ within the quality indicator $QI_j^{p_i}$. 

\begin{equation}
\begin{aligned}
M^{QI^{p^{+}_i}_{j}} \cap M^{QI^{p^{-}_i}_{j}} &= \emptyset \\
\quad \forall i \in \{1, 2, \ldots, |P|\}, \forall j \in \{1, \ldots, |QI_{\text{total}}|\}
\end{aligned}
\label{eq:NonOverlappingMetrics}
\end{equation}

Equation~\ref{eq:NonOverlappingMetrics} represents the nonoverlapping constraint in which the metrics measuring desirable properties are distinct and do not overlap with those measuring undesirable properties for any given purpose in any quality indicator. It is a critical aspect of ensuring that the assessment of synthetic data quality is clear, unambiguous, and accurately reflects the distinct characteristics of desirable and undesirable properties. 

\begin{equation}
M_{\text{total}}^{+} = \bigcup_{i=1}^{|P|} \bigcup_{j=1}^{|QI_{\text{total}}|} M^{QI^{p^{+}_i}_{j}}
\label{eq:MtotalPlus}
\end{equation}

\begin{equation}
M_{\text{total}}^{-} = \bigcup_{i=1}^{|P|} \bigcup_{j=1}^{|QI_{\text{total}}|} M^{QI^{p^{-}_i}_{j}}
\label{eq:MtotalMinus}
\end{equation}

Equations~\ref{eq:MtotalPlus} and~\ref{eq:MtotalMinus} represent the unified set that contains all unique metrics used for all purposes and quality indicators to evaluate desirable properties and undesirable properties, respectively. 

\begin{equation}
M_{\text{total}} = M_{\text{total}}^{+} \cup M_{\text{total}}^{-}
\label{eq:Mtotal}
\end{equation}

Equation~\ref{eq:Mtotal} provides a comprehensive collection of all metrics used in the assessment of Quality Indicators for any purpose. 

\subsubsection{Relative Importance of Metrics for each Quality Indicator}
The importance of each metric in evaluating the performance of a Quality Indicator (QI) is quantified by assigning weights. These weights reflect how significantly each metric contributes to the assessment of QI for different purposes. 

The weights are organized into two matrices, one for metrics associated with desired properties and the other for metrics related to undesired properties. 

For metrics associated with desirable properties, we define a weight matrix $W^{M^{+}}$ as shown in Equation~\ref{eq:MetricWeightsPositive}.

\begin{equation}
W^{M^{+}} = \begin{bmatrix}
w_{11}^{M^{+}} & w_{12}^{M^{+}} & \cdots & w^{M^{+}}_{1|M_\text{total}^{+}|} \\
w_{21}^{M^{+}} & w_{22}^{M^{+}} & \cdots & w^{M^{+}}_{2|M_\text{total}^{+}|}\\
\vdots & \vdots & \ddots & \vdots \\
w_{|P|1}^{M^{+}} & w_{|P|2}^{M^{+}} & \cdots & w^{M^{+}}_{|P||M_\text{total}^{+}|}
\end{bmatrix}
\label{eq:MetricWeightsPositive}
\end{equation}

In equation~\ref{eq:MetricWeightsPositive}, $w^{M^+}_{ij}$ denotes the weight of the $j^{th}$ metric for the purpose $i^{th}$, with respect to desirable properties. 

\begin{equation}
f: W^{M^{+}} \rightarrow M_{\text{total}}^{+}
\label{eq:MappingFunction}
\end{equation}
Equation~\ref{eq:MappingFunction} establishes the relationship between the weights in $W^{M^{+}}$ and their corresponding metrics in $M_{\text{total}}^{+}$. This association focuses on metrics that assess desirable properties for each purpose and quality indicator, thereby ensuring that the weights accurately reflect the relevance of each metric in overall evaluation of synthetic data quality. 

Similarly, for metrics associated with undesirable properties, we define the weight matrix $W^{M^{-}}$ as shown in Equation~\ref{eq:MetricWeightsNegative}.
\begin{equation}
W^{M^{-}} = \begin{bmatrix}
w_{11}^{M^{-}} & w_{12}^{M^{-}} & \cdots & w^{M^{-}}_{1|M_\text{total}^{-}|} \\
w_{21}^{M^{-}} & w_{22}^{M^{-}} & \cdots & w^{M^{-}}_{2|M_\text{total}^{-}|}\\
\vdots & \vdots & \ddots & \vdots \\
w_{|P|1}^{M^{-}} & w_{|P|2}^{M^{-}} & \cdots & w^{M^{-}}_{|P||M_\text{total}^{-}|}
\end{bmatrix}
\label{eq:MetricWeightsNegative}
\end{equation}

In equation~\ref{eq:MetricWeightsNegative}, $w^{M^-}_{ij}$ denotes the weight of the $j^{th}$ metric for the purpose $i^{th}$, but in this case the focus is on undesirable properties. 

\begin{equation}
g: W^{M^{-}} \rightarrow M_{\text{total}}^{-}
\label{eq:MappingFunctionNegative}
\end{equation}

Equation~\ref{eq:MappingFunctionNegative} introduces the mapping function $g$, which links the weights in $W^{M^{-}}$ to their corresponding metrics in $M_\text{total}^{-}$. This setup is pivotal in ensuring that the weights mirror the importance of each metric in evaluating negative aspects of the Quality Indicators, contributing to the thorough and balanced assessment of synthetic data quality. 

\begin{align}
W^{M^{+}} &= w^{M^{+}}_{ij} \in [0, 1] \nonumber \\
&\quad \forall i \in \{1, 2, \ldots, |P|\},  \forall j \in \{1, 2, \ldots, |M_{\text{total}}^{+}|\}
\label{eq:QIweights+}
\end{align}

\begin{align}
W^{M^{-}} &= w^{M^{-}}_{ij} \in [0, 1] \nonumber \\
&\quad \forall i \in \{1, 2, \ldots, |P|\}, \forall j \in \{1, 2, \ldots, |M_{\text{total}}^{-}|\}
\label{eq:QIweights-}
\end{align}

In both $W^{M^{+}}$ and $W^{M^{-}}$, the values of $w^{M^+}_{ij}$ and $w^{M^-}_{ij}$ are between 0 and 1 indicate varying degrees of importance or undesirability. A weight closer to 1 implies higher significance or undesirability, depending on whether it is in $W^{M+}$ (positive, desirable aspects) or $W^{M}-$ (negative, undesirable aspects). In contrast, a weight closer to 0 indicates lesser importance or undesirability. 

The weights in each row of $W^{M^{+}}$ and $W^{M^{-}}$ should be normalized so that they sum to 1 ensuring a balanced assessment of metrics within each purpose $p_i$. 

For the weights of positive metrics:

\begin{equation}
\sum_{j=1}^{|M_{\text{total}}^{+}|} w^{M^+}_{ij} = 1 \quad \forall i \in \{1, 2, \ldots,|P|\}
\label{eq:NormalizationConstraintPositiveMetrics}
\end{equation}

For the weights of negative metrics:

\begin{equation}
\sum_{j=1}^{|M_{\text{total}}^{-}|} w^{M^-}_{ij} = 1 \quad \forall i \in \{1, 2, \ldots,|P|\}
\label{eq:NormalizationConstraintNegativeMetrics}
\end{equation}

 Equations~\ref{eq:NormalizationConstraintPositiveMetrics} and~\ref{eq:NormalizationConstraintNegativeMetrics} represent normalization constraints for both $W^{M^{+}}$and $W^{M^{-}}$. This normalization constraint thus ensures that the collective weight of all the metrics for each purpose is normalized, providing a balanced approach to evaluating each quality indicator. 

\subsection{Proposed Workflow}
We divide our proposed scheme into four phases: a) specification setting, b) metric evaluation, c) ranking generation, and d) auditing phase. 

\subsubsection{Specification Setting}
In the Specification Setting phase, the Product Manager collaborates with multiple stakeholders to define the specifications for synthetic data generation. The specifications include the following criteria:

\begin{itemize}
\item \textbf{Purposes for Synthetic Data Generation:} A comprehensive list of purposes for which the synthetic datasets are to be generated (refer to Equation~\ref{eq:purpose}).
    
\item \textbf{Quality Indicators per Purpose:} A detailed list of Quality Indicators for each identified purpose (refer to Equation~\ref{eq:QIs}).
    
\item \textbf{Importance of Quality Indicators:} The relative importance of these Quality Indicators for each purpose is outlined, incorporating corresponding constraints to maintain balance (refer to Equations~\ref{eq:QIWeightsPositive} and~\ref{eq:QIweights}, as well as Equation~\ref{eq:NormalizationConstraintPositiveQIs}).
    
\item \textbf{Metrics for Desired Properties:} A catalog of metrics measuring desired properties for each Quality Indicator across purposes, ensuring compliance with relevant constraints (refer to Equation~\ref{eq:M+} and Equation~\ref{eq:NonOverlappingMetrics}).

\item \textbf{Metrics for Undesired Properties:} A catalog of metrics measuring undesired properties for each Quality Indicator across purposes, ensuring compliance with relevant constraints (refer to Equation~\ref{eq:M-} and Equation~\ref{eq:NonOverlappingMetrics}).
    
\item \textbf{Importance of Metrics for Desired Properties:} The relative significance of metrics assessing desired properties is defined, aligning with their corresponding weights and ensuring adherence to normalization constraints (refer to Equations~\ref{eq:MetricWeightsPositive}, \ref{eq:MappingFunction}, and \ref{eq:QIweights+}, along with Equation~\ref{eq:NormalizationConstraintPositiveMetrics}).
    
\item \textbf{Importance of Metrics for Undesired Properties:} The relative significance of metrics assessing undesired properties is also established, with appropriate weighting and adherence to normalization constraints (refer to Equations~\ref{eq:MetricWeightsNegative}, \ref{eq:MappingFunctionNegative}, and \ref{eq:QIweights-}, as well as Equation~\ref{eq:NormalizationConstraintNegativeMetrics}).
\item \textbf{Categorization of Metrics Based on Their Numerical Characteristics:} All evaluation metrics across various purposes are categorized into three distinct types based on their numerical characteristics: a) metrics where a lower value indicates better performance, denoted by $M_{LB}$), b) metrics where a higher value indicates better performance, denoted by $M_{HB}$, and c) metrics where performance is optimal when values are closer to a specified constant, dentoed by $M_{CC}$. This categorization allows for a more nuanced and accurate ranking of generators. Without these categorization, such nuances could be missed, leading to inaccurate comparisons. 

\end{itemize}

These definitions collectively ensure compliance with principles such as data minimization, purpose limitation, and data protection by design and default. In addition, they establish criteria for ranking synthetic data generators. Consequently, the Product Manager registers all these definitions on the consortium blockchain to ensure transparency, auditability, and accountability.

At this stage, it is important that product manager realizes the rule for assigning metrics into desirable and undesirable properties. The rule is primarily centered on the intended purpose. Metrics that align with the primary goals and requirements of the purpose should be categorized as desirable properties. Metrics that represent aspects of synthetic data that should be minimized or controlled should be categorized as undesirable properties. 

It is crucial to understand that the same metric can be desirable in one purpose and undesirable in another. This variability underscores the importance of a thorough evaluation of purpose when categorizing metrics. Hence, it is the responsibility of the product manager to consider how the high or low values of each metric impact the general purpose for which the synthetic data are generated before defining the equations~\ref{eq:M+}, and~\ref{eq:M-}.

\subsubsection{Metric Evaluations}
In the Metric Evaluations phase, the Data Scientist, having access to the real data, experiments with various synthetic data generators. The effectiveness of these generators is evaluated against predefined criteria retrieved from the consortium blockchain. For a given purpose $p_i$, the Data Scientist focuses on two key sets of metrics:
Metrics representing desirable properties, as defined in Equation~\ref{eq:M+}.
Metrics representing undesirable properties, as defined in Equation~\ref{eq:M-}. 

Across all purposes, the data scientist focuses on metrics comprising desirable and undesirable properties, as defined in Equation~\ref{eq:Mtotal}. The evaluation process involves assessing how well each synthetic data generator performs in assessing these metrics. The number of synthetic data generators considered is given by Equation~\ref{eq:generators}. 

\begin{equation}
    T = \{T_1, T_2, \ldots, T_{|T|}\}
    \label{eq:generators}
\end{equation}

The evaluation results are quantitatively captured in a matrix $E$ with rows representing a synthetic data generator $T$ and columns representing the list of metrics as shown in Equation~\ref{eq:Mtotal}.  
The evaluation matrix \( E \) is mathematically defined as:

\begin{equation}
E = \begin{bmatrix}
e_{11} & e_{12} & \cdots & e_{1|M_\text{total}|} \\
e_{21} & e_{22} & \cdots & e_{2|M_\text{total}|} \\
\vdots & \vdots & \ddots & \vdots \\
e_{|T|1} & e_{|T|2} & \cdots & e_{|T||M_\text{total}|}
\end{bmatrix}
\label{eq:EvaluationMatrix}
\end{equation}

In this matrix:
\begin{itemize}
    \item Each element \( e_{ij} \) represents the performance score of the \( i^{th} \) synthetic data generator with respect to the \( j^{th} \) metric. This score is derived from the addition of multiple evaluations in the datasets generated by each synthetic data generator.
    \item The size of the $E$ matrix is $|T| \times |M_\text{total}|$. 
\end{itemize}

The matrix $E$ containing values along with $T$ is recorded in the consortium blockchain by the data scientist. 

\subsubsection{Ranking Generation}

The primary objective in the Ranking Generation phase is to systematically rank synthetic data generators based on their effectiveness in adhering to predefined criteria for a given purpose $p_i$. Specifically, the objective is to quantify the adherence to desirable properties and the avoidance of undesirable properties. 

To facilitate this, transformation matrices $E^+_{p_i}$ and $E^-_{p_i}$ are generated from the matrix $E$ for each purpose $p_i$. The matrix $E^+_{p_i}$ captures the performance of each experimented generator against the metrics that measure desirable properties. Each row corresponds to a specific synthetic data generator $x$ in $T$, and each column represents a desirable metric $k$ in $M^{+}_{\text{total}}$. The structure of the $E^+_{p_i}$ matrix is shown in Equation~\ref{eq:EvaluationResultsMatrixPositive}.

\begin{equation}
E^+_{p_i}= \begin{bmatrix}
e^{+}_{11} & e^{+}_{12} & \cdots & e^{+}_{1|M^{+}{\text{total}}|} \\
e^{+}_{21} & e^{+}_{22} & \cdots & e^{+}_{2|M^{+}{\text{total}}|} \\
\vdots & \vdots & \ddots & \vdots \\
e^{+}_{|T|1} & e^{+}_{|T|2} & \cdots & e^{+}_{|T||M^{+}{\text{total}}|}
\end{bmatrix}
\label{eq:EvaluationResultsMatrixPositive}
\end{equation}

Similarly, the matrix $E^-_{p_i}$ records the performance of each experimented generator against the metrics that measure undesirable properties. Each row corresponds to a specific synthetic data generator $x$ in $T$, and each column represents a desirable metric $l$ in $M^{-}_{\text{total}}$. The structure of the $E^-_{p_i}$ matrix is shown in Equation~\ref{eq:EvaluationResultsMatrixNegative}.

\begin{equation}
E^+_{p_i}= \begin{bmatrix}
e^{-}_{11} & e^{-}_{12} & \cdots & e^{-}_{1|M^{-}{\text{total}}|} \\
e^{-}_{21} & e^{-}_{22} & \cdots & e^{-}_{2|M^{-}{\text{total}}|} \\
\vdots & \vdots & \ddots & \vdots \\
e^{-}_{|T|1} & e^{-}_{|T|2} & \cdots & e^{-}_{|T||M^{-}{\text{total}}|}
\end{bmatrix}
\label{eq:EvaluationResultsMatrixNegative}
\end{equation}

The objective function for each generator $T_x$ for a purpose $p_i$ is mathematically formulated as follows:

\begin{equation}
\begin{aligned}
\text{Score}_{T_x} = &\left\{ \sum_{j=1}^{|QI_{\text{total}}|} W_{ij} \times \left( \sum_{k=1}^{|M_{\text{total}}^{+}|} e^{+}_{ixk} \times W^{M^+}_{ik} \right) \right. \\
& \left. - \sum_{j=1}^{|QI_{\text{total}}|} W_{ij} \times \left( \sum_{l=1}^{|M_{\text{total}}^{-}|} e^{-}_{ixl} \times W^{M^-}_{il} \right) \right\} \\
& \quad \forall x \in \{1, 2, \ldots, |T|\} \\
& \quad \forall j \in \{1, 2, \ldots, |QI_{\text{total}}|\} \\
& \quad \forall k \in \{1, 2, \ldots, |M^{+}{\text{total}}|\} \\
& \quad \forall l \in \{1, 2, \ldots, |M^{-}{\text{total}}|\} \\
& \quad \forall i \in \{1, 2, \ldots, |P|\}
\end{aligned}
\label{eq:ObjectiveFunction}
\end{equation}

Equation~\ref{eq:ObjectiveFunction} calculates a composite score for each synthetic data generator $T_x$ for the specified purpose $p_i$. The computation considers the hierarchical structure of the weights: the significance of each Quality Indicator (QI) and the relative importance of each metric within those QIs. The approach ensures that the overall score not only reflects the performance of each generator on individual metrics, but also emphasizes the significance of those metrics in the broader context of each QI. 

By computing $Score_{Tx}$, for each generator, we can rank them according to their overall effectiveness in meeting the specifications set by the product manager. This ranking approach facilitates the identification of the most suitable synthetic data generators that best align with the specific requirements for each purpose $p_i$.

\begin{algorithm}
\caption{Rank Generators Across All Purposes}
\label{alg:RankAcrossPurposes}
\begin{algorithmic}[1]
    \Statex \textbf{Inputs:} 
    \Statex $P$: A set of all purposes for which synthetic datasets are generated
    \Statex $T$: A set of generators
    \Statex $E$: A raw metric performance score matrix 
    \Statex $M_{\text{LB}}$: A set of metrics, categorized into "lower is better".
    \Statex $M_{\text{HB}}$: A set of metrics, categorized into "higher is better". 
    \Statex $M_{\text{CC}}$: A set of ordered pair containing metrics, categorized into "closer to a constant is better" and the associated constant.
    \Statex $W^{M^{+}}$: A desired property weight matrix 
    \Statex $W^{M^{-}}$: An undesired property weight matrix
    \Statex $W$: A Quality Indicator weight matrix
 
    \Statex \textbf{Output:}
    \Statex A rank list of generators, sorted in ascending order for each purpose, where the lowest numerical rank signifies the most suitable generator for that respective purpose 
  \Procedure{RankGeneratorsAllPurposes}{$P$, $T$, $E$, $W$, $M_{\text{LB}}$, $M_{\text{HB}}$,  $M_{\text{CC}}$, $W^{M^{+}}$, $W^{M^{-}}$}
    \State Initialize an empty dictionary $AllRankings$
    \For{$p_i$ in $P$}
        \State Extract $W^{M^{+}}_{p_i} \text{ from } W^{M^{+}}$
        \State Extract $W^{M^{-}}_{p_i} \text{ from } W^{M^{-}}$
        \State $M^+_{p_i} \gets \text{ column (metric) names of }W^{M^+}_{p_i}$
        \State $M^-_{p_i} \gets \text{ column (metric) names of }W^{M^-}_{p_i}$  
       \State $(E^{+}_{p_i}, E^{-}_{p_i}) \gets \text{Transformation}($
        \Statex \hspace{\algorithmicindent}\hspace{\algorithmicindent}\hspace{\algorithmicindent}\hspace{\algorithmicindent} $E, T, p_i, M_{\text{LB}}, M_{\text{HB}}, M_{\text{CC}}, M^+_{p_i}, M^-_{p_i})$     
         \State Extract $W_{p_i}$ from $W$ 
        \State $Sorted \gets \text{RankGenerators}(p_i, E^{+}_{p_i}, E^{-}_{p_i},$
        \Statex \hspace{\algorithmicindent}\hspace{\algorithmicindent}\hspace{\algorithmicindent}\hspace{\algorithmicindent} $W_{p_i}, W^{M^{+}}_{p_i}, W^{M^{-}}_{p_i}, T)$
        \State $AllRankings[p_i] \gets Sorted$
    \EndFor
    \State Write $AllRankings$ to the ledger hosted by a consortium blockchain
    \State \textbf{return} $AllRankings$
\EndProcedure
\end{algorithmic}
\end{algorithm}

\begin{algorithm}
\caption{Ranking per Metric Transformation}
\label{alg:transformation}
\begin{algorithmic}[1]
\Statex \textbf{Inputs:}
 \Statex $p_i$: The specific purpose or scenario under evaluation, for which the ranking of generators is generated
\Statex $T$: The set of all synthetic data generators that are evaluated. These generators are represented as rows in the matrix $E$
\Statex $E$: The raw performance score matrix where each entry $E_{t,m}$ denotes the performance score of the generator $t \in T$ on metric $m$. The matrix encapsulates the performance of all generators in various evaluation metrics \Statex $M_{\text{LB}}$: A set of metrics, classified as "lower is better"
\Statex $M_{\text{HB}}$: A set of metrics, classified as "higher is better"
\Statex $M_{\text{CC}}$: A set of ordered pairs $(m, C_m)$ containing metrics categorized as "closer to a constant is better" and the associated constant
    \Statex $M^+_{p_i}$: A set of metrics that represent the desired properties in the context of $p_i$
    \Statex $M^-_{p_i}$: A set of metrics representing the undesired properties in the context of $p_i$
\Statex \textbf{Outputs:} 
\Statex Transformation matrices containing the new transformed evaluation values $E^{+}_{p_i}$ and $E^{-}_{p_i}$ for a given purpose $p_i$ 
\Procedure{Transformation}{$p_i$, $T$, $E$, $M_{\text{LB}}$, $M_{\text{HB}}$,  $M_{\text{CC}}$, $M^+_{p_i}$, $M^-_{p_i}$}
\State Initialize matrices $E^{+}_{p_i}$ and $E^{-}_{p_i}$ with dimensions $|T| \times |M^+_{p_i}|$ and $|T| \times |M^-_{p_i}|$ respectively.
\State $M^{p_i}_\text{total} \gets M^+_{p_i} \cup M^-_{p_i}$
\For{each metric \( m_j \in M^{p_i}_\text{total} \)}
    \State Create a list \( L_{m_j} \) containing pairs of \( (t_k, E[t_k][m_j]) \) for each \( t_k \in T \)fo
    \If{\( m_j \in M_{\text{LB}} \)}
        \State Sort \( L_{m_j} \) in ascending order based on $E[t_k][m_j]$.
    \ElsIf{\( m_j \in M_{\text{HB}} \)}
        \State Sort \( L_{m_j} \) in descending order based on $E[t_k][m_j]$.
    \ElsIf{\( m_j \in M_{\text{CC}} \)}
        \State Sort \( L_{m_j} \) based on the closeness of $E[t_k][m_j]$ to \( C_{m_j} \)
    \EndIf
    \State Assign ranks to each \( t_k \) in \( L_{m_j} \) based on its position in the sorted list
    \State Inverse ranks to each \(t_k \) in \( L_{m_j} \). 
    \State Store these inverted ranks in a dictionary \( \text{RankDict}_{m_j} \) with generator identifiers as keys
    
\EndFor
\For{each $m_j$ in $M^{p_i}_\text{total}$ and a $t_k \in T$}
    \If{$m_j \in M^{p_i^+}$}
        \State $E^{+}_{p_i}[t_k][m_j] = \text{RankDict}_{m_j}[t_k]$
    \EndIf
    \If{$m_j \in M^{p_i^-}$}
        \State $E^{-}_{p_i}[t_k][m_j] = \text{RankDict}_{m_j}[t_k]$
    \EndIf
\EndFor
\State \textbf{return} $E^{+}_{p_i}$ and $E^{-}_{p_i}$
\EndProcedure
\end{algorithmic}
\end{algorithm}

\begin{algorithm}
\caption{Rank Generators for a Purpose \( p_i \)}
\label{alg:ranking}
\begin{algorithmic}[1] 
    \Statex \textbf{Inputs:} 
    \Statex $p_i$: The specific purpose or scenario under evaluation for which the generator ranking is generated \Statex $T$: The set of all synthetic data generators that are evaluated. These generators are represented as rows in the matrix $E$
    \Statex $E^+_{p_i}$: A transformed matrix containing rank scores for all metrics categorized as desired properties
    \Statex $E^-_{p_i}$: A transformed matrix containing rank scores for all metrics categorized as undesired properties
    \Statex $W_{p_i}$: The weights of the quality indicator in the context of a purpose $p_i$ 
    \Statex $W^{M^{+}}_{p_i}$: The desired property metric weights in the context of a purpose $p_i$
    \Statex $W^{M^{-}}_{p_i}$: The undesired property metric weights in the context of a purpose $p_i$ 
    \Statex \textbf{Outputs:} 
    \Statex A sorted list of generators for the specific purpose $p_i$ is arranged in ascending order, where the lowest numerical rank represents the generator most suitable for the specific purpose $p_i$
   \Procedure{RankGenerators}{$p_i$, $T$, $E^+_{p_i}$, $E^-_{p_i}$, $W_{p_i}$, $W^{M^{+}}_{p_i}$, $W^{M^{-}}_{p_i}$}
    \State Initialize an empty list $Scores$
    \State $|M^+_{p_i}| \gets \text{number of metrics in } W^{M^{+}}_{p_i}$
    \State $|M^-_{p_i}| \gets \text{number of metrics in } W^{M^{-}}_{p_i}$

    \For{$t_k \in T$}
        \State $\text{OverallScore}_i \gets 0$
        \State $\text{DesiredScore}_i \gets 0$
        \State $ \text{UndesiredScore}_i \gets 0$
        
        \For{$m_j \in M^+_{p_i}$}
            \State $\text{DesiredMetricWeight} \gets W^{M^{+}}_{p_i}[m_j]$
            \State $\text{QICategoryWeight} \gets W_{p_i}[\text{category } m_j]$
            \State $\text{DesiredTransformedValue} \gets E^+_{p_i}[t_k][m_j]$
            \State $\text{DesiredScore}_i \gets \text{DesiredScore}_i + (\text{DesiredTransformedValue} \times \text{DesiredMetricWeight} \times \text{QICategoryWeight})$
        \EndFor

        \For{$m_j \in M^-_{p_i}$}
            \State $\text{UndesiredMetricWeight} \gets W^{M^{-}}_{p_i}[m_j]$
            \State $\text{QICategoryWeight} \gets W_{p_i}[\text{category } m_j]$
            \State $\text{UndesiredTransformedValue} \gets E^-_{p_i}[t_k][m_j]$
            \State $\text{UndesiredScore}_i \gets \text{UndesiredScore}_i - (\text{UndesiredTransformedValue} \times \text{UndesiredMetricWeight} \times \text{QICategoryWeight})$
        \EndFor
        \State $\text{OverallScore}_i \gets \text{DesiredScore}_i - \text{UnDesiredScore}_i$
        \State Append $\text{OverallScore}_i$ and $t_k$ as a pair to $Scores$
    \EndFor

    \State Sort $t_k$ within $Scores$ based on $OverallScore$ in descending order.
    \State Write the $SortedScores$ for \( p_i \) to the ledger hosted by the consortium blockchain
    \State \textbf{return} $SortedScores$ for \( p_i \) recorded in the blockchain
    \EndProcedure
\end{algorithmic}
\end{algorithm}

Algorithm~\ref{alg:RankAcrossPurposes} presents a thorough and flexible approach to evaluate and rank synthetic data generators for a variety of defined purposes. This framework systematically considers the performance of each generator in selected metrics, as well as the relative importance of these metrics, as determined by the specific requirements associated with each purpose in the set $P$. This ensures that the ranking process is not only performance-based, but also aligns with the nuanced demands and priorities of each distinct purpose.

The summary of Algorithm~\ref{alg:RankAcrossPurposes} is as follows: The algorithm commences by establishing the necessary inputs as outlined in the Inputs part of the Algorithm~\ref{alg:RankAcrossPurposes}. Then it initializes an empty dictionary $AllRankings$ to store the calculated scores of the generators for each purpose.
For each specified purpose $p_i$ within $P$, the algorithm performs the following steps:
    \begin{itemize}
        \item It constructs the matrices $E^{+}_{p_i}$ and $E^{-}_{p_i}$ from $E$ by employing the Algorithm~\ref{alg:transformation}.   These matrices capture the performance of each generator in terms of metrics identified as desirable or undesirable. The transformation applied is tailored to reflect the inherent nature of each metric, highlighting the algorithm's adaptability to accommodating diverse evaluation needs. The purpose of the transformation function is to prepare the inputs ($E$) provided by the data scientist so that they can be applied effectively to the ranking equation~\ref{eq:ObjectiveFunction}. This involves ensuring that performance values are uniformly interpreted in $E$ where higher values in any metric represent better performance. The function categorizes these values into desired ($Eˆ{+}_{p_i}$) and ($E^{-}_{p_i}$) performance categories. 
        \item Then extracts the relevant Quality Indicator weights ($W_{p_i}$) and metric weights ($W^{M^{+}}_{p_i}$ and $W^{M^{-}}_{p_i}$) or the current purpose.
        \item Subsequently, the Algorithm~\ref{alg:ranking} is invoked with these inputs, calculating and organizing the generators based on their overall performance scores, taking into account both desirable and undesirable characteristics.
    \end{itemize}

The algorithm compiles the sorted rankings for each purpose within the $AllRankings$ dictionary. After all purposes have been processed, $AllRankings$ is committed to a consortium blockchain and returned. This process ensures a holistic and transparent record of generator rankings for each distinct purpose, reflecting a balance between performance metrics and specific purpose requirements.

The algorithm~\ref{alg:transformation} provides a procedure for transforming the raw performance scores of synthetic data generators into matrices that reflect their performance in relation to the desired and undesired properties for a specific evaluation purpose. The summary of Algorithm~\ref{alg:transformation} is as follows: $E^{+}_{p_i}$ and $E^{-}_{p_i}$, are initialized to store transformed scores. Their dimensions are determined by the number of generators and the number of metrics categorized as desired or undesired for the specific purpose. For each metric $m_j$, the algorithm creates a list of pairs consisting of generator identifiers and their corresponding scores for $m_j$. These lists are then sorted according to the nature of the metric:
\begin{itemize}
\item For metrics where lower scores are better ($M_{\text{LB}}$), the list is sorted in ascending order.
\item For metrics where higher scores are better ($M_{\text{HB}}$), it is sorted in descending order.
\item For optimal metrics at a specific constant ($M_{\text{CC}}$), it is sorted based on the proximity to the constant value associated with each metric.
\end{itemize}

After sorting, each generator is assigned a rank based on its position in the sorted list for each metric. A lower value in each metric denotes a better performer. However, to align with the ranking Equation~\ref{eq:ObjectiveFunction}, where higher values represent superior performance, the ranks are inverted. These inverted ranks are stored in a dictionary for convenient access. Subsequently, the algorithm populates $E^{+}_{p_i}$ and $E^{-}_{p_i}$ using the ranks of the dictionary. For each generator and metric, if the metric is desirable for the purpose, its rank is stored in $E^{+}_{p_i}$; if it is undesirable, the rank is stored in $E^{-}_{p_i}$. Finally, the algorithm returns the matrices $E^{+}_{p_i}$ and $E^{-}_{p_i}$, which now reflect the transformed performance scores of the generators regarding the desirable and undesirable properties for the specific purpose $p_i$. These matrices are then utilized in the Algorithm~\ref{alg:RankAcrossPurposes}.

The algorithm~\ref{alg:ranking} ranks synthetic data generators tailored for a specific purpose $p_i$. Initially, it initializes an empty dictionary named $Scores$ to store the calculated scores for each generator. The algorithm iterates through each generator $t_k$ in $T$, computing three distinct scores: the desired score and the undesired score, based on the performance on both desirable and undesirable metrics, in addition to the overall score derived by subtracting the undesired score from the desired score.

For each metric $m_j$ in $M^+_{p_i}$ and $M^-_{p_i}$, the desired and undesired scores are determined by multiplying the transformed value of the metric from $E^{+}_{p_i}$ or $E^{-}_{p_i}$ by their corresponding weights $W_{p_i}, W^{M^{+}}_{p_i}$, and $W^{M^{-}}_{p_i}$. The desired scores are aggregated across all metrics in $M^+_{p_i}$, and similarly, the undesired scores are aggregated across all metrics in $M^-_{p_i}$.

Subsequently, the overall score for each generator $t_k$ is computed by subtracting its undesired scores from its desired scores. If the overall score of a generator $t_k$ is negative, it indicates that the generator exhibits more undesirable properties than desired ones; otherwise, it suggests that the generator exhibits more desirable properties than undesirable ones.

These scores, along with their respective $t_k$, are stored in the $Scores$ dictionary. The dictionary then sort$t_k$ based on their overall score in descending order, where the generator$t_k$ with the highest overall score represents the best performance among the set of generators. 

Following this, the algorithm extracts a sorted list of generators from $Scores$ and records it in a ledger hosted on a consortium blockchain. This ensures transparency, accountability, and auditability of the ranking process.

Finally, the sorted list of generators for $p_i$ is returned as output, which is further utilized in Algorithm~\ref{alg:RankAcrossPurposes}.

\subsubsection{Auditing phase}

\begin{algorithm}
\caption{Auditing Process for Verification}
\label{alg:DetailedAuditing}
\begin{algorithmic}[1]
 \Statex \textbf{Inputs:} 
 \Statex $P$: A set of all purposes for which synthetic datasets are generated
    \Statex $T$: A set of generators
    \Statex $E$: A raw metric performance score matrix 
    \Statex $M_{\text{LB}}$: A set of metrics, categorized into "lower is better".
    \Statex $M_{\text{HB}}$: A set of metrics, categorized into "higher is better". 
    \Statex $M_{\text{CC}}$: A set of ordered pair containing metrics, categorized into "closer to a constant is better" and the associated constant.
    \Statex $W^{M^{+}}$: A desired property weight matrix 
    \Statex $W^{M^{-}}$: An undesired property weight matrix
    \Statex $W$: A quality indicator weight matrix 
    \Statex \textbf{Output:}
   \Statex A boolean audit result. If true, it indicates successful verification; if false, it denotes inconsistencies
\Procedure{AuditRankings}{$P$, $T$, $E$, $M_{\text{LB}}$, $M_{\text{HB}}$, $M_{\text{CC}}$, $W^{M^{+}}$, $W^{M^{-}}$, $W$}
    \State Initialize a variable \( \text{isConsistent} \) as \( \text{True} \)
    \For{each purpose \( p_i \) in \( P \)}
        \Statex \textbf{Verify Product Manager's Specification:}
        \State - Compare \( W_{p_i} \), \( W^{M^{+}}_{p_i} \), \( W^{M^{-}}_{p_i} \) with corresponding records in Blockchain.
        \State - If any inconsistency is found in weight matrices, flag as inconsistent
        \Statex \textbf{Verify Data Scientist's Evaluation Matrices:}
        \State - Cross-check \( E \) and \( T \) against their records in the Blockchain
        \State - If evaluation matrices do not match, flag as inconsistent
        \Statex \textbf{Replicate Ranking Process:}
        \State - Retrieve $SortedGenerators$ for $p_i$ from blockchain
        \State - Use \( E \), \( W_{p_i} \), \( W^{M^{+}}_{p_i} \), \( W^{M^{-}}_{p_i} \), $M_{\text{LB}}$, $M_{\text{HB}}$, $M_{\text{CC}}$ and \( T \) for recalculation of $SortedGenerators$ 
        \Statex \textbf{Compare Rankings:}
        \State - Match recalculated rankings with the rankings recorded in the Blockchain
        \State - If there is a discrepancy in rankings, flag as inconsistent
        \State - Record any discrepancies or inconsistencies identified in the blockchain
    \EndFor
    \Statex \textbf{Conclude Audit:}
    \State - If any inconsistencies are found, \( \text{isConsistent} \) is set to \( \text{False} \)
    \State - Record status of isConsistent to the ledger hosted by the consortium blockchain
    \State \textbf{return} \( \text{isConsistent} \)
\EndProcedure
\end{algorithmic}
\end{algorithm}

The objective of the audit phase is to perform an independent review by an auditor to verify that the ranking of synthetic data generators was completed accurately and in accordance with the established criteria. 

The algorithm~\ref{alg:DetailedAuditing} summarizes the audit procedure. The auditor reviews the predefined criteria and metrics the Product Manager sets, ensuring that they align with legal, ethical, and organizational standards. This involves analyzing the all the inputs provided by the product manager towards setting the criteria for ranking synthetic data generators for each purpose. The auditor cross-references these criteria with the records in the consortium blockchain to verify their consistency and appropriateness. 

Moreover, the external auditor conducts a meticulous review of the inputs provided by the Data Scientists. This includes a thorough examination of the evaluation matrices $E$ and $T$, which are crucial in determining the performance of each synthetic data generator against the defined metrics. In addition to this, the auditor accesses the compiled ranking list $AllRankings$, utilized by the Data Scientist to select the optimal synthetic data for the Product Manager's requirements. This comprehensive analysis by the auditor ensures the accuracy and integrity of the evaluation process, confirming that the rankings and selections reflect the true performance of the generator.  

The auditor meticulously reviews the ranking algorithm and its implementation to verify its adherence to correct protocols, particularly focusing on how it integrates the evaluation results with the weights assigned to Quality Indicators (QIs) and metrics. To validate the accuracy of the final rankings, the auditor replicates the ranking process. This is achieved by invoking the ranking algorithms embedded within the smart contracts to compare with $AllRankings$. 

In addition to the detailed review and verification process, the auditor plays a crucial role in maintaining the integrity and transparency of the system by recording the audit results on the blockchain. This involves documenting the results of their examination, including any discrepancies identified or confirmations of the ranking algorithm's accuracy. By recording these results on the blockchain, the auditor ensures that the audit trail is immutable and verifiable if needed. 

The algorithms~\ref{alg:RankAcrossPurposes}~\ref{alg:transformation}~\ref{alg:ranking}~\ref{alg:DetailedAuditing} developed to evaluate and rank synthetic data generators form an integral part of a robust framework. This framework ensures a comprehensive and context-sensitive assessment of each generator. Key to the functionality of this system are smart contracts, which are automated self-executing contracts with the terms of the agreement directly written into code.

These smart contracts can be invoked upon the recording of inputs from data scientists and product managers on the blockchain. When these inputs are recorded, the smart contracts activate the evaluation and ranking algorithms. This automation ensures a seamless, efficient, and transparent process. By integrating these algorithms into smart contracts, the framework gains several advantages:
\begin{itemize}
\item \textbf{Immutability:} Once deployed, the terms within the smart contract cannot be altered, ensuring the integrity of the evaluation process.

\item \textbf{Transparency:} All stakeholders have visibility into the rules and conditions set within the smart contracts, promoting trust and accountability.

\item \textbf{Efficiency:} The automation of the evaluation and ranking process reduces the need for manual intervention, speeding up the decision-making process.

\item \textbf{Accuracy:} By codifying the evaluation criteria and ranking methodologies into smart contracts, the potential for human error is significantly reduced. Of course, we assume the smart contracts are thoroughly code reviewed before being deployed on the blockchain network. 

\item \textbf{Traceability:} Each transaction of all end users (product manager, data scientist, and auditor) is recorded on the blockchain, providing a transparent and traceable audit trail.
\end{itemize}
In essence, the integration of these algorithms with blockchain technology and smart contracts creates a dynamic and reliable ecosystem for synthetic data generation. This system not only streamlines the evaluation and ranking of data generators, but also upholds the principles of fairness, transparency, and accuracy, which are crucial in building trust among users and stakeholders in the synthetic data domain.

\section{Integrating the Proposed Ranking Algorithms with a Blockchain Framework}
\label{sec:conversion}
\subsection{Background on Sawtooth}

Sawtooth~\cite{Sawtooth} is an enterprise blockchain platform for building distributed ledger applications and networks. The primary reason for selecting Sawtooth for our integration of algorithms with blockchain technology lies in its ability to simplify application development. Sawtooth achieves this through a separation of the application domain from the core system, facilitated by a feature known as the "transaction family." This separation allows developers to write smart contracts in their preferred programming languages, improving flexibility. Specifically, we opt for Python to implement our algorithms, leveraging the extensive statistical libraries available in Python for our purposes. 

Figure~\ref{fig:sawtooth} shows an overview of a Sawtooth blockchain network. At its core, a Sawtooth network is comprised of a peer-to-peer network of Sawtooth nodes, each of which includes several key components essential for network's operation. 
\begin{figure}[htbp]
\centerline{\includegraphics[width=\columnwidth]{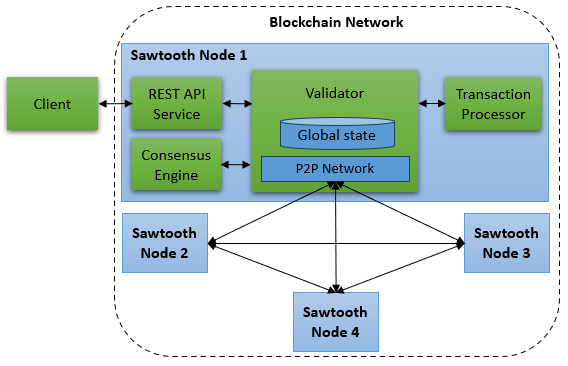}}
\caption{An Overview of a Sawtooth blockchain network}
\label{fig:sawtooth}
\end{figure}

\begin{itemize}
\item\textbf{Validator:} The core component of a Sawtooth node, known as the validator, defines the node's configuration, including network and component endpoints, the list of peers, and the minimum and maximum number of required peers. It is also responsible for validating transactions and blocks received through the REST API or the peer-to-peer network.
\item\textbf{REST API:}  A simple interface allows clients to interact with the validator. The REST API is the component that is bound to the validator for this purpose. 

\item\textbf{Transaction processor:}  Transaction processor (smart contract) includes the business logic and implemented algorithms of the application. The transaction processor communicates with the validator component to get or set the data that is required for performing the business logic algorithms.

The implemented transaction processor handles the transaction received from the validator. It decodes and deserializes the payload of the transaction to extract the command and its list of arguments, which was initiated and serialized by the client. Then, the corresponding function inside the transaction processor is run on the basis of the extracted command. Each function of the transaction processor performs an action based on a defined command. 

In our use case, the action purpose is either populating the ledger with the generators, quality indicators, desired and undesired properties, and quality indicator weights, or computing the method rankings. For this purpose, the transaction processor gets the ledger state from the validator, performs the required actions on it, and finally sets the new state on the ledger.

\item\textbf{Global state:} The application data is stored in the Global state inside the validator component of a Sawtooth node. The data structure for storing data is a Merkle-Radix tree. Sawtooth core system takes the responsibility of creating this data structure and managing the stored data. The part that is related to the application level is constructing address schemes for each data that is going to be stored in the leaf nodes of the tree. The address must be computable by any Sawtooth node (specifically by the validator component) or the client that needs to access it. Therefore, the address must be defined in the application on both the client side and the transaction processor side in a deterministic way. 

We construct the addresses using the hex-encoded hash values of a specific string to be deterministically calculable on the transaction processor side and on the client side for registering into and reading from the ledger. Each address scheme includes a namespace prefix. We define separate namespaces for each category of our data, i.e. generators, quality indicators, quality indicator weights, desirable and undesirable properties, and rankings, to be able to access the Merkle-radix tree of each of them independently.

\item\textbf{Consensus engine:} A Sawtooth node needs a component that implements the consensus algorithm. Consensus engine component cooperates with the validator to run the consensus algorithm. Sawtooth supports two types of consensus algorithms for a blockchain network deployment: PoET (Proof of Elapsed Time) and PBFT (Practical Byzantine Fault Tolerance).

PoET uses a leader-election lottery system to decide who has the right to publish a block. Based on this algorithm, each Sawtooth node must generate a random waiting time and sleep for the generated random period. The node that completes its waiting time and wakes up first can commit a block and publish it to the network. However, PBFT is a voting-based consensus algorithm. In each round of publishing a block, the network chooses a leader. The leader publishes its block, and the other nodes vote on the block to commit or reject it. 

Given our use case's operational context within an enterprise environment, we inherently assume a partially synchronous system. This assumption reflects the typical conditions of enterprise systems, where communication delays are expected to be bounded, despite being subject to occasional fluctuations. In light of this framework, our decision to select PBFT over PoET for our blockchain deployment is driven by two key considerations, particularly relevant in a partially synchronous environment. First, PBFT provides better safety and liveness guarantees than PoET. PBFT guarantees that the system remains in a correct state, provided that the number of Byzantine nodes in the network does not exceed one third of the total network. This feature ensures that transactions are inevitably processed, offering a higher degree of reliability. Second, PBFT delivers deterministic consensus outcomes, ensuring that once a consensus on blocks is reached, these blocks are irrevocably committed across all Sawtooth nodes, effectively negating the possibility of network forks. 
\end{itemize}


\subsection{Smart Contracts for Ranking Synthetic Data Generators and Auditing the Ranking Process}

  \begin{figure}[htbp]
    \centerline{\includegraphics[width=\columnwidth]{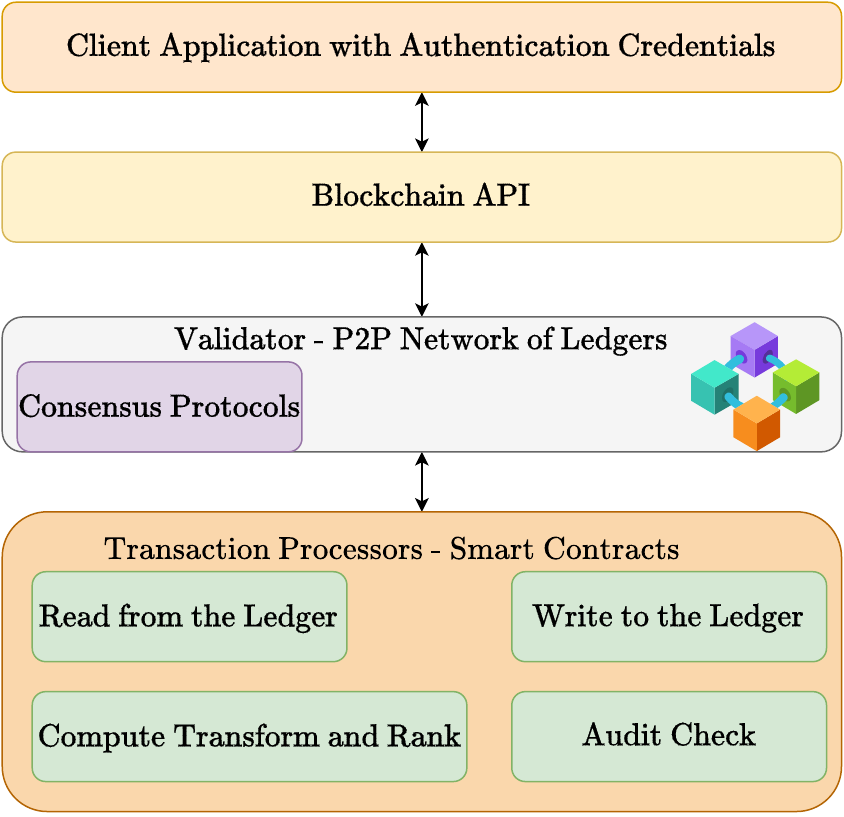}}
    \caption{Overview of the Proposed Blockchain Framework for Ranking Synthetic Data Generators.}
    \label{fig:blockchainFramework}
    \end{figure}

In our use case, the development of smart contracts is based on the ranking algorithms proposed in Section~\ref{sec:proposed}. Our smart contracts are designed with two primary objectives. The first is to automatically rank synthetic data generators for a specified purpose, based on input from product managers and data scientists. The second objective is to enable auditors to automatically review the inputs from product managers and data scientists for a specific purpose and to verify the final ranking outcomes of synthetic data generators for that purpose.

We propose four main modules within the smart contract to achieve these goals: a) Write to the Ledger, b) Read from the Ledger, c) Compute Transform and Rank, and d) Audit Check, as illustrated in Figure~\ref{fig:blockchainFramework}. Authenticated and authorized clients can invoke these modules in smart contracts by providing appropriate attributes as input. The public and private keys associated with user roles are utilized to establish the user identity, and a validator component verifies this identity before smart contracts are invoked.

 \begin{figure}[htbp]
    \centerline{\includegraphics[width=\columnwidth]{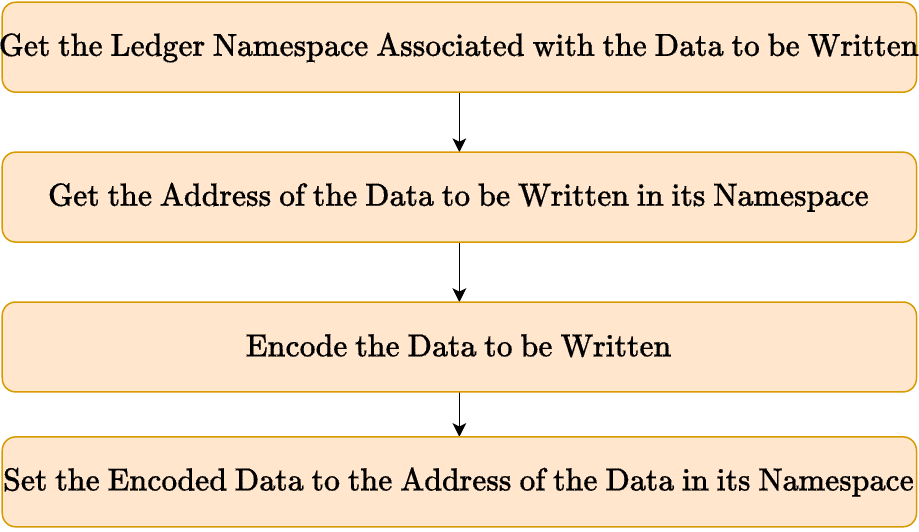}}
    \caption{Overview of Writing Data through Smart Contracts to the Sawtooth Ledger: The datasets include inputs from product managers, data scientists, rankings of synthetic data generators, and audit check results. The namespaces for each of these datasets are predefined.}
    \label{fig:WriteDatatoLedger}
    \end{figure}

\begin{table*}[htbp]
\centering
\caption{Data types to be recorded on and read from the Ledger have their respective namespaces predefined within the smart contracts.}
\label{tab:dataTypes}
\begin{tabular}{p{4cm}|p{4cm}|p{4cm}|p{4cm}}
\toprule
Product Manager & Data Scientists & Compute Transform and Rank & Audit Check Result \\
\midrule
Purposes, weight distributions for desired ($WM^{+}$) and undesired properties ($WM^{-}$), quality indicator weights ($W$) and metric classifications ($M_{LB}$, $M_{HB}$, $M_{CC}$) & A raw performance score of synthetic data generators for a specific purpose ($E$) & Rankings of synthetic data generators for a specific purpose & Results of the audit check for a specific purpose \\
\bottomrule
\end{tabular}
\end{table*}

The "Write to the Ledger" module can be invoked by Product Managers and Data Scientists to register their input on the ledger. Additionally, the "Compute Transform and Rank" and "Audit Check" modules may also invoke the "Write to the Ledger" module to record the final ranking for a given purpose and the results of the audit on the ledger.

The workflow for the "Write to Ledger" module is illustrated in Figure~\ref{fig:WriteDatatoLedger}. In our use case, all the data types listed in Table~\ref{tab:dataTypes} adhere to the same process for recording in the Sawtooth Ledger. The namespaces are predefined for each data type to be recorded in the ledger. Depending on the data type, the appropriate namespace and the address within that namespace are identified. The data intended for recording are first encoded in UTF-8 string format. This encoded data are then assigned to the data's specific address within its namespace, utilizing the internal Sawtooth APIs.

  \begin{figure}[htbp]
    \centerline{\includegraphics[width=\columnwidth]{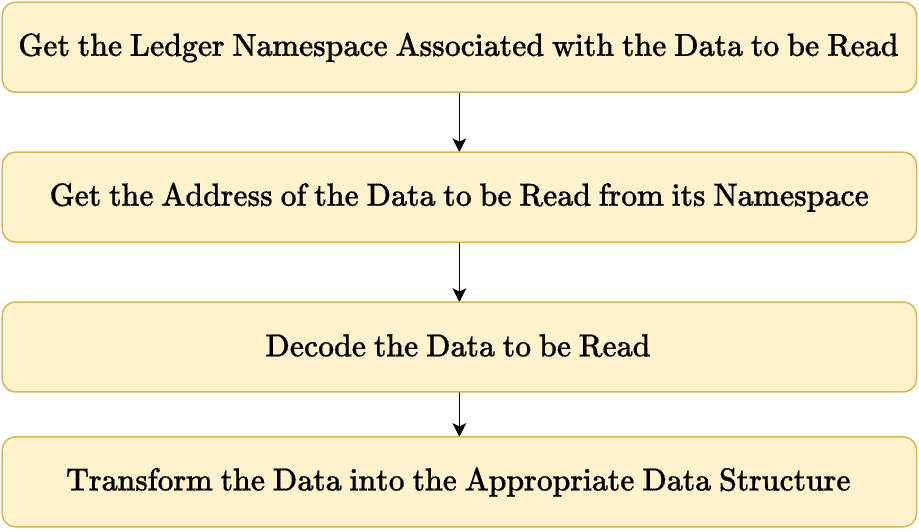}}
    \caption{Overview of Reading Data through the Smart Contract from the Sawtooth Ledger.}
    \label{fig:ReadDataFromLedger}
    \end{figure}

The "Read from the Ledger" module can be invoked by Product Managers, Data Scientists, Auditors, and the "Compute Transform and Rank" and "Audit Check" modules. The workflow for the "Read from Ledger" module is illustrated in Figure~\ref{fig:ReadDataFromLedger}. In our scenario, all data types listed in Table~\ref{tab:dataTypes} undergo the same procedure to retrieve data from the Sawtooth ledger. The predefined namespaces accommodate the data types that need to be read. Based on the type of data to be accessed, the relevant namespace is identified and the address of the data to be read is located. The data is then decoded and transformed into the suitable data structure. If the data are to be presented to the client application, it is converted into JSON format; if it is to be used within the "Compute Transform and Rank" and "Audit Check" modules of the smart contract, it is converted into the requisite data structure format.

The "Compute Transform and Rank" module can be invoked by Product Managers, Data Scientists, Auditors, and the "Audit Check" module. The primary goal of this workflow is to compute the final ranking of synthetic data generators based on the inputs recorded on the blockchain by both the Product Manager and the Data Scientist. The module comprises two logical functions: a) transformation and b) ranking. The transformation function prepares the inputs provided by the Data Scientists for the ranking process. The ranking function uses the scoring technique described in Section~\ref{sec:proposed} to calculate the final ranking for a specific purpose, taking into account the input of the Product Manager and the transformed inputs from the Data Scientists. The final result of the ranking process is then recorded on the blockchain. 

The "Audit Check" module can be invoked by Auditors. The primary objective of this module is to verify three key criteria: a) whether the inputs provided by the Product Manager correspond with the inputs recorded on the blockchain; b) whether the inputs provided by the Data Scientists match the inputs recorded on the blockchain; and c) whether the computed ranking results are consistent with the ranking results stored on the ledger. If any of these criteria are not met, the result of the audit check is recorded as false on the blockchain. Otherwise, it is recorded as true.

\subsection{Deployment of Smart Contracts to the Sawtooth Network} \label{sec:smartcontractDeployment}

The transaction processor includes the business logic of the application and contains implementation of all the algorithms introduced for our application. Since all Sawtooth nodes have to run the same smart contract in a Sawtooth network, the developer has to replicate the same transaction processor python file among all the Sawtooth nodes. In order to connect the transaction processor to the validator, the URL of the validator is set in our transaction processor Python code. After adding this configuration, our transaction processor can run by the following command on each of the Sawtooth nodes:
sudo python3 ./synthrank\_tp-v3.py

By running the transaction processor on all the nodes, it will wait to be invoked and start its functions as follows when a client has a request:

\begin{enumerate}

\item \textbf{Receiving transactions from the validator:} The validator receives the batch from the client. The batch is propagated across the P2P network of Sawtooth node validators. Validators extract transactions of the batch and send each valid transaction to its corresponding transaction processor which is registered before on the validator. For example, transactions of the same transaction family are sent to their own transaction processor.
    
    \item \textbf{Executing the transaction:} The transaction processor that receives the transaction from the validator extracts the payload of the transaction, decodes, and deserializes the payload to get the action needed to execute the transaction. The transaction processor runs the function related to the action command in the transaction. It gets and sets the state in the validator as needed to complete the transaction execution. Figure~\ref{fig:cwFunction} shows the function written to execute the weight command of the quality indicator. First, the state address in which the quality indicator weight of a purpose is going to be registered is calculated based on the previously defined data model. Then, the weight values of the quality indicator with their purpose name are encoded and registered in their specific address in the ledger. 
    \begin{figure}[htbp]
    \centerline{\includegraphics[width=\columnwidth]{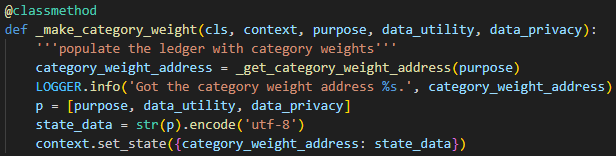}}
    \caption{Populating the ledger with the quality indicator weights}
    \label{fig:cwFunction}
    \end{figure}
    
    As another example of handling transactions by the transaction processor, Figure~\ref{fig:rankingFunction} illustrates part of the function used to rank the methods for different purposes. Finally, the computed rankings are registered in the specified state address to be retrievable deterministically based on the defined data model.  
    \begin{figure}[htbp]
    \centerline{\includegraphics[width=\columnwidth]{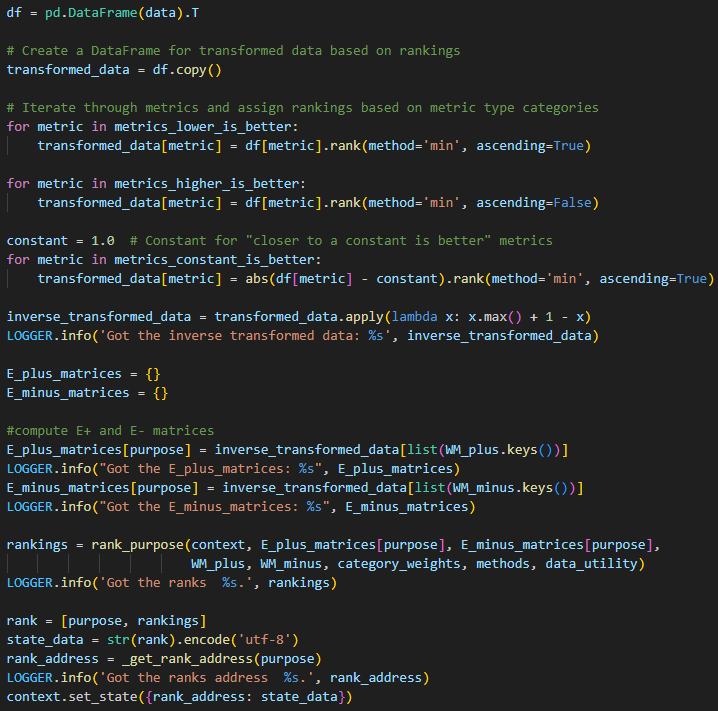}}
    \caption{Populating the ledger after computation of the rankings based on the registered values}
    \label{fig:rankingFunction}
    \end{figure}

    \item \textbf{Building agreement among validators:} Validators start creating blocks containing the received batches based on the consensus algorithm. As we are using PBFT consensus, one leader validator publishes its candidate block across the validator network. Validators receiving the candidate block start validating all batches and transactions inside the block. The Sawtooth core system handles the whole process of validating and committing the transaction, and this step is hidden from the scope of the application.

\end{enumerate}

 
\subsection{Design of Client Application Interface to Blockchain API} \label{sec:clientAPI}

Any changes in the global state must be initiated by a transaction. The client is responsible for creating a transaction and submitting it to the validator component of the Sawtooth node through the REST API service. The transaction includes the data value that will be registered in the blockchain ledger. 

We divide the client functionalities into two parts to achieve the responsibilities of the client. The first part is to provide a command line interface for the client to initiate the transactions. Our defined commands in the proposed application support the application actions such as registering methods, quality indicators, quality indicator weights, desired and undesired properties, computing and registering the ranks, showing any registered data, and auditing all the registered data in comparison with the given inputs. The second part of the client functionality is responsible for getting the registered commands, creating a transaction and a batch, and sending them to the REST API. These two parts are deployed as two Python files (synthrank.py and synthrank\_client.py). In the following, we explain how the client is prepared and interacts with the blockchain API.

\begin{enumerate}

    \item \textbf{Creating user signing keys:} The validator component of the Sawtooth node needs to confirm the identity of the transaction sender for the sake of the privacy and security of the application. Therefore, a private key file and a public key file have to be generated on the validator for each of our users, i.e. the product manager, the data scientist, and the auditor, to be able to authorize the transaction senders. Figure~\ref{fig:keygen} shows how the product manager is defined and therefore will be identified on the Sawtooth node.
    \begin{figure}[htbp]
    \centerline{\includegraphics[width=\columnwidth]{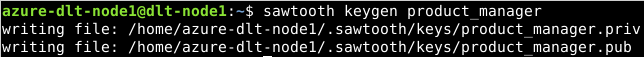}}
    \caption{Generating a private key and public key for a client}
    \label{fig:keygen}
    \end{figure}

    \item \textbf{Registering the commands:} The user runs the Python file of the client application where the command-line interface for initiating transactions is defined. In a terminal, the user inputs a command based on a defined format. Table~\ref{tab:commands} presents the supported commands for our application users:
    \begin{table*}[h!]
    \centering
    \begin{tabular}{c p{12cm}}
        \toprule
        Command & Description \\
        \midrule
        qi qi.txt --key \{user\_name\} & Registering quality indicators written in a file by the product manager \\
        cw weights.txt --key \{user\_name\} & Registering quality indicator weights (category weights) written in a file by the product manager \\
        wmp WM\_plus.txt --key \{user\_name\} & Registering WM\_plus written in a file by the product manager \\
        wmp WM\_plus.txt --key \{user\_name\} & Registering WM\_minus written in a file by the product manager \\
        method inputs.txt --key \{user\_name\} & Registering methods written in a file by the data scientist \\
        qos compute.txt & Computing QoS score of all purposes written in a file and registering the ranked result \\
        rank \{purpose\_name\} & Getting the QoS score rank of all the registered methods for a given purpose \\
        ranks & Getting the QoS score rank of all the registered purposes \\
        audit --key \{user\_name\} & Auditing the files given to the auditor and comparing with the registered data in the ledger by the auditor, showing the audit process, and registering the final audit result \\
        isConsistent & Getting the registered audit result \\
        methods & Getting the registered methods \\
        cws & Getting the registered category weights \\
        wmps & Getting the registered wm\_plus\\
        wmms & Getting the registered wm\_minus \\
        qis & Getting the registered quality indicators \\
        \bottomrule
    \end{tabular}
    \caption{Our application supported commands}
    \label{tab:commands}
    \end{table*}
    
    For example, Figure~\ref{fig:manager} shows how the product manager registers quality indicator weights based on a defined format in the table. The command "cw" indicates the registration of quality indicator weights, and the given file "weights.txt" which is shown in Figure~\ref{fig:categoryWeightFile} includes the weights of the data utility and the data privacy weights of the purposes that are to be registered in the ledger. The product manager has to enter its registered key name in the validator to verify its user. The given key "product\_manager" is the one that was created in the previous step.  
    \begin{figure}[htbp]
    \centerline{\includegraphics[width=\columnwidth]{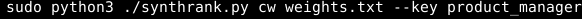}}
    \caption{The command of registering the quality indicator weights by the product manager}
    \label{fig:manager}
    \end{figure}
    \begin{figure}[htbp]
    \centerline{\includegraphics[width=\columnwidth]{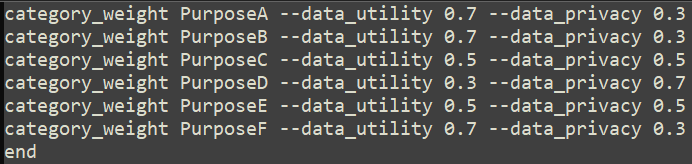}}
    \caption{Content of the weights.txt file for registering the quality indicator weights of different purposes}
    \label{fig:categoryWeightFile}
    \end{figure}

    \item \textbf{Creating and signing the transaction and batch:} After registering the command through the command line interface in the terminal, the client code handled by synthrank.py analyzes the input command to detect the action and its arguments. Then, the rest of the client functions are handled by the synthrank\_client.py file. The background code in this file encodes the action and the arguments received from synthrank.py and stores them as the payload of a transaction. Thereafter, the transaction header is created including the public key of the signer, the transaction family name and version, the permitted state addresses for reading the ledger and writing into it, the list of dependent transactions, the hash of the payload, public key of the batcher, and a nonce for the transaction. As a consequence, the transaction is created from the header and the encoded payload. Before submitting the transaction to the validator, the transaction is put inside a batch. Figure~\ref{fig:transaction} shows a snippet of the client code where it creates the transaction and the batch.
    \begin{figure}[htbp]
    \centerline{\includegraphics[width=\columnwidth]{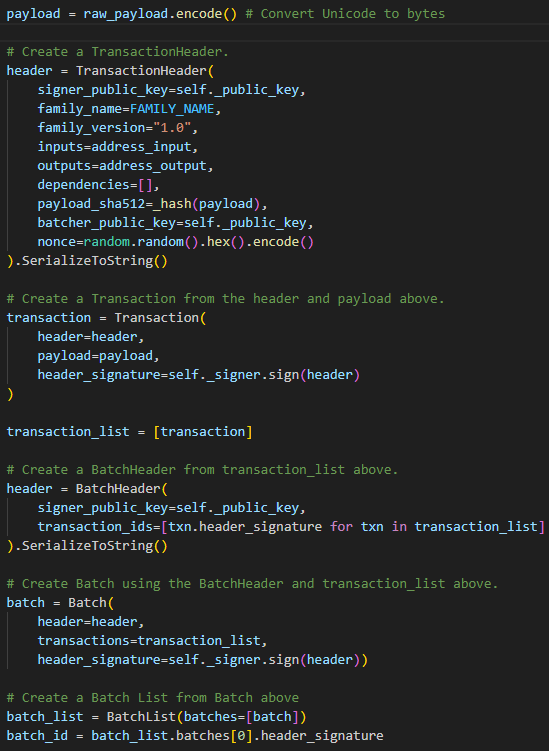}}
    \caption{Creating a transaction and a batch before submitting to the validator}
    \label{fig:transaction}
    \end{figure}
    
    \item \textbf{Submitting the batch to the validator:} The created batch from the previous step is sent to the validator via the REST API and the function shown in Figure~\ref{fig:rest}. This function is the last action done on the client side before waiting to get the response from the validator. 
    \begin{figure}[htbp]
    \centerline{\includegraphics[width=\columnwidth]{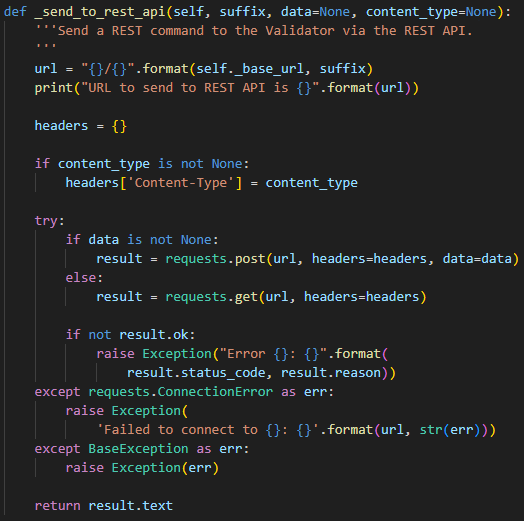}}
    \caption{Submitting batches to the validator via REST API}
    \label{fig:rest}
    \end{figure}

    \item \textbf{Getting acknowledgement response on the client:} After executing the transaction by the transaction processor and updating the global state based on the agreement between the validators, the client will be notified about its committed or rejected transaction on its terminal via a JSON response received from the REST API. The response shown for the client is like what is shown in Figure~\ref{fig:response}.
    \begin{figure}[htbp]
    \centerline{\includegraphics[width=\columnwidth]{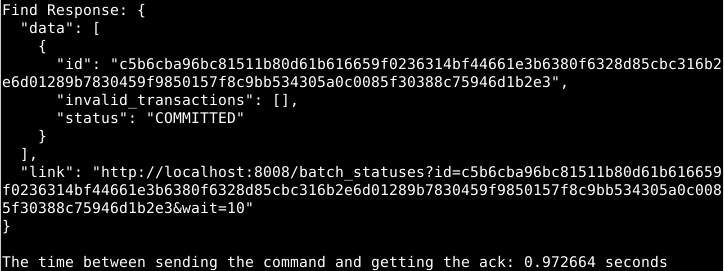}}
    \caption{Getting response of the committed transaction from the REST API}
    \label{fig:response}
    \end{figure}

    \item \textbf{Showing registered data for the client:} The client can monitor the result of the registered data by typing the defined commands. For example, according to Table~\ref{tab:commands}, the "ranks" command in our application gives the result of rankings computation for all the registered purposes. The client code is responsible for decoding the result after getting the REST API response. Figure~\ref{fig:ranks} shows a snippet of the output that the client gets by sending the command to the validator via REST API and receiving the state of the global state for rankings.
    \begin{figure}[htbp]
    \centerline{\includegraphics[width=\columnwidth]{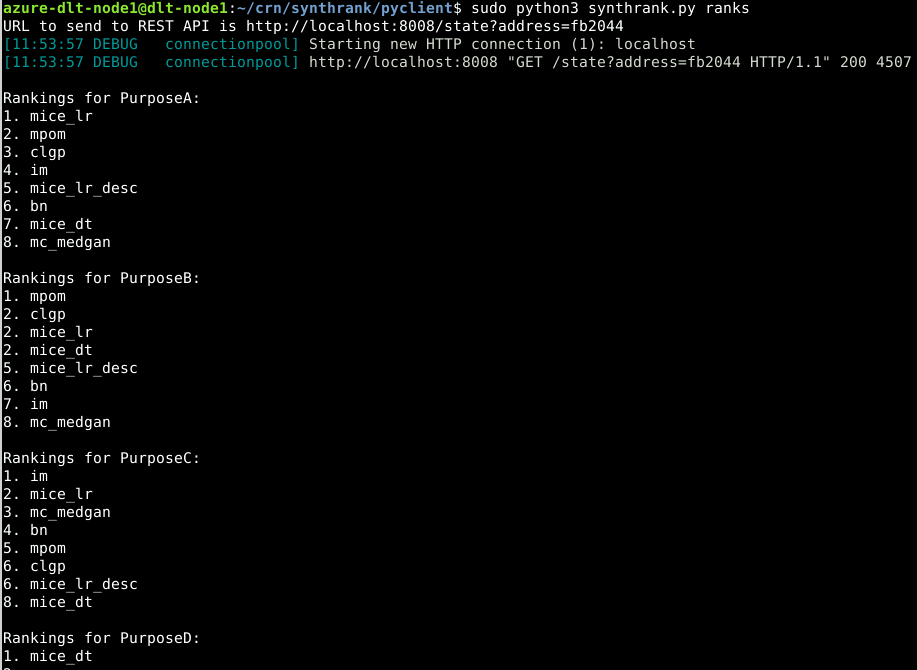}}
    \caption{Monitoring the registered data by the client}
    \label{fig:ranks}
    \end{figure}

    \item \textbf{Auditing by the authorized client:} The auditor is the only user who is allowed to check the consistency of the registered data into the ledger with the provided data by the product manager and the data scientist. The auditor gets the input files from the two other users and compare content of each with what is stored in the global state of the blockchain. This function is initiated by running the relevant command in Table~\ref{tab:commands}. Then, the auditor user will follow the steps as described in Algorithm~\ref{alg:DetailedAuditing} and get the output as shown in Figure~\ref{fig:audit}.    
    \begin{figure}[htbp]
    \centerline{\includegraphics[width=\columnwidth]{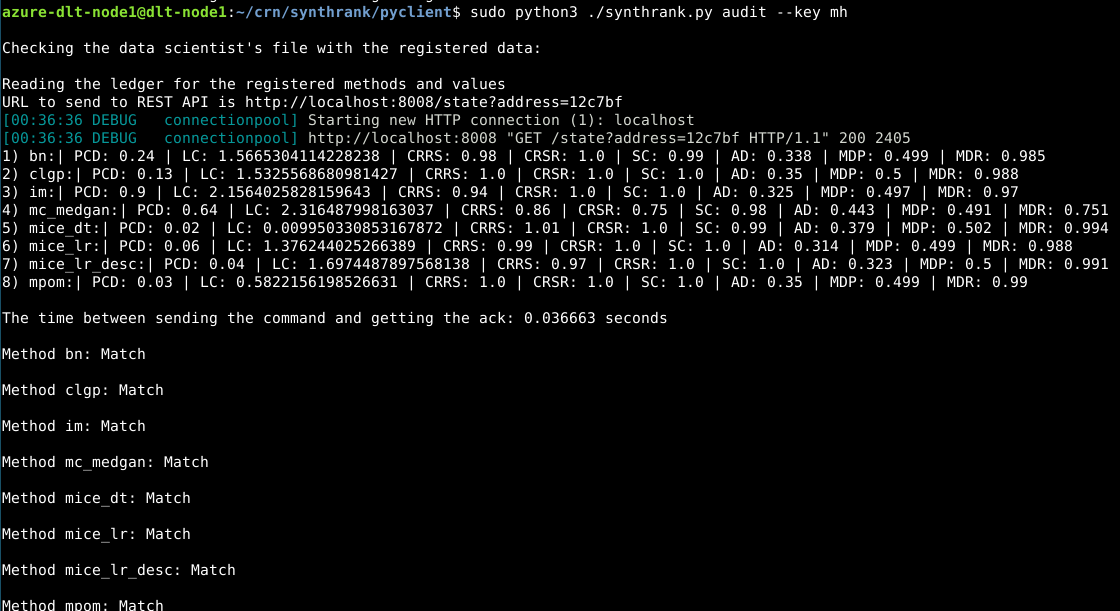}}
    \caption{Auditing the consistency between the registered data and the given files to the auditor}
    \label{fig:audit}
    \end{figure}

\end{enumerate}
\section{Experiments and Results}
\label{sec:experiments}
The experiments conducted in this study aim to validate the ability of the proposed framework to classify synthetic data generators for various purposes/scenarios. The approach focuses on evaluating these generators against a set of Quality Indicators (QIs) and metrics, reflecting their performance in terms of similarity and data privacy.

In addition to demonstrate the framework's ability to rank synthetic data generators, experiments are conducted with Sawtooth~\cite{Sawtooth} to demonstrate the workflow with various user roles, product manager, data scientists, blockchain developer and auditor. The details of the experimental setup are discussed in Section~\ref{sec:setup}.

To maintain clarity and alignment for the reader, the terminologies used in this section for synthetic data generators, quality indicators, and associated metrics are consistent with those used in~\cite{goncalves2020generation}, from which the data sets for our experiments were borrowed.

\begin{table}[h!]
    \centering
    \begin{tabular}{cp{5cm}}
        \toprule
        Notation & Description \\
        \midrule
        PCD & Pairwise Correlation Distance \\
        LC & Log Cluster \\
        CRRS & Cross classification train on real, test on synthetic \\
        CRSR & Cross classification train on synthetic, test on real \\
        SC & Support Coverage \\
        AD & Attribute Disclosure \\
        MDP & Membership Disclosure Precision \\
        MDR & Membership Disclosure Recall \\
        \bottomrule
    \end{tabular}
    \caption{Notations for Data Utility and Privacy Metrics. The definitions of PCD, LC, CRRS, CRSR, SC, MDP, and MDR are available in~\cite{goncalves2020generation}.}
    \label{tab:metrics}
\end{table}

\begin{table}[h!]
    \centering
    \begin{tabular}{cp{5cm}}
        \toprule
        Notation & Description \\
        \midrule
        im & Independent Marginals  \\
        bn & Bayesian networks  \\
        mpom & Mixture of the product of multinomials \\
        clgp & Categorical latent Gaussian process \\
        mc\_medgan & Multi categorical extension to General Adversarial networks \\
        mice\_lr & Multiple Imputations by Chained Equations with Logistic Regressions and ordered by the number of categories in an ascending manner \\
        mice\_lr\_desc & Multiple Imputations by Chained Equations with Logistic Regressions and ordered by the number of categories in descending manner \\
        mice\_dt & Multiple Imputations by Chained Equations with Decision Tree as classifier in an ascending order  \\
        \bottomrule
    \end{tabular}
    \caption{Notations and Descriptions for various Synthetic Data Generators. The detailed definitions of im, bn, mpom, clgp, mc\_medgan, mice\_lr, mice\_lr\_desc, and mice\_dt are available in~\cite{goncalves2020generation}.}
    \label{tab:techniques}
\end{table}

\subsection{Input Datasets}
Our input datasets, referenced in Table~\ref{tab:similarity} and Table~\ref{tab:data_privacy}, are derived mainly from the experimental results of Goncalves et al.~\cite{goncalves2020generation}. In their study, eight synthetic data generators were applied to the Breast Small SEER research data set (approximately 170,000 samples)~\cite{seer2023}, which produced values for various metrics. All data utility metrics in~\cite{goncalves2020generation}) for which values are available have been included in our tables. However, values only presented through the charts in the original work are not included due to the challenges in visual extraction. Our approach and experiments, as discussed in Section~\ref{sec:discussions}, are scalable to a large number of quality indicators and metrics per quality indicator. The notation for these metrics and the synthetic data generators are described in Tables~\ref{tab:metrics} and~\ref{tab:techniques}, respectively.

\begin{table}[h!]
    \centering
    \begin{tabular}{|l|c|c|c|c|c|}
        \hline
        Generators & PCD & LC & CRRS & CRSR & SC \\
        \hline
        im & 0.9 & -3.62 & 0.94 & 1.0 & 1.0 \\
        bn & 0.24 & -7.47 & 0.98 & 1.0 & 0.99 \\
        mpom & 0.03 & -10.47 & 1.0 & 1.0 & 1.0 \\
        clgp & 0.13 & -7.63 & 1.0 & 1.0 & 1.0 \\
        mc\_medgan & 0.64 & -2.12 & 0.86 & 0.75 & 0.98 \\
        mice\_lr & 0.06 & -8.3 & 0.99 & 1.0 & 1.0 \\
        mice\_lr\_desc & 0.04 & -6.8 & 0.97 & 1.0 & 1.0 \\
        mice\_dt & 0.02 & -11.25 & 1.01 & 1.0 & 0.99 \\
        \hline
    \end{tabular}
    \caption{Data utility metric values for BREAST-small SEER datasets for the eight synthetic data generators.}
    \label{tab:similarity}
\end{table}
\begin{table}[h!]
    \centering
    \begin{tabular}{|l|c|c|c|}
        \hline
        Generators & AD & MDP & MDR \\
        \hline
        im & 0.325 & 0.497 & 0.97 \\
        bn & 0.338 & 0.499 & 0.985 \\
        mpom & 0.350 & 0.499 & 0.99 \\
        clgp & 0.350 & 0.500 & 0.988 \\
        mc\_medgan & 0.443 & 0.491 & 0.751 \\
        mice\_lr & 0.314 & 0.499 & 0.988 \\
        mice\_lr\_desc & 0.323 & 0.500 & 0.991 \\
        mice\_dt & 0.379 & 0.502 & 0.994 \\
        \hline
    \end{tabular}
    \caption{Data Privacy metric values for BREAST-small SEER datasets for the eight synthetic data generators.}
    \label{tab:data_privacy}
\end{table}
\subsection{Quality Indicators, Metrics, Desired and Undesired Properties}
Our framework is engineered for flexibility and scalability, accommodating a variety of purposes, as detailed in Section~\ref{sec:discussions}. To holistically evaluate the effectiveness of synthetic data generators, we consider two primary Quality Indicators (QIs) in our experiments. 

\begin{itemize}
\item \textbf{Data Utility}: This QI gauges the extent to which synthetic data mirrors the statistical properties of real data. The specific metrics for this indicator are detailed in Table~\ref{tab:similarity}.
\item \textbf{Data Privacy}: This QI assesses the degree of privacy preservation inherent in synthetic data. Table~\ref{tab:data_privacy} lists the metrics that form this quality indicator.
\end{itemize}

We evaluated eight metrics in two quality indicators. PCD, CRRS, CRSR, LC, SC, AD, MDP, and MDR. These metrics are classified into three categories based on the interpretation of their values. For metrics such as PCD, AD, LC, MDP, and MDR, lower values indicate superior performance of the generators. In contrast, for the SC metric, higher values denote better performance. Lastly, for metrics such as CRRS and CRSR, values closer to 1 are indicative of optimal performance. 

Each of these eight metrics can be assigned to desired properties ($W^{M^{+}}$) or undesired properties ($W^{M^{-}}$) for each purpose ($p_i$). When a metric is categorized as a desired property, a generator's high performance in that metric positively influences (or rewards) its ranking. In contrast, if the same metric is classified as an undesired property, superior performance in that metric negatively impacts (or penalizes) the generator’s ranking.

These QIs along with the metrics and categorizing the metrics into desirable and undesirable properties are integral to our assessment, offering a comprehensive view of each generator's strengths and limitations for different purposes. 

\subsection{Experiment Objectives}
\label{sec:objective}
In this section, we outline the research objectives guiding our investigation into the development of our proposed ranking algorithm for synthetic data generators. This algorithm is designed to balance the importance of desired and undesired properties, tailoring these evaluations to fit specific purposes. These purposes are defined through unique weight distributions assigned to quality indicators and metrics, with further categorization of metrics into desirable and undesirable attributes. The research questions posed aim to dissect the algorithm's responsiveness and adaptability to these configurations, its comparative performance against the state-of-the-art ranking algorithm, and its integration within a permissioned blockchain framework, as outlined below.

\begin{enumerate}
\item How can the correctness of the proposed ranking algorithm be effectively validated? 

\item How do different weight distributions assigned to metrics influence the ranking result across various purposes?

\item What is the effect of varying weight distributions assigned to quality indicators on the ranking results for different purposes? 

\item How does categorizing a metric as either a desired or undesired property impact the ranking results? 

\item How does the performance of proposed solution compared to baseline solutions?

\item What is the performance of blockchain read and write operations within the context of this framework? 
\end{enumerate}

Our contributions and the interpretation of the experimental results are structured around these questions, establishing a solid foundation for evaluating our proposed framework.

To address the initial question regarding the validity of the proposed ranking algorithm, we established a set of ranking outcomes derived manually from~\cite{goncalves2020generation} as a benchmark for comparison with the output of the algorithm. For the validation process, we opted for Kendall's Tau~\cite{kendall1938new} and Spearman's Rho~\cite{spearman1961proof} as our chosen metrics. 

Kendall's Tau assesses the degree of concordance or discordance between the ranks assigned by two methods (in this case, the ground truth ranking outcome and estimated ranking outcome), considering all pairs of observations.  Kendall's Tau is calculated as follows: 
\begin{equation}
\label{eq:kendall}
    \tau = \frac{{\text{{Number of concordant pairs}} - \text{{Number of discordant pairs}}}}{{\text{{Total number of pairs}}}}
\end{equation}

In Equation~\ref{eq:kendall}, the number of concordant pairs represents the count of pairs of observations that have the same order in both rankings, while the number of discordant pairs represents the count of pairs of observations that have different orders in the two rankings. The total number of pairs represents the total number of pairs of observations. Kendall's Tau ranges from -1 to 1, where 1 indicates perfect agreement, -1 indicates perfect disagreement, and 0 indicates no association.

Similarly to Kendall's Tau, Spearman's Rho~\cite{spearman1961proof} calculates the strength and direction of the relationship between two rankings by considering the differences in ranks assigned to corresponding observations. It complements Kendall's Tau in providing a comprehensive assessment of the similarity between rankings.

Spearman's Rho (\( \rho \)) is calculated as:

\begin{equation}
\label{eq:spearman}
    \rho = 1 - \frac{{6 \times \sum(\text{{squared rank differences}})}}{{N \times (N^2 - 1)}}
\end{equation}
In Equation~\ref{eq:spearman}, \( \sum(\text{{squared rank differences}}) \) represents the sum of the squared differences in ranks between corresponding observations in the two rankings and \( N \) is the number of observations. The resulting value of \( \rho \) ranges from -1 to 1, where 1 indicates perfect agreement between the rankings, -1 indicates perfect disagreement, and 0 indicates no association.

These metrics were chosen because of their ability to precisely measure the correlation between the ranking results produced by two different configurations of the same ranking algorithm or by two distinct ranking algorithms in identical settings.

By addressing the first four questions, we validate the innovative aspects of our ranking algorithm. This includes its ability to dynamically adjust to different purposes through variable weight distributions and metric categorizations, showcasing its adaptability and the nuanced approach it offers towards evaluating synthetic data generators.

By comparing our proposed ranking algorithm with baseline ranking algorithms, we demonstrate that our algorithm provides more substantial benefits beyond the state of the art.

Furthermore, by investigating the performance of blockchain operations within our framework, we illustrate the practicality of embedding our ranking algorithm in a permissioned blockchain. This integration improves accountability, auditability, and transparency in the ranking process, representing a significant advancement over existing ranking methodologies.

Through all the research questions addressed, the results from our experiments will provide evidence supporting our contributions, demonstrating the proposed framework's ability to advance the state-of-the-art in synthetic data generator ranking.

\subsection{Configurations of the Proposed Framework for Experiments}
In this section, we present the configurations, including weight distributions for metrics and quality indicators, as well as the categorization of metrics into desired and undesired properties, used to address each of the research questions discussed in Section~\ref{sec:objective}. Given the absence of sufficient real-world data for weight distributions and categorizations, we embarked on a sensitivity analysis to rigorously evaluate the impacts of variations in weight distributions and categorization of metrics. This methodological approach ensures a robust and scientifically rigorous examination of the configurations' sensitivity, thereby enabling a comprehensive understanding of their influence under different hypothetical scenarios.  For all of our investigations, the metric values are derived from Tables~\ref{tab:similarity} and~\ref{tab:data_privacy}.

\begin{table*}[t]
\centering
\caption{Scenarios and Descriptions}
\label{tab:scenarios_descriptions}
\begin{tabular}{|c| p{5cm}| p{5cm}|}
\hline
\textbf{Scenario} & \textbf{Goals} & \textbf{Example applications} \\ \hline
Purpose A  & Ranking generators prioritizing a balance between data utility and data privacy properties with uniform distribution of quality indicator weights & Controlled public release of data \\ \hline
Purpose B  & Ranking generators prioritizing only data utility properties with uniform distribution of quality indicator weights & Education and training  \\ \hline
Purpose C  & Ranking generators prioritizing only data privacy properties with uniform distribution of quality indicator weights & Software development and testing   \\ \hline
\end{tabular}
\end{table*}

\begin{table*}[t]
\centering
\footnotesize
\caption{Uniform Distribution of Metrics Weights for Purpose A, Purpose B, Purpose C.}
\label{tab:correctness_metric_weight}
\begin{tabular}{|l|c|c|c|c|c|c|c|c|}
\hline
Purposes   & PCD   & CRRS & CRSR & AD   & MDP  & MDR  & LC   & SC   \\ \hline
Purpose A  & 0.125 &  0.125  & 0.125  & 0.125 & 0.125 & 0.125 & 0.125  & 0.125   \\ \hline
Purpose B  & 0.20 & 0.20 & 0.20  & - & - & - & 0.20 & 0.20 \\ \hline
Purpose C  &   -   & -  & - &  $\frac{1}{3}$ & $\frac{1}{3}$ & $\frac{1}{3}$& -   & -    \\ \hline
\end{tabular}
\end{table*}

\begin{table*}[t]
\centering
\footnotesize
\caption{Manually Curated Ground Truth Rankings}
\label{tab:groundTruthRanking}
\begin{tabular}{|c|c|c|c|c|c|c|c|c|}
\hline
\textbf{Purpose} & \textbf{mice\_lr} & \textbf{mpom} & \textbf{clgp} & \textbf{im} & \textbf{mice\_lr\_desc} & \textbf{bn} & \textbf{mice\_dt} & \textbf{mc\_medgan} \\ \hline
Purpose A         & 1                 & 2             & 3             & 4           & 5                       & 6            & 7                 & 8                   \\ \hline
Purpose B         & 2                 & 1             & 2             & 7           & 5                       & 6            & 2                 & 8                   \\ \hline
Purpose C         & 2                 & 5             & 6             & 1           & 6                       & 4            & 8                 & 3                   \\ \hline
\end{tabular}
\end{table*}

\begin{table*}[t]
\centering
\footnotesize
\caption{Uniform Distribution of Quality Indicator Weights}
\label{tab:correctness_QIWeights}
\begin{tabular}{|c|c|c|}
\hline
\textbf{Purposes} & \textbf{Data Utility} & \textbf{Data Privacy} \\ \hline
Purpose A         & 0.5          & 0.5      \\ \hline
Purpose B         & 0.5          & 0.5             \\ \hline
Purpose C         & 0.5          & 0.5             \\ \hline
\end{tabular}
\end{table*}

\begin{table*}[t]
\centering
\footnotesize
\caption{Predicted Rankings From the Proposed Algorithm}
\label{tab:predictedRanking}
\begin{tabular}{|c|c|c|c|c|c|c|c|c|}
\hline
\textbf{Purpose} & \textbf{mice\_lr} & \textbf{mpom} & \textbf{clgp} & \textbf{im} & \textbf{mice\_lr\_desc} & \textbf{bn} & \textbf{mice\_dt} & \textbf{mc\_medgan} \\ \hline
Purpose A         & 1                 & 2             & 3             & 4           & 5                       & 6            & 7                 & 8                   \\ \hline
Purpose B         & 2                 & 1             & 2             & 7           & 5                       & 6            & 2                 & 8                   \\ \hline
Purpose C         & 2                 & 5             & 6             & 1           & 6                       & 4            & 8                 & 3                   \\ \hline
\end{tabular}
\end{table*}

\begin{table}[h]
\centering
\caption{Kendall's Tau and Spearman's Rho Values}
\label{tab:kendall_spearman_values}
\begin{tabular}{|c|c|c|}
\hline
\textbf{Purpose} & \textbf{Kendall's Tau} & \textbf{Spearman's Rho} \\ \hline
Purpose A         & 1.0                     & 1.0                      \\ \hline
Purpose B         & 1.0                     & 1.0                      \\ \hline
Purpose C         & 1.0                     & 0.9999999999999998      \\ \hline
\end{tabular}
\end{table}

\begin{table*}[h]
\centering

\begin{minipage}{.45\linewidth}
\centering
\caption{Weight Distributions for $WM^+$ (Desired Properties):  A Specification for Understanding the Metric Weight Distributions' Impact on Rankings}
\label{tab:metric_weight_wm+}
\begin{tabular}{|l|c|c|c|c|c|c|c|c|}
\hline
Purpose   & PCD  & CRRS & CRSR & AD   & MDP  & MDR  & LC   & SC   \\ \hline
Purpose D  & 0.80 & - & 0.20 & - & -    & -    & -    & -    \\ \hline
Purpose D'  & 0.10 & -  & 0.90 & - & - & - & -    & -    \\ \hline
Purpose E  & - & - & - & 0.05& 0.05 & 0.90 & - & - \\ \hline
Purpose E'  & - & - & - & 0.90& 0.05 & 0.05 & - & - \\ \hline
\end{tabular}
\end{minipage}%
\hfill
\begin{minipage}{.45\linewidth}
\centering
\caption{Weight Distributions for $WM^-$ (Undesired Properties):  A Specification for Understanding the Metric Weight Distributions' Impact on Rankings}
\label{tab:metric_weight_wm-}
\begin{tabular}{|c|c|c|c|c|c|c|c|}
\hline
 PCD  & CRRS & CRSR & AD   & MDP  & MDR  & LC   & SC   \\ \hline
 -    & -    & -    & 0.20    & 0.80 & - & - & - \\ \hline
 -& - & - & 0.90& 0.10 &  -   & - & - \\ \hline
 0.2   & 0.2    & 0.2    & -& - & - & 0.2    & 0.2    \\ \hline
 0.02   & 0.02 & 0.04 & -    & -    & -    & 0.02    & 0.9    \\ \hline
\end{tabular}
\end{minipage}
\end{table*}

\begin{table*}[h]
\centering
\begin{minipage}{.45\linewidth}
\centering
\caption{Weight Distributions for $WM^+$ (Desired Properties): A Specification for Understanding the Quality Indicator Distributions' Impact on Rankings}
\label{tab:qi_weight_wm+}
\begin{tabular}{|l|c|c|c|c|c|c|c|c|}
\hline
Purpose   & PCD  & CRRS & CRSR & AD   & MDP  & MDR  & LC   & SC   \\ \hline
Purpose A  & 0.80 & - & - & - & 0.20 & -  & -    & -    \\ \hline
Purpose B  & 0.80 & -  & - & - & 0.20 & - & -    & -    \\ \hline
Purpose C  & - & - & - & 0.05& 0.05 & 0.90 & - & - \\ \hline
Purpose D  & - & - & - & 0.05& 0.05 & 0.90 & - & - \\ \hline
\end{tabular}
\end{minipage}%
\hfill
\begin{minipage}{.45\linewidth}
\centering
\caption{Weight Distributions for $WM^-$ (Undesired Properties): A Specification for Understanding the Quality Indicator Distributions' Impact on Rankings}
\label{tab:qi_weight_wm-}
\begin{tabular}{|c|c|c|c|c|c|c|c|}
\hline
 PCD  & CRRS & CRSR & AD   & MDP  & MDR  & LC   & SC   \\ \hline
 -    & -    & -    & 0.80   & - & - & 0.20 & - \\ \hline
 -& - & - & 0.80& - &  -   & 0.20 & - \\ \hline
 0.2   & 0.2    & 0.2    & -& - & - & 0.2    & 0.2    \\ \hline
 0.2   & 0.2 & 0.2 & -    & -    & -    & 0.2    & 0.2   \\ \hline
\end{tabular}
\end{minipage}
\end{table*}

\begin{table*}[h]
\centering
\begin{minipage}{.45\linewidth}
\centering
\caption{Weight Distributions for $WM^+$ (Desired Properties): A Specification for Understanding the Role of Desired and Undesired Properties}
\label{tab:desired_weight_wm+}
\begin{tabular}{|l|c|c|c|c|c|c|c|c|}
\hline
Purpose   & PCD  & CRRS & CRSR & AD   & MDP  & MDR  & LC   & SC   \\ \hline
Purpose H  & 0.80 & - & 0.20 & - & -    & -    & -    & -    \\ \hline
Purpose H'  &  & -  & - & 0.20 & 0.80 & - & -   & -    \\ \hline
Purpose I  & - & - & - & 0.05& 0.05 & 0.90 & - & - \\ \hline
Purpose I'  & 0.2 & - & - & - & - & - & 0.2 & 0.6 \\ \hline
\end{tabular}
\end{minipage}%
\hfill
\begin{minipage}{.45\linewidth}
\centering
\caption{Weight Distributions for $WM^-$ (Undesired Properties): A Specification for Understanding the Role of Desired and Undesired Properties}
\label{tab:desired_weight_wm-}
\begin{tabular}{|c|c|c|c|c|c|c|c|}
\hline
 PCD  & CRRS & CRSR & AD   & MDP  & MDR  & LC   & SC   \\ \hline
 -    & -    & -    & 0.20    & 0.80 & - & - & - \\ \hline
 0.80& - & 0.20 & -& - &  -   & - & - \\ \hline
0.2 & - & - & - & - & - & 0.2 & 0.6 \\ \hline
- & - & - & 0.05& 0.05 & 0.90 & - & - \\ \hline
\end{tabular}
\end{minipage}
\end{table*}

\begin{table*}[h]
\centering
\begin{minipage}{.45\linewidth}
\centering
\caption{Weight Distributions for $WM^+$ (Desired Properties): A Specification for Understanding the Performance of Sawtooth.}
\label{tab:blockchain_weight_wm+}
\begin{tabular}{|l|c|c|c|c|c|c|c|c|}
\hline
Purpose   & PCD  & CRRS & CRSR & AD   & MDP  & MDR  & LC   & SC   \\ \hline
Purpose A  & 0.40 & - & 0.60 & - & - & -  & -    & -    \\ \hline
Purpose B  & - & 0.40  & 0.30 & 0.30 & - & - & -    & -    \\ \hline
Purpose C  & - & - & - & - & - & 0.30  & - & 0.70 \\ \hline
Purpose D  & - & - & - & - & 0.40 & - & 0.60 & - \\ \hline
Purpose E  & - & - & 0.2 & 0.40 & 0.40 & - & - & - \\ \hline
Purpose F  & - & 0.10 & 0.10 & 0.80 & - & - & - & - \\ \hline
\end{tabular}
\end{minipage}%
\hfill
\begin{minipage}{.45\linewidth}
\centering
\caption{Weight Distributions for $WM^-$ (Undesired Properties): A Specification for Understanding the Performance of Sawtooth.}
\label{tab:blockchain_weight_wm-}
\begin{tabular}{|c|c|c|c|c|c|c|c|}
\hline
 PCD  & CRRS & CRSR & AD   & MDP  & MDR  & LC   & SC   \\ \hline
 -    & -    & -    & -   & 0.60 & 0.40 & 0.20 & - \\ \hline
 0.40 & - & - & -& 0.60 &  - & - & - \\ \hline
0.70  & -   & -& -& - & - & 0.30 & -    \\ \hline
-   & 0.30 & 0.70 & -    & -    & -    & -   & -   \\ \hline
- & 0.40 & - & -    & -    & -    & 0.60   & -   \\ \hline
0.60   & - & - & -    & 0.40    & -    & -    & -   \\ \hline
\end{tabular}
\end{minipage}
\end{table*}

\begin{table}[h]
\centering
\footnotesize
\begin{tabular}{|c|c|c|}
\hline
\textbf{Purposes} & \textbf{Data utility} & \textbf{Privacy} \\ \hline
Purpose A         & 0.5          & 0.5      \\ \hline
Purpose B         & 0.9          & 0.1             \\ \hline
Purpose C         & 0.1          & 0.9             \\ \hline
\end{tabular}
\caption{Quality Indicator (Category) Weights for Different Purposes: A Specification to Assess Variability in Ranking and to Evaluate Comparison of Proposed Method V1 and V2 with the Baseline Method. }
\label{tab:comparison_category_weights}
\end{table}

\begin{table}[h]
\centering
\footnotesize
\begin{tabular}{|c|c|c|}
\hline
\textbf{Purposes} & \textbf{Data utility} & \textbf{Privacy} \\ \hline
Purpose D         & 0.5                   & 0.5            \\ \hline
Purpose D'        & 0.5                   & 0.5              \\ \hline
Purpose E        & 0.5                   & 0.5              \\ \hline
Purpose E'         & 0.5                   & 0.5             \\ \hline
\end{tabular}
\caption{Quality Indicator (Category) Weights for Different Purposes:  A Specification for Understanding the Metric Weight Distributions' Impact on Rankings }
\label{tab:metric_category_weights}
\end{table}

\begin{table}[h]
\centering
\footnotesize
\begin{tabular}{|c|c|c|}
\hline
\textbf{Purposes} & \textbf{Data utility} & \textbf{Privacy} \\ \hline
Purpose F         & 0.9               & 0.1      \\ \hline
Purpose F'         & 0.1                  & 0.9             \\ \hline
Purpose G        & 0.1                  & 0.9              \\ \hline
Purpose G'         & 0.9                   & 0.1             \\ \hline
\end{tabular}
\caption{Quality Indicator (Category) Weights for Different Purposes:  A Specification for Understanding the Quality Indicator Weight Distributions' Impact on Rankings }
\label{tab:qi_category_weights}
\end{table}

\begin{table}[h]
\centering
\footnotesize
\begin{tabular}{|c|c|c|}
\hline
\textbf{Purposes} & \textbf{Data utility} & \textbf{Privacy} \\ \hline
Purpose H         & 0.5               & 0.5      \\ \hline
Purpose H         & 0.5                 & 0.5             \\ \hline
Purpose I         & 0.5                  & 0.5              \\ \hline
Purpose I'         & 0.5                   & 0.5             \\ \hline
\end{tabular}
\caption{Quality Indicator (Category) Weights for Different Purposes: A Specification for Understanding the role of  Desired and Undesired Properties }
\label{tab:desired_category_weights}
\end{table}

\begin{table}[h]
\centering
\footnotesize
\begin{tabular}{|c|c|c|}
\hline
\textbf{Purposes} & \textbf{Data utility} & \textbf{Privacy} \\ \hline
Purpose A         & 0.7          & 0.3      \\ \hline
Purpose B         & 0.7          & 0.3             \\ \hline
Purpose C         & 0.5          & 0.5             \\ \hline
Purpose D         & 0.3          & 0.7             \\ \hline
Purpose E         & 0.5          & 0.5             \\ \hline
Purpose F         & 0.7          & 0.3             \\ \hline
\end{tabular}
\caption{Quality Indicator (Category) Weights for Different Purposes: A Specification for Understanding the Performance of Sawtooth.}
\label{tab:blockchain_qi}
\end{table}

\begin{enumerate}
\item How can the correctness of the proposed ranking algorithm be effectively validated? 
\begin{itemize}
    \item To address this question, Section~\ref{sec:correctness} utilizes scenarios defined in Table~\ref{tab:scenarios_descriptions}, metric weights, and Quality Indicator weights from Tables~\ref{tab:correctness_metric_weight} and~\ref{tab:correctness_QIWeights}, respectively. Additionally, the ground-truth rankings from Table~\ref{tab:groundTruthRanking} are employed.

\end{itemize}

  \item  How do different weight distributions assigned to
metrics influence the ranking result for various
purposes?
  \begin{itemize}
    \item Section~\ref{sec:metric_weight} employs the configurations outlined in Tables~\ref{tab:metric_weight_wm+},~\ref{tab:metric_weight_wm-} and~\ref{tab:metric_category_weights} to explore this question.  
  \end{itemize}

 \item  What is the effect of varying weight distributions
assigned to quality indicators on ranking results
for different purposes?
  \begin{itemize}
    \item  To investigate this question, Section~\ref{sec:qi_Weight} uses the configurations in Tables~\ref{tab:qi_weight_wm+},~\ref{tab:qi_weight_wm-} and~\ref{tab:qi_category_weights}.  
  \end{itemize}

   \item  How does categorizing a metric as either a desired or
undesired property impact the ranking results?
  \begin{itemize}
    \item  Section~\ref{sec:desired} explores this question employing the setups described in Tables~\ref{tab:desired_weight_wm+},~\ref{tab:desired_weight_wm-} and~\ref{tab:desired_category_weights}.   
  \end{itemize}

   \item How does the performance of proposed solution compared to baseline solutions?
  \begin{itemize}
\item   To asssess the performance of the proposed solution relative to the baseline solutions, Section~\ref{sec:comparison} uses the predefined scenarios detailed in Table~\ref{tab:scenarios_descriptions}. For the implementation of baseline solutions, the metric weights are sourced from Table~\ref{tab:correctness_metric_weight}. Ground-truth rankings are available in Table~\ref{tab:groundTruthRanking} to establish a benchmark for comparison. In particular to the proposed solution, the weights from Table~\ref{tab:correctness_metric_weight} are specifically applied as undesired metric weights, adding a unique dimension to how the proposed solution is evaluated against the baselines. The QI weights for the proposed solution is sourced from Table~\ref{tab:correctness_QIWeights}.
  \end{itemize}

    \item  What is the performance of blockchain read and write
operations in the context of this framework?
  \begin{itemize}
    \item To answer this question, Section~\ref{sec:blockchain} uses the configurations presented in Table~\ref{tab:blockchain_weight_wm+},~\ref{tab:blockchain_weight_wm-}, and~\ref{tab:blockchain_qi}. 

  \end{itemize}  
\end{enumerate}

\subsection{Experimental setup}
\label{sec:setup}
We have established a permissioned blockchain network utilizing Sawtooth version 1.2 at four nodes. These nodes are hosted on virtual machines running Ubuntu 18.04, strategically located in the Azure Norway East region, to optimize performance and reliability for our application. The deployment of our Sawtooth network is critical for the operation of our application, necessitating that each node operates both the same consensus engine and transaction processor to maintain network integrity.

Our choice of the PBFT protocol (Practical Byzantine Fault Tolerance) as the consensus engine is pivotal, ensuring network reliability as long as the count of faulty nodes remains below one-third of the total network size (n). This requirement underlines the necessity of operating at least four Sawtooth nodes to effectively manage Byzantine failures. Tailored to our specific use case, we have allocated node responsibilities as follows: one node for data scientists, one for the Product Manager, and two for the Auditors, thereby ensuring a comprehensive coverage of all essential roles within our network. Below, we detail the procedural steps for configuring a robust Sawtooth network.

\begin{enumerate}

\item \textbf{Installing Sawtooth components on machines:} We installed Sawtooth 1.2 and its components on four Ubuntu 18.04 machines. After installing Sawtooth on the nodes, we have access to configure and run the Sawtooth components.

\item \textbf{Generating validator keys:} In order to involve each of the Sawtooth nodes in the process of validating transactions and blocks, we generated validator keys on each node as shown in Figure~\ref{fig:sawadm}. This command generates and stores private and public keys of the validator.  
    \begin{figure}[htbp]
    \centerline{\includegraphics[width=\columnwidth]{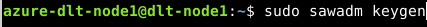}}
    \caption{Generating validator keys}
    \label{fig:sawadm}
    \end{figure}

By default, the public and private keys of the validator are accessible as presented in Figure~\ref{fig:validator_keys}. The validator's public key is used for finding a validator in the peer-to-peer network, and the validator's private key is used for creating the blocks.
    \begin{figure}[htbp]
    \centerline{\includegraphics[width=\columnwidth]{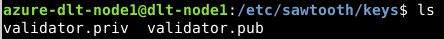}}
    \caption{Accessing validator keys}
    \label{fig:validator_keys}
    \end{figure}

\item \textbf{Creating the genesis block:} Before running all the Sawtooth nodes, we chose one of the nodes as the creator of the genesis block. A genesis block includes the initial blockchain network settings like the consensus algorithm and public keys of the other validators. Only the first node creates the genesis block and the other validator nodes, public keys of which are defined in the settings of the genesis block, access the network configuration while joining the network.

To be able to create the genesis block on the first node and set or change Sawtooth settings, we need to have a user key in our Sawtooth node. Figure~\ref{fig:user_keys} shows the command to generate user keys and the location where the keys are stored.
    \begin{figure}[htbp]
    \centerline{\includegraphics[width=\columnwidth]{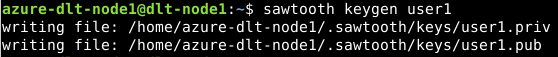}}
    \caption{Generating user keys}
    \label{fig:user_keys}
    \end{figure}

The user can then create a batch with the setting proposal for the genesis block and define the output file name of the generated settings (see Figure~\ref{fig:settings_proposal}).
    \begin{figure}[htbp]
    \centerline{\includegraphics[width=\columnwidth]{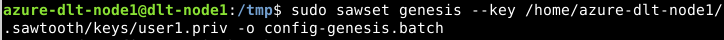}}
    \caption{Creating a settings proposal for the genesis block}
    \label{fig:settings_proposal}
    \end{figure}

The user needs to create another batch to initialize the consensus settings for the genesis block as shown in figure~\ref{fig:consensus_settings}. In the consensus settings command, the user defines the output file of the generated consensus settings batch, the consensus algorithm and its version, and also the public key of each PBFT member involved in the consensus process.   
    \begin{figure}[htbp]
    \centerline{\includegraphics[width=\columnwidth]{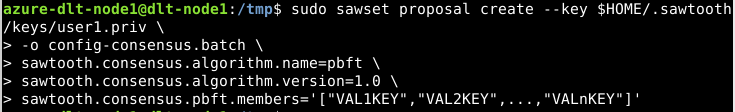}}
    \caption{Creating a consensus settings batch for the genesis block}
    \label{fig:consensus_settings}
    \end{figure}

Finally, the Sawtooth validator user combines the separate batches into a single genesis batch with the command in Figure~\ref{fig:genesis_batch} that will be committed in the genesis block after the first node has started.
    \begin{figure}[htbp]
    \centerline{\includegraphics[width=\columnwidth]{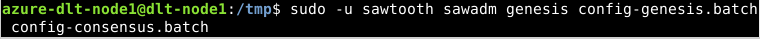}}
    \caption{Creating a genesis batch for the genesis block}
    \label{fig:genesis_batch}
    \end{figure}

\item \textbf{Configuring the validators:} All the Sawtooth component bind settings and endpoint settings are configured in the validator configuration file. Figure~\ref{fig:validator_config} shows how to create the validator configuration file based on the default example file avaliable in /etc/sawtooth directory.
    \begin{figure}[htbp]
    \centerline{\includegraphics[width=\columnwidth]{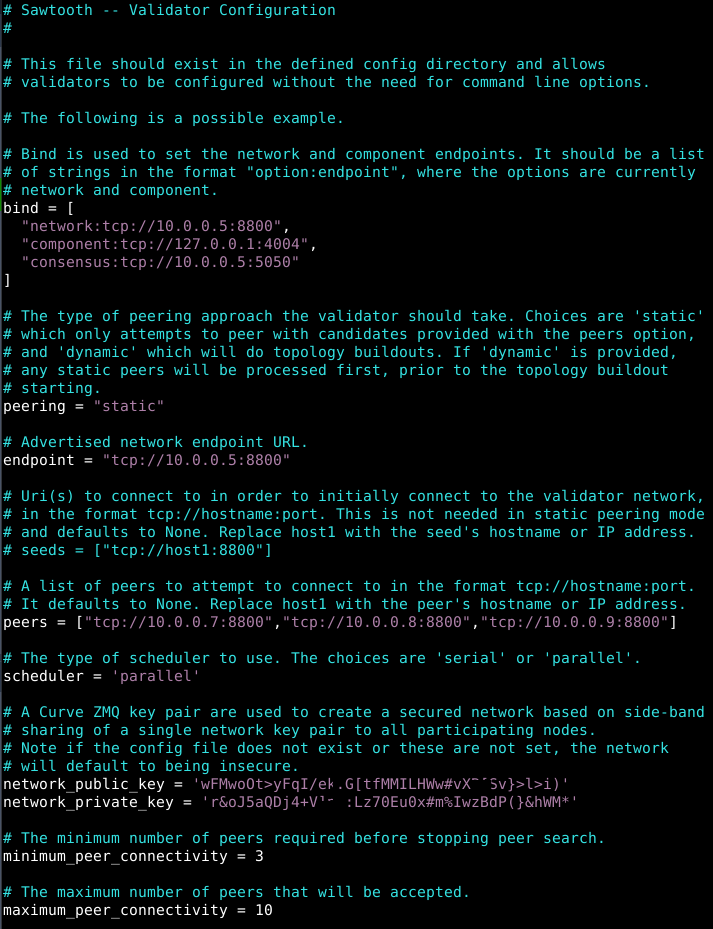}}
    \caption{Creating a validator configuration file}
    \label{fig:validator_config}
    \end{figure}

The network and component endpoints and the list of peers must be set based on the IP addresses of the Sawtooth components and the network peers. The rest of the settings in the file are set by default and can be modified according to the requirements of the use case. 
\item \textbf{Running the components:} The final step of deploying the Sawtooth network is to start the Sawtooth nodes. Starting the first node will create the genesis block and the initial settings of the network. So, it is important to run all components of the first node before the other peers. Therefore, for each node, the validator components are started in a separate terminal window by their related command:
\begin{itemize}
    \item{Validator:} sudo -u sawtooth sawtooth-validator -v
    \item{REST API:} sudo -u sawtooth sawtooth-rest-api -v 
    \item{Transaction processor settings:} sudo -u sawtooth settings-tp -v
    \item{Consensus engine:} sudo -u sawtooth pbft-engine -v --connect tcp://{IP}:{Port} 
\end{itemize}
\end{enumerate}
Now the Sawtooth network and its components are ready for deploying smart contracts on it. For this setup, we have written smart contracts in Python version 3.6.9, also known as transaction families in Sawtooth terminology. These smart contracts are designed for product managers, data scientists, and auditors to register their respective input onto the Sawtooth blockchain.

Input registration is expected to follow a sequential process. First, the product manager provides the specifications. Following this, the data scientist reviews the criteria set by the product manager, generates synthetic data using multiple generators, and then delivers the performance results of these synthetic data based on the metrics established by the product manager. As soon as the data scientist's inputs are recorded in Sawtooth, the smart contract implementing Algorithms 1 and 2 computes the ranking of generators and stores these results in Sawtooth ledger. An auditor then accesses files containing inputs collected by the product manager from stakeholders and those generated by the data scientists. These are then verified against the recorded inputs in Sawtooth according to Algorithm 3.

Sections~\ref{sec:smartcontractDeployment} and~\ref{sec:clientAPI} provide more detailed information about the implementation of smart contracts and the invocation of these smart contracts by clients (product managers, data scientists, and auditors).

\subsection{Results}

In this section, we discuss our findings on each of the research questions discussed in Section~\ref{sec:objective}.

\subsubsection{Correctness of the Proposed Ranking Algorithm}
\label{sec:correctness}

To validate the correctness of a proposed ranking algorithm, we compared the ground truth ranking with the ranking outcome of the proposed algorithm (refer to Table~\ref{tab:predictedRanking}) using Kendall's Tau and Spearman's Rho metrics. Kendall's Tau and Spearman's Rho are both measures of rank correlation used to assess the similarity between two rankings.

In Table~\ref{tab:kendall_spearman_values}, we present the results of Kendall's Tau and Spearman's Rho for all three scenarios. The findings suggest a perfect alignment between the estimated ranking and the ground truth. In particular, the ground truth did not specify any metric weights or Quality Indicator (QI) weights, nor did it categorize the metrics as desired or undesired for any scenario. Consequently, we categorized our metric weights as only desired, assuming a uniform weight distribution across all desired metrics and QI weights in the configurations of our proposed solution across all three scenarios. Since there was no categorization of metrics as undesired, undesired metric weights were set to 0, resulting in an undesired score of 0 and only desired scores were computed for each generator. Furthermore, given the uniformity of the desired metric weights and the QI weights, their impact on computing the final ranking result was negligible.

\subsubsection{Impact of Metric Weight Distributions}
\label{sec:metric_weight}

Our analysis focused on two pairs of purposes: a) Purpose D and Purpose D'; and b) Purpose E and Purpose E'. Each pair consists of identical metric types, uniform weight distributions of quality indicators, and the same categorizations of metrics into desired and undesired properties, with the primary distinction being the specific weight distribution of the metric.

\begin{table}[htbp]
    \centering
    \caption{Impact of Metric Weight Distributions on Similarity Ranking Scores }
    \begin{tabular}{lcc}
        \toprule
        \textbf{Scenarios} & \textbf{Kendall's Tau} & \textbf{Spearman's Rho} \\
        \midrule
        Purpose D \& D' & 0.357 & 0.381 \\
        Purpose E \& E' & -0.071 & -0.024 \\
        \bottomrule
    \end{tabular}
    \label{tab:metric_weight_comparison}
\end{table}

Table~\ref{tab:metric_weight_comparison} provides key insights. For the Purpose D and D' scenarios, where the metric weight distributions are notably different, the moderate positive correlation between rankings suggests that some generators may perform better than others in those metrics. This also indicates that certain generators exhibit stronger performance across the specified metrics, leading to a consistent pattern of rankings across different metric weight distributions. Essentially, generators performing well for one metric weight distribution (Purpose D) tend to perform well for the other (Purpose D') as well.

In contrast, the Purpose E and E' scenarios exhibit a very weak negative correlation between rankings. This suggests that certain generators, which may perform well under one set of metric weight distributions, may not maintain their performance when the distributions change. The high sensitivity of generators to changes in metric weights implies that these generators may have performance characteristics that are highly influenced by specific metrics and their associated weights.

Our results illustrate how the weight configurations and the selection of metrics can significantly impact the ranking outcome in the proposed solution. 
\subsubsection{Influence of Quality Indicators Weight Distributions}
\label{sec:qi_Weight}

Our analysis focused on two pairs of purposes: a) Purpose F and Purpose F'; and b) Purpose G and Purpose G'. Each pair consists of identical metric types, metric weight distributions, and categorizations of metrics into desired and undesired properties, with the primary distinction being the specific weight distributions of the quality indicator.

\begin{table}[htbp]
    \centering
    \caption{Impact of QI Weight Distributions on Similarity Ranking Scores }
    \begin{tabular}{lcc}
        \toprule
        \textbf{Scenarios} & \textbf{Kendall's Tau} & \textbf{Spearman's Rho} \\
        \midrule
        Purpose F \& F' & 0.07 & 0.11 \\
        Purpose G \& G' & 0.14 & 0.21 \\
        \bottomrule
    \end{tabular}
    \label{tab:QI_weight_comparison}
\end{table}

The analysis in Table~\ref{tab:QI_weight_comparison} reveals that despite the drastic changes in the weights of the quality indicator (QI) between pairs of purposes (F \& F', G \& G'), there is still some degree of similarity in the ranking outputs. However, the correlations indicate that these similarities are relatively weak, suggesting that changes in QI weights have influenced the ranking outcomes in our proposed solution to some extent. However, the slightly higher correlation coefficients observed for Purpose G \& G' compared to Purpose F \& F' suggest a slightly stronger consistency in the classification of outputs when the desired and undesired properties are separated into distinct categories (such as data utility and data privacy metrics).

\subsubsection{Role of Desired and Undesired Properties}
\label{sec:desired}

Our analysis focused on two pairs of purpose: a) Purpose H and Purpose H'; and b) Purpose I and Purpose I'. For each pair, we swapped the metrics in the desired properties with those in undesired properties, while maintaining the same metric weights and quality indicator weight distributions. 

\begin{table}[htbp]
    \centering
    \caption{Impact of Categorizing Metrics as Desired and Undesired Properties on Similarity Ranking Scores }
    \begin{tabular}{lcc}
        \toprule
        \textbf{Scenarios} & \textbf{Kendall's Tau} & \textbf{Spearman's Rho} \\
        \midrule
        Purpose H \& H' & -0.28 & -0.42 \\
        Purpose I \& I' & -0.57 & -0.76 \\
        \bottomrule
    \end{tabular}
    \label{tab:desired_weight_comparison}
\end{table}

Table~\ref{tab:desired_weight_comparison} provides crucial insight into the analysis and highlights the substantial impact of switching metrics between desired and undesired properties on the similarity ranking scores for both pairs of purposes. In particular, the comparison indicates that Purpose I \& I' are more profoundly affected compared to Purpose H \& H'.

The notably larger negative values observed for both Kendall's Tau and Spearman's Rho metrics suggest that generators within the Purpose I \& I' exhibit heightened sensitivity to changes in categorization of metrics compared to those in the H \& H'.

\subsubsection{Comparison with Baselines}
\label{sec:comparison}

We compare the proposed solution with two baselines in three scenarios presented in Table~\ref{tab:scenarios_descriptions} : a) weighted sum normalized average score b) weighted sum rank-derived score. 

In weighted sum normalized average score, each metric is first normalized according to whether higher or lower values are preferable, thereby ensuring that all metrics are on a comparable scale. This is followed by applying weights that reflect the importance of each metric within the given scenario. The weighted values are then averaged to produce a single score that reflects the overall performance of each generator on the metrics considered. 

\begin{table}[htbp]
    \centering
    \caption{Comparison between weighted normalized average score and ground truth}
    \begin{tabular}{lcc}
        \toprule
        \textbf{Scenarios} & \textbf{Kendall's Tau} & \textbf{Spearman's Rho} \\
        \midrule
        Purpose A & 0.28 & 0.38 \\
        Purpose B  & 0.79 & 0.90 \\
        Purpose C & 0.61 & 0.79 \\
        \bottomrule
    \end{tabular}
    \label{tab:baseline1}
\end{table}

Table~\ref{tab:baseline1} summarizes the correlation between the weighted normalized average scores and the ground truth in three scenarios.

Purpose A demonstrates relatively low correlation values with Kendall's Tau and Spearman's Rho, indicating weak alignment between the weighted scores and the ground truth. Scenario B shows significantly higher correlation values, with Kendall's Tau and Spearman's Rho, suggesting strong agreement between the weighted scores and the ground truth.

Purpose C exhibits moderate to high correlation values with Kendall's Tau and Spearman's Rho, indicating a good degree of consistency between the weighted scores and the ground truth. Overall, the results indicate varying levels of agreement between the algorithm's outputs and the actual ground truth across different scenarios, with Purpose B displaying the best alignment and Purpose A the least.

In essence, the weighted average score approach does not perform as well as the proposed solution, which demonstrated excellent consistency of ranking outcomes between the proposed solution and the ground truth (refer to Section~\ref{sec:correctness}).

\begin{table}[htbp]
    \centering
    \caption{Comparison between weighted rank derived score and ground truth}
    \begin{tabular}{lcc}
        \toprule
        \textbf{Scenarios} & \textbf{Kendall's Tau} & \textbf{Spearman's Rho} \\
        \midrule
        Purpose A & 0.9999999999999998 & 1 \\
        Purpose B  & 1 & 1 \\
        Purpose C & 1 & 0.9999999999999998 \\
        \bottomrule
    \end{tabular}
    \label{tab:baseline2}
\end{table}

Our second baseline solution for comparison is weighted sum rank-derived score. In our baseline implementation, we have adopted an inspired variant of the ranking methodology proposed by~\cite{yan2022multifaceted}. While~\cite{yan2022multifaceted} focus on ranking models based on individual dataset performances, our adaptation works with preaveraged data. This shift addresses the practical constraint of using existing average metric scores, avoiding the need for a granular dataset analysis. Consequently, our implementation modifies the ranking mechanism to suit these average scores while preserving the fundamental principle. 

In our baseline implementation of weighted rank-derived score, instead of directly averaging the metric values, each metric is first ranked for all generators. These ranks are then weighted using the same set of predefined importance weights as in the normalized average score approach. The final score for each model is derived by calculating the average of these weighted ranks. This method focuses on the relative standing of each generator rather than their absolute metric values. 

Table~\ref{tab:baseline2} shows that, in all scenarios (Purpose A, Purpose B and Purpose C), both Kendall's Tau and Spearman's Rho are at scores of 1. This indicates complete alignment between the weighted rank-derived scores and the ground truth across all scenarios, suggesting that the weighted rank-derived scores perform as well as the proposed solution.

\begin{table}[htbp]
    \centering
    \caption{Comparison between the weighted rank-derived score and the proposed solution, when metrics are categorized solely as undesired in the proposed solution and the same metrics are categorized as desired in the baseline.}
    \begin{tabular}{lcc}
        \toprule
        \textbf{Scenarios} & \textbf{Kendall's Tau} & \textbf{Spearman's Rho} \\
        \midrule
        Purpose A & -0.28 & -0.38 \\
        Purpose B  & -0.79 & -0.90 \\
        Purpose C & -0.61 & -0.79 \\
        \bottomrule
    \end{tabular}
    \label{tab:baseline2}
\end{table}
Table~\ref{tab:baseline2} provides a comparison between the weighted rank-derived scores and the proposed solution in three different scenarios (A, B, and C). In these cases, metrics are categorized solely as undesired in the proposed solution, while the same metrics are treated as desired in the baseline.

The correlations, measured using Kendall's Tau and Spearman's Rho, are negative across all scenarios. These negative values indicate an inverse relationship between the weighted rank-derived scores and the proposed solution, suggesting significant differences in how each approach handles metrics classified as undesired versus desired.

In the baseline weighted sum ranking score approach, scores are derived considering only desired attributes, with varying weights assigned to modulate the significance of these attributes in determining rankings. This method primarily adjusts the rankings by increasing the influence of more important metrics through higher weights and decreasing the influence of less critical metrics through lower weights. However, this approach lacks the flexibility to adequately express the negative impact of undesired attributes—metrics that should detract from a score rather than contribute to it. By categorizing certain metrics as undesired and adjusting the scores to reflect their negative impact, we can more accurately represent their true influence. The proposed method allows active reduction of the score based on the presence of these undesirable characteristics, which can significantly alter ranking outcomes, especially when such attributes are pronounced.

In essence, these methodological variances can lead to differing ranking outcomes between the baseline and proposed methods for the same purpose, reflecting a more multifaceted evaluation in the proposed method. The weighted ranking-derived score method may be suitable for scenarios prioritizing simplicity and direct comparisons (without the need of undesired properties and QI weights), while the proposed method is advantageous in contexts requiring a detailed and nuanced understanding of performance across diverse metrics.

\subsubsection{Performance of Read and Write Operations to Sawtooth}
\label{sec:blockchain}

\begin{figure*}[htbp]
\centering
\includegraphics[width=\textwidth]{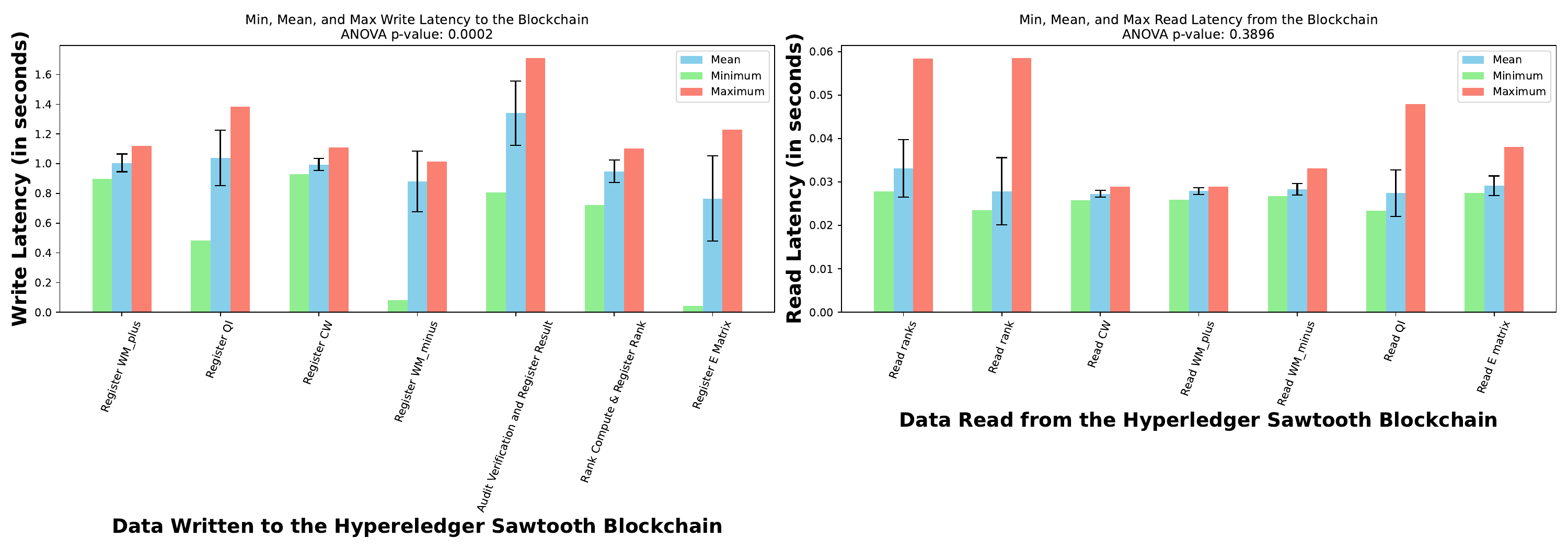}
\caption{Write Latency and Read Latency of Sawtooth. Table~\ref{tab:blockchain} defines the X axis for write latency. Table~\ref{tab:blockchain_read} defines the X axis for read latency.}
\label{fig:blockchain}
\end{figure*}

The primary objective of our performance experiments is to characterize the latency involved in read and write operations associated with our ranking and audit algorithms on the Sawtooth blockchain. The experimental setup is described in Section~\ref{sec:setup}. 

In our write-latency characterization experiment, we measured and compared the write latencies of inputs from data scientists, product managers, and auditors. These data sets, written on the blockchain, are described in Table~\ref{tab:blockchain} and referenced on the X-axis in "Write Latency to the Blockchain" (Figure~\ref{fig:blockchain}).

\begin{table}[h!]
    \centering
    \begin{tabular}{p{3cm}|p{4.8cm}}
        \toprule
        Dataset Notations & X Axis Description of Write Latency in Figure~\ref{fig:blockchain} \\
        \midrule
        Register E Matrix & E matrix from data scientists, where each row represents a generator and each column represents the results of the metric evaluated on the generator. The write latency described in Figure~\ref{fig:blockchain} for DS Input (E) represents the time taken to record one row of the E matrix in the blockchain.
 \\ \hline 
        Audit Verification and Register Result & The time taken to execute Algorithm 3. The execution of Algorithm 3 involves parsing the input files provided by the data scientist and the product manager, which require validation against blockchain records. Additionally, this process includes retrieving the rank data from the ledger and cross-verifying it against the auditor's dynamically computed ranks. Upon completion of these steps, the final result of the audit is recorded on the Sawtooth blockchain. \\ \hline 
        Register QI & The time taken to record each row in the QI matrix, where each row represents a Quality Indicator, and each column represents metrics that include both desired and undesired properties. This input is provided by the product manager.\\ \hline 
        Register WM\_minus & The time taken to record each row in the $WM^{-}$ matrix, where each row represents the purpose, and each column represents the weights for metrics associated with undesired properties. This input is provided by the product manager.\\ \hline 
        Register WM\_plus & The time taken to record each row in the $WM^{+}$ matrix, where each row represents the purpose, and each column represents the weights for metrics associated with desired properties. This input is provided by the product manager. \\ \hline 
        Register CW & The time taken to record each row in the QI weight matrix where the row represents the purpose and the column represents the weights for the QIs that are relevant for the purpose. This input is provided by the product manager. \\ \hline 
        Rank Compute and Register Rank & The time taken to compute the ranking per purpose and record the ranking per purpose in the blockchain. The computations and registering the ranks to the blockchain is done by smart contracts. \\ 
        
        \bottomrule
    \end{tabular}
    \caption{Notations and descriptions for the datasets that are written to the blockchain}. 
    \label{tab:blockchain}
\end{table}

In this experiment, each data set was written to the blockchain 10 times, resulting in 10 samples per data set. The results, presented in Figure~\ref{fig:blockchain}, revealed an ANOVA p-value of 0.0002. This suggests statistically significant difference in write latencies across the seven datasets, as the p-value is less than the conventional 0.05 threshold. Our experiments also revealed that when we compared the write latencies without the audit part (Audit Verification and Result), the ANOVA p-value was 0.1194 indicating that there was no statistically significant write latency across the six datasets. This concludes that audit verification and writing the result back to the blockchain is relatively more time-consuming than other operations, which is expected due to the nature of the work involved in the audit phase. 

The confidence intervals 95\% for each data set, illustrated in Figure~\ref{fig:blockchain}, provide ranges within which the true mean latencies are estimated to lie, reflecting the precision of our latency measurements.

Similarly, we conducted read latency characterization experiments (as shown in the "Read Latency from the Blockchain" Figure~\ref{fig:blockchain}). These experiments measured the read latency for various datasets, as described in Table~\ref{tab:blockchain_read}, representing different data retrieval scenarios from the blockchain.

The read latency experiments' ANOVA test resulted in a p-value of 0.3896, indicating that there were no significant differences in read latencies across the datasets. This outcome suggests that the observed variations in latency are likely attributable to random chance. Furthermore, the 95\% confidence intervals for these datasets, as detailed in Figure~\ref{fig:blockchain}, offer insight into the variability and reliability of our latency measurements.

Notably, the write latencies generally tend to be higher than the read latencies. This difference can be attributed to the fact that write operations involve the execution of consensus protocols with other nodes in the blockchain network, whereas read operations require access from only a single node in the network.

In essence, while this pilot study has provided valuable insight in characterizing the read and write latencies, future experiments with increased sample sizes and with different datasets can help deepen our understanding and validate our findings. 

\begin{table}[h!]
    \centering
    \begin{tabular}{c|p{4.8cm}}
        \toprule
        Dataset Notations & X Axis Description of Read Latency in Figure~\ref{fig:blockchain} \\
        \midrule
        Read QI and Metrics &  Time taken to read all quality indicators and associated metrics representing both desired and undesired properties of the blockchain. 
 \\ \hline 
        Read ranks &  Time taken to read ranks of all generators for all purposes from the blockchain. \\ \hline 
        Read rank  & Time taken to read the rank list of generators for a specific purpose from the blockchain.\\ \hline 
       Read WM\_plus & Time taken to read all rows in the $W^{M^{+}}$ matrix from the blockchain, where each row represents the purpose, and each column represents weights for metrics associated with desired properties.\\ \hline 
        Read WM\_minus & Time taken to read all rows in the $W^{M^{-}}$ matrix from the blockchain, where each row represents the purpose, and each column represents weights for metrics associated with undesired properties. \\ \hline 
        Read E matrix & Time taken to read the entire matrix E on the blockchain, with each row representing a generator and each column representing the results of the evaluated metric on the generator.\\ \hline 
        Read CW  & Time taken to read the entire QI weight matrix on the blockchain where each row represents the purpose and each column represents the weight associated with the QIs relevant to that purpose. \\ 
        
        \bottomrule
    \end{tabular}
    \caption{Notations and descriptions for the datasets read from the blockchain}. 
    \label{tab:blockchain_read}
\end{table}

\subsubsection{Implications and Observations}
The results of our experimental investigations validated the efficacy of the proposed framework in accurately ranking synthetic data generators according to the defined specifications for various purposes. This framework demonstrates exceptional adaptability to a variety of metric categorizations and can manage distinct configurations for each purpose. In addition, its flexibility extends to accommodating different transformation functions. 

Importantly, the framework integrates with a consortium blockchain technology, enhancing its capabilities in terms of accountability, auditability, and transparency. These features are particularly vital in contexts where the compliance regulations are very demanding. 

In summary, the proposed framework stands out for its multifaceted approach. Its versatility in handling different scenario-specific needs, responsiveness to metric weight distributions, and balanced consideration of desired and undesired properties make it a comprehensive tool for guiding the selection of synthetic data generators in sensitive and varied applications like healthcare.
\section{Discussions}
\label{sec:discussions}
In this section, we discuss the correctness, scalability, generalization, computational complexity, and limitations of the proposed framework. Furthermore, we discuss how the proposed framework mitigates risks from the security threats introduced in Section~\ref{sec:threats}.
\subsection{Correctness of the framework}

\begin{lemma}[Algorithm~\ref{alg:RankAcrossPurposes}]
\label{lemma_correctness:alg1}
For each purpose $p_i$, Algorithm~\ref{alg:RankAcrossPurposes} correctly ranks generators based on their performance in a set of metrics, considering the weights of desirable and undesirable properties and quality indicators. 
\end{lemma}

\begin{proof}
The algorithm initializes $E^+_{p_i}$ and $E^-_{p_i}$ matrices using the transformation Algorithm~\ref{alg:transformation} which correctly categorizes and ranks generators for each metric based on its nature and classification as desirable or undesirable property for $p_i$. The algorithm then applies Algorithm~\ref{alg:ranking} for each $p_i$, correctly calculating the overall score for each generator by summing the weighted ranks of $E^+_{p_i}$ and subtracting the weighted ranks of $E^-_{p_i}$. The algorithm sorts the generators based on these scores, ensuring that the highest score (indicating the best performance) gets the best possible rank. Since each step in the algorithm follows logically and is based on defined inputs and rules, the output ranking for each $p_i$ is correct according to the specified criteria. 
\end{proof}

\begin{lemma}[Algorithm~\ref{alg:transformation}]
\label{lemma_correctness:alg2}
Algorithm~\ref{alg:transformation} correctly transforms raw performance scores into ranked scores for both desirable and undesirable metrics, considering the nature of each metric. 
\end{lemma}

\begin{proof}
The algorithm classifies each metric into one of three categories: Lower is better, higher is better, or closer to a constant is better. It ranks each generator according to the nature of its corresponding metric: for "lower is better" metrics, lower scores receive higher ranks; for "higher is better" metrics, higher scores are ranked higher; and for "closer to a constant is better" metrics, scores nearest to the specified constant are ranked highest. This ranking approach ensures an accurate reflection of each generator's performance relative to the intended metric criteria. To standardize the ranking, the algorithm inverts the ranks, aligning all metrics on a scale where higher scores indicate superior performance. Then it populates $E^+_{p_i}$ and $E^-_{p_i}$ with these adjusted ranks, effectively representing the performance of the generators for each designated purpose $p_i$.
\end{proof}

\begin{lemma}[Algorithm~\ref{alg:ranking}]
\label{lemma_correctness:alg3}
Algorithm~\ref{alg:ranking} accurately ranks generators for a given purpose $p_i$ based on their performance on both desirable and undesirable metrics. It comprehensively considers the weights assigned to each metric, including both desired and undesired metrics weights, along with the weights of quality indicators for each metric. 
\end{lemma}

\begin{proof}
The algorithm calculates a score for each generator by summing the products of their performance ranks in desirable metrics multiplied by the corresponding weights and subtracting the products of their performance ranks in undesirable metrics multiplied by their respective weights. This method of score calculation ensures that the contribution of each metric is accurately reflected in the overall performance of a generator for $p_i$. Sorting the generators based on these calculated scores in descending order ensures that the best performing generator (with the highest score) is ranked top. This ranking process effectively identifies the best performing generator according to the criteria defined for $p_i$.
\end{proof}

\begin{lemma}[Algorithm~\ref{alg:DetailedAuditing}]
\label{lemma_correctness:alg4}
Algorithm~\ref{alg:DetailedAuditing} correctly verifies the consistency and accuracy of rankings produced by Algorithms~\ref{alg:RankAcrossPurposes},~\ref{alg:transformation}, and~\ref{alg:ranking}. 
\end{lemma}

\begin{proof}
The algorithm cross-checks the input data, the transformation process, and the final ranking against the records stored in the ledger, ensuring the integrity and consistency of the entire process. Then, replicates the process for each $p_i$, using the same input data and methods to ensure that the recorded rankings match the recalculated rankings. Any discrepancies found during this process indicate inconsistency, leading to a negative audit result. The thoroughness of the audit process and its reliance on blockchain technology for verification ensure the correctness of the ranking outcomes and the reliability of the process. 
\end{proof}

Given the lemmas established for Algorithms~\ref{alg:RankAcrossPurposes},~\ref{alg:transformation},~\ref{alg:ranking}, and~\ref{alg:DetailedAuditing}, we can formulate the following theorem for the entire framework:
\begin{theorem}[Correctness of the Framework]
\label{theorem:correctness_framework}
The framework comprising Algorithms~\ref{alg:RankAcrossPurposes},~\ref{alg:transformation},~\ref{alg:ranking}, and~\ref{alg:DetailedAuditing}, when implemented as smart contracts on a consortium blockchain, ensures an accurate, transparent, and tamper-proof process for  ranking synthetic data generators across multiple purposes. 
\end{theorem}

\begin{proof}
Lemmas~\ref{lemma_correctness:alg1},~\ref{lemma_correctness:alg2},~\ref{lemma_correctness:alg3},and~\ref{lemma_correctness:alg4} establish the correctness of Algorithms~\ref{alg:RankAcrossPurposes},~\ref{alg:transformation},~\ref{alg:ranking}, and~\ref{alg:DetailedAuditing} respectively. The integration of these algorithms into a coherent framework ensures that the entire process, from raw data to final rankings, is consistent, accurate, and verifiable. When all these algorithms are implemented as smart contracts, the execution of these algorithms is automated, ensuring that the logic defined in algorithms is followed precisely. Furthermore, every step in the algorithms is transparent and recorded on the blockchain, making it possible to verify the correctness of each step. Additionally, the blockchain's inherent security features protect against tampering and unauthorized alterations, ensuring the integrity of the ranking process.
\end{proof}

Although the framework leverages blockchain technology to ensure transparency, accuracy and tamper resistance in the ranking process, it is important to note that the evaluation of blockchain vulnerabilities, such as potential security flaws, scalability issues, or consensus problems, falls outside the scope of this research. This work assumes that the underlying blockchain infrastructure operates as intended without delving into its potential vulnerabilities or the broader security landscape that affects blockchain technologies.

Consequently, the assertions of Theorem 1 are contingent on the operational integrity and assumed security of the blockchain infrastructure, as described. This theorem presupposes that the blockchain framework functions flawlessly according to design specifications, and thus any deviations or vulnerabilities in the blockchain technology itself are not covered within the scope of this theorem's validation. 

\subsection{Computational Complexity of the Framework}
All the algorithms are executed as smart contracts in a blockchain network consisting of $n$ nodes. 
\begin{lemma}[Computational Complexity of Algorithm~\ref{alg:transformation}]
\label{lemma:algo1_complexity}
The computational complexity of the Transformation Algorithm~\ref{alg:transformation} is $O(T \log T \times |M_{p_i}|)$, where $T$ is the number of generators and $|M_{p_i}|$ is the total number of metrics evaluated for a given purpose $p_i$. The operational overhead scales with the number of nodes ($n$) in the consortium blockchain network, but this does not affect computational complexity.
\end{lemma}

\begin{proof}
We analyze the steps in Algorithm~\ref{alg:transformation} and their associated computational complexities:

\begin{enumerate}
    \item \textbf{Initialization of Matrices $E^{+}_{p_i}$ and $E^{-}_{p_i}$:}
    Each matrix has dimensions $|T| \times |M_{p_i}|$. Initializing these matrices involves setting up storage for $T$ generators on $|M_{p_i}|$ metrics for each matrix, leading to a complexity of $O(T \times |M_{p_i}|)$. Although this step is linear with respect to the number of elements, it provides a foundation for subsequent operations.

    \item \textbf{Ranking Based on Metric Nature:}
    For each metric $m_j$ in $M_{p_i}$, a list of $T$ elements is sorted, which has a complexity of $O(T \log T)$. This is the most computationally intensive step, as it is performed for each metric independently.

    \item \textbf{Inverse Ranking:}
    After sorting, each generator $t_k$ in the list of metric $m_j$ is assigned an inverse rank. Inverse ranking involves assigning ranks such that the best position gets the highest rank, and so forth. This computation, though linear, is critical, as it inversely scales each rank within the range from 1 to $T$. The complexity of this step is $O(T)$ for each metric, performed after sorting.

    \item \textbf{Populating $E^{+}_{p_i}$ and $E^{-}_{p_i}$:}
    This step involves iterating over each metric in $M_{p_i}$ and each generator $t_k$ to assign ranks from the inverse ranking in the matrices. Given that each assignment is a direct operation and occurs for each combination of generators $T$ and metrics $|M_{p_i}|$, this step also has a complexity of $O(T \times |M_{p_i}|)$.

    \item \textbf{Blockchain Overhead:}
    The involvement of a blockchain network primarily affects communication and replication between nodes and does not influence the computational complexity of the algorithm itself. The overhead related to consensus mechanisms and data replication depends on the number of nodes $n$ and is typically considered separate from computational complexity calculations.

\end{enumerate}

The dominant factor in the computational complexity of the algorithm is the sorting operation performed for each metric, $O(T \log T)$. Given that this operation is performed for each metric in the total set of metrics $|M_{p_i}|$, the total computational complexity is $O(T \log T \times |M_{p_i}|)$. Although initialization and population of matrices involve linear operations per metric, they do not exceed the complexity introduced by the sorting step. Therefore, the comprehensive computational complexity of the transformation Algorithm is confirmed as stated.

\end{proof}

\begin{lemma}[Computational Complexity of Algorithm~\ref{alg:ranking}]
\label{lemma:algo2_complexity}
The computational complexity of the Ranking Algorithm~\ref{alg:ranking} is $O(T \times |M_{p_i}| + T \log T)$, where $T$ is the number of generators and $|M_{p_i}|$ is the total number of metrics for the purpose $p_i$. This complexity includes the computation of scores and the final sorting of generators. The complexity scales with the number of nodes ($n$) in the consortium blockchain network for blockchain operations.
\end{lemma}

\begin{proof}
We analyze the steps in Algorithm~\ref{alg:ranking} and their associated computational complexities:

For each generator $t_k \in T$, the algorithm computes a score by iterating on all metrics in $M_{p_i}^+$ and $M_{p_i}^-$. The complexity of calculating the score for a generator is $O(|M_{p_i}|)$, and for all $T$ generators, it is $O(T \times |M_{p_i}|)$.
After computing scores, the algorithm sorts the list of $T$ generators based on their scores, which has a complexity of $O(T \log T)$.
The complexity of writing and retrieving data from the blockchain scales with the number of nodes $n$ in the network. However, this does not affect the complexity order for the ranking process.

Combining these complexities, the dominant factor is the combination of score computation and sorting. Thus, the total computational complexity of Algorithm~\ref{alg:ranking} is $O(T \times |M_{p_i}| + T \log T)$.
\end{proof}

\begin{lemma}[Computational Complexity of Algorithm~\ref{alg:RankAcrossPurposes}]
\label{lemma:algo3_complexity}
The computational complexity of Algorithm~\ref{alg:RankAcrossPurposes} is primarily determined by the complexities of the invoked subroutines - the Transformation Algorithm (Algorithm~\ref{alg:transformation}) and the Ranking Algorithm (Algorithm~\ref{alg:ranking}). The overall complexity is $O(P \times (T \log T \times |M_\text{total}| + T \log T))$, where $P$ is the number of purposes, $T$ is the number of generators, and $|M_\text{total}|$ is the total number of metrics in $P$ purposes. Furthermore, complexity scales with the number of nodes ($n$) in the blockchain network of the consortium due to blockchain operations.
\end{lemma}

\begin{proof}
The primary computational tasks in Algorithm~\ref{alg:RankAcrossPurposes} involve the following steps: First, iterating through each purpose in $P$. Second, invoking the Transformation Algorithm for each purpose, which has a complexity of $O(T \log T \times |M_{p_i}|)$ according to the Lemma~\ref{lemma:algo1_complexity}.
Third, invoking the ranking algorithm for each purpose, which has a complexity of $O(T \times |M_{p_i}| + T \log T)$ as per Lemma~\ref{lemma:algo2_complexity}.

Since the algorithm iterates for $P$ purposes, the total complexity is a product of the number of purposes and the complexities of the invoked subroutines. Thus, the overall complexity of Algorithm~\ref{alg:RankAcrossPurposes} is $O(P \times (T \log T \times |M_\text{total}| + T \log T))$.

The complexity of blockchain operations depends on the number of nodes $n$ in the consortium blockchain network. However, these operations do not change the overall order of complexity.
\end{proof}
\begin{lemma}[Computational Complexity of Algorithm~\ref{alg:DetailedAuditing}]
\label{lemma:algo4_complexity}
The computational complexity of the Auditing Process (Algorithm~\ref{alg:DetailedAuditing}) primarily depends on the verification of the ranking calculations and the blockchain operations. The complexity is $O(P \times (T \log T \times |M_\text{total}| + T \log T))$ for the ranking verification in $P$ purposes. Furthermore, complexity scales with the number of nodes ($n$) in the blockchain network of the consortium due to read and write operations of the blockchain.
\end{lemma}

\begin{proof}
The Auditing Process Algorithm involves the following steps:
First, verification of the specifications and evaluation matrices which is $O(1)$ for each purpose in $P$.
Afterward, recalculation and verification of the rankings for each purpose, which follows the same complexity as the ranking Algorithm, that is, $O(T \times |M_{p_i}| + T \log T)$.

Since the algorithm iterates for $P$ purposes, the overall complexity for the ranking verification is $O(P \times (T \log T \times |M_\text{total}| + T \log T))$.

Blockchain operations (both read and write) add complexity that scales with the number of nodes $n$ in the network. However, these operations do not change the overall order of complexity of the auditing process.
\end{proof}

\begin{theorem}[Overall Computational Complexity of the Framework]
\label{theorem:overall_complexity}
The overall computational complexity of the framework, encompassing Algorithms~\ref{alg:transformation}, \ref{alg:RankAcrossPurposes}, \ref{alg:ranking}, and \ref{alg:DetailedAuditing}, is given by $O(2 \times P \times (T \log T \times |M_{p_i}| + T \log T))$. This complexity scales with the number of purposes ($P$), the number of generators ($T$), and the total number of metrics for each purpose ($|M_{p_i}|$), and reflects the duplication of processes due to the auditing algorithm. The impact of blockchain operations, while significant for operational overhead, does not change the order of computational complexity.
\end{theorem}

\begin{proof}
From the lemmas pertaining to each algorithm:
\begin{itemize}
\item Algorithm~\ref{alg:transformation} has a complexity of $O(T \log T \times |M_{p_i}|)$ for each purpose.
\item  Algorithm~\ref{alg:RankAcrossPurposes} orchestrates the ranking process in all purposes, invoking Algorithms~\ref{alg:transformation} and \ref{alg:ranking} for each purpose. Its complexity is thus $O(P \times (T \log T \times |M_{p_i}| + T \log T))$.
\item Algorithm~\ref{alg:ranking} has a complexity subsumed under the complexity of Algorithm~\ref{alg:RankAcrossPurposes}.
\item Algorithm~\ref{alg:DetailedAuditing} repeats the verification ranking process, thus doubling the computational load of Algorithm~\ref{alg:RankAcrossPurposes}. Therefore, the computational complexity of the auditing process is similar to that of the ranking process.
\end{itemize}
Considering that the auditing process essentially duplicates the ranking computations, the overall complexity of the framework is effectively doubled. Thus, the overall computational complexity of the framework is $O(2 \times P \times (T \log T \times |M_{p_i}| + T \log T))$.

Although blockchain operations add operational overhead, they are primarily network and I/O bound and do not change the overall order of computational complexity of the framework.
\end{proof}

\subsection{Mitigating the risks due to security threats}
The framework uses a consortium blockchain to ensure the immutable recording of all critical actions, including the establishment of specification criteria, the submission of evaluation metrics, and the auditing results. Each user role (product manager, data scientist, and auditor) is assigned unique identity credentials. Furthermore, the permissions for read and write operations are tailored according to each role's specific requirements in Sawtooth. Every write operation within the system is designed to be tamper resistant and traceable back to the specific role that initiated these actions. 

To ensure security and privacy while accommodating the distinct roles of Product Managers, Data Scientists, and Auditors, we employ the concept of transactor key permission as outlined in Sawtooth~\cite{Sawtooth_key}. This method governs the submission of transactions and batches based on the signing keys. Upon receiving a batch from a client, the validator only processes batches if their batch signers have the authority to submit transactions. In our system, specific roles are delineated: Only the Product Manager is authorized to submit transactions to register quality indicators, weights for these indicators, and desired and undesired properties. Data scientists have the exclusive right to register methods, while auditors are the sole entities allowed to submit transactions for auditing the registered data. In contrast, transactions related to computing ranks and displaying registered data are permissible by any user, as accepted by the validator. 

The perimission model, coupled with the tamper-resistant capabilities of Sawtooth, mitigates the potential for repudiation attacks. This makes it difficult for any involved party, such as the product manager, data scientist, blockchain developer, or auditor, to deny their recorded actions. Additionally, all actions are replicated in all nodes in the blockchain network, which significantly reduces the risk of repudiation attacks within the framework.

In the realm of a consortium blockchain, the control and management of nodes are distributed among a variety of trusted entities or departments. This decentralized structure impedes the efforts of any minority group that might attempt to replace legitimate nodes with compromised ones. As a result, the framework enhances its resilience against colluding attacks, especially when implementing PBFT consensus protocols. Additionally, configuring the network with the appropriate number of nodes ensures tolerance among the anticipated minority exhibiting arbitrary behaviors. 

Moreover, the consortium blockchain inherently offers an immutable audit trail. This feature is instrumental in ensuring that any data entered into the blockchain is permanently recorded and traceable back to its source. Such traceability is crucial in identifying the origins of data, enabling quick detection and resolution of any attempts at data poisoning. By incorporating these measures, the framework effectively reduces the risks associated with data tampering and ensures a secure and trustworthy evaluation environment.

\subsection{Scalability Analysis}
Our scalability analysis begins by acknowledging the foundational performance benchmarks identified in the relevant literature regarding the Sawtooth framework~\cite{sawtoothBenchmark}. These benchmarks highlight two critical performance thresholds: a linear increase in throughput with input transaction rates up to approximately 1000 transactions per second (tx/sec), and the impact of batch size on throughput, demonstrating a nearly linear relationship until about 2300 tx/sec. Beyond these points, performance degradation due to queue timeouts and batch processing failures becomes pronounced. This empirical understanding provides an essential backdrop for our analysis, especially considering the additional computational complexity introduced by the ranking algorithm, notably $O(2 \times P \times (T \log T \times |M_{p_i}| + T \log T))$. This complexity scales with the number of purposes ($P$), the number of generators ($T$), and the total number of metrics for each purpose ($|M_{p_i}|$).  Duplication of processes due to the auditing algorithm also doubles the computational load, potentially impacting throughput and latency.

Our experimental setup mirrors a conventional use case within a single organization. Here, a product manager issues a single transaction for each command, as detailed in Table~\ref{tab:commands}. This transaction, addressing multiple purposes, is then consolidated into a single batch. Similarly, when a data scientist issues a transaction corresponding to a client command, also referenced in Table~\ref{tab:commands}, it encapsulates several generators and metric values and is likewise organized into an individual batch. As such, the transaction volume per use case remains inherently constrained, illustrating the framework's operational efficiency and the streamlined processing capacity within a typical organizational context.

However, extending the proposed framework as a software-as-a-service (SaaS) solution to a multitude of companies would markedly increase the transaction volume, thus elevating the importance and urgency of a comprehensive scalability analysis. Such an analysis, while beyond the scope of our current investigation, becomes a critical area for future research.

\subsection{Social and Ethical Implications}
Several studies have been conducted to understand the impact of blockchain technology on social and ethical issues~\cite{rahimzadeh2018ethics, dierksmeier2020blockchain, hyrynsalmi2020blockchain, tang2020ethics,  haque2021gdpr, upadhyay2021blockchain, sharif2022ethics, agerskov2023ethical}, highlighting issues such as energy efficiency, the balance between transparency and privacy, and the incentives for organizations to adopt permissioned blockchains. Our proposed framework adopts the consensus protocol PBFT, offering an energy-efficient alternative to PoW mechanisms~\cite{hussein2023evolution}. Additionally, our methodology for ranking synthetic data generators does not inadvertently reveal any sensitive information from the personal data used in training these generators. 

A recent comprehensive literature study~\cite{VALLEVIK2024105413} on the evaluation of synthetic data generators revealed a notable lack of fairness and carbon footprint metrics. This finding highlights the imperative for a collective effort to integrate these metrics into evaluations, fostering a move towards more ethical and sustainable Artificial Intelligence development. This integration is pivotal not only for ethical advancement but also for ensuring compliance with forthcoming regulations such as the Artificial Intelligence Act~\cite{act2021proposal}, thus enhancing practical applicability. Given that our proposed framework meticulously records and audits all criteria for ranking synthetic data generators, it will transparently reflect the organization's prioritization of fairness and carbon footprint metrics in the ranking process.

Our framework significantly improves organizational compliance and auditability, while also improving transparency and trust between stakeholders in the use of synthetic data. By evaluating both desired and undesired properties and their relevance for specific purposes, our ranking algorithm ensures fairness, provided that the inputs (see Table~\ref{tab:dataTypes}) recorded on the blockchain remain unbiased and not subject to any security attacks. This approach is crucial in domains such as healthcare, finance and public services, where decisions based on synthetic data could have significant social impacts.

\subsection{Limitations of the framework}
The potential limitations of our framework are listed as follows:

Our framework is heavily dependent on the permissioned blockchain infrastructure. Hence, the framework's functionality is intrinsically linked to the stability and security of the permissioned blockchain infrastructure. Any vulnerability in this infrastructure could adversely affect the framework. The read and write latency of blockchain operations is influenced by the number of nodes within the blockchain and the network latency. Hence, the actual performance of the framework may vary depending on the blockchain's efficiency, particularly in configurations with numerous nodes. 

Moreover, one of the core criticisms of permissioned blockchains compared to public blockchain is that they reintroduce centralization into a technology that is fundamentally designed to be decentralized. By restricting who can participate in the network, permissioned blockchains concentrate power in the hands of a few, which can lead to concerns about the abuse of authority.

Furthermore, our experiments aimed at a conventional use case within a single organization. Accordingly, the transaction volume is inherently constrained, hence the throughput (number of transactions processed per second) limit is the same as the default limit offered by the Sawtooth. The proposed framework has no special scalability techniques implemented to improve the throughput of the Sawtooth network. 

Additionally, in our experiments, given the absence of real-world data for weight distributions and categorizations of desired and undesired properties, we needed to do sensitive analysis and create hypothetical scenarios/purposes to analyze the characteristics of the proposed ranking algorithm. Furthermore, even though the proposed framework can work for a broad range of synthetic data types, the metric values and the chosen generator type in our experiments are meant only for the synthetic tabular health data type. 

In scenarios requiring human validation (such as in healthcare for synthetic data quality assessment), our framework considers the outcome of this validation as another quantifiable metric. In situations where multiple human validations are necessary, each instance is treated as a distinct measurable metric. 
\section{Related Work}
\label{sec:relatedWork}

To assess and compare the performance of different data generators, the recommended practice is to use multiple quality measures to comprehensively address the multifaceted nature of data quality~\cite{jordon2018measuring, jordon2018pate, yoon2019time, goncalves2020generation, van2021decaf, alaa2022faithful, dankar2022multi, liu2022goggle, norcliffe2023survivalgan}.
However, manually assessing a wide range of quality measures is labor intensive and impractical. Consequently, in the literature, several evaluation frameworks have been developed to allow for the ranking and comparison of generators~\cite{arnold2020really,chundawat2022tabsyndex,yan2022multifaceted, hernadez2023synthetic,pathare2023comparison}.
 
Current research in the ranking of synthetic data generators has focused predominantly on identifying and preserving desired properties in the synthetic data generated. For example, \cite{chundawat2022tabsyndex} introduced a universal indexing metric that averages five metrics focusing on the utility aspects of the data. However, this metric overlooks contextual aspects and concentrates solely on a single quality indicator (data utility). ~\cite{pathare2023comparison} compared various synthetic data generators using only similarity metrics such as propensity score and log cluster metrics. 

Moreover, existing studies~\cite{arnold2020really,chundawat2022tabsyndex,yan2022multifaceted, hernadez2023synthetic,pathare2023comparison} have not explicitly accounted for the evaluation of undesirable properties while ranking synthetic data generators. This omission means there is no assurance of adhering to principles such as purpose limitation (ensuring data is used only for intended purposes), data minimization, and data protection by design and default in the synthetic data. Consequently, the final rankings could mislead decision makers in choosing a synthetic data generator that may not be comprehensive in its evaluation, potentially overlooking crucial aspects of data security and compliance. 

In situations where data must be aggregated from multiple stakeholders, ensuring accountability, auditability, and transparency is essential to guarantee the reliability and trustworthiness of the framework~\cite{veeraragavan2023securing}. Stakeholders, for example, must be responsible for submitting accurate and reliable data. Conducting audits during the aggregation process can help identify potential issues or inconsistencies in the final results. Ultimately, transparent reporting is crucial to ensure a clear understanding of the process and instill confidence in the reliability of the results. 

However, previous methodologies, as discussed in~\cite{arnold2020really,chundawat2022tabsyndex,yan2022multifaceted, hernadez2023synthetic,pathare2023comparison} exhibit limitations in terms of transparency, accountability, and auditability. To address these shortcomings, our proposed framework leverages blockchain technology and smart contracts, advancing the state-of-the-art in ranking synthetic data generators. Furthermore, the proposed framework also offers a strong protection against a threat model involving repudiation, colluding, and poisoning attacks.

\subsection{Application data management using blockchain}
\label{sec:data_management}

Blockchain technology has found extensive application across various domains, as evidenced by numerous survey papers. This section will highlight those applications without delving into each individual work.

The most notable application domain of blockchain is in cryptocurrencies, with surveys by~\cite{bonneau2015sok} and~\cite{tschorsch2016bitcoin} examining Bitcoin and other cryptocurrencies. Blockchain enhances digital financial assets with increased security, transparency, and traceability. Studies by~\cite{ali2018applications} and~\cite{lao2020survey} have focused on the IoT network domain, where blockchain facilitates business entities accessing and providing IoT data without centralized control, thus improving security, transaction management, and data verification. Research in the healthcare sector by~\cite{de2020survey} and~\cite{arbabi2022survey} shows how blockchain and smart contracts enable the sharing of secure health data, benefiting from decentralization, trustlessness, immutability, traceability and transparency of the blockchain.

The surveys by~\cite{xie2019survey} and~\cite{ullah2023blockchain} offer a comprehensive review of blockchain technology in smart cities, improving services through the unique attributes of the blockchain. Integration with cloud applications is explored in works by~\cite{yang2019integrated}~\cite{gai2020blockchain},~\cite{nguyen2020integration},and~\cite{sharma2020blockchain},  highlighting the synergy between blockchain and cloud-based systems to improve functionality, performance and security. Lastly,~\cite{eberhardt2017or} and~\cite{vo2018research} discuss the role of the blockchain in data management and analysis, emphasizing enhanced information protection and smart contract management.

Positioning our work within this broader context, our research work aims to advance the evaluation of synthetic data generators through a comprehensive and nuanced approach. Our methodology includes: a) context-sensitive assessments tailored to the nuances of specific scenarios, b) robust accountability mechanisms, c) thorough auditability processes, and d) a commitment to enhanced transparency.  Through this, we contribute not only to the synthetic data generation (Generative AI) field but also to the larger blockchain ecosystem, particularly in data management practices.

\section{Future Research Direction}\label{sec:futureWork}
To improve the applicability of our proposed framework and address future challenges, several directions are worth exploring. 
\begin{itemize}
\item \textbf{Adapting to a SaaS model:} Transitioning the framework to a SaaS model could significantly extend its reach across various organizations. This transition will require a focus on scalability and security, particularly in handling high transaction volumes from diverse groups of clients. 
 
\item \textbf{Domain-Specific Social and Ethical Impacts:} Investigating the social and ethical implications within specific domains is crucial for practical adoption. A detailed analysis of fairness and trustworthiness of the proposed framework in different application areas will help identify unique concerns and opportunities for improvements. 

\item \textbf{Permissioned Blockchain Framework Comparisons:} While the current implementation utilizes Sawtooth, future research should extend to other permissioned blockchain frameworks like Hyperledger Fabric~\cite{androulaki2018hyperledger} and Besu~\cite{soelman2021permissioned}. Comparing these frameworks in supporting the proposed ranking methodology can illuminate differences in performance, security, and suitability for various use cases.

\item \textbf{Interoperability Challenges:} Exploring interoperability among various permissioned blockchain frameworks can facilitate collaboration within a consortium of organizations in supporting the proposed ranking methodology, each utilizing different blockchain technologies. This could lead to standardized practices for cross-framework interaction. 

\item \textbf{Experimental Expansion:} Broadening experimental efforts to include a wider range of data and metric types, enhancing the robustness and applicability of the proposed ranking algorithm across different contexts and datasets such as images, videos, and audio.

\item \textbf{Economic models for Sustainability:} Developing economic models that integrate synthetic data generation and blockchain technology can promote sustainable practices within these fields. Researching economic incentives, market structures, and regulatory frameworks supportive of sustainable blockchain and synthetic data solutions can provide valuable insights for multiple industries. 

\item \textbf{Integration with existing Tools:} Integrating with tools like Synthcity~\cite{qian2023synthcity} and testing a wide range of generators across datasets will validate the framework's capabilities, highlighting its potential for widespread applicability and adaptability in the rapidly evolving domain of synthetic data generation.  

\item \textbf{Security and Privacy Analysis:} Exploring threat models, alongside implementing comprehensive security and privacy measures within the input generation process, is pivotal for enhancing the integrity and impartiality of inputs provided by product managers and data scientists in the ranking of synthetic data generators. 

\item \textbf{Standards and Benchmarks for Ranking Synthetic Data Generators:} The community dedicated to ranking synthetic data generators is poised to investigate standardized datasets and benchmarks. This will include defining weight distributions tailored for various scenarios (or purposes) and classifying metrics into desired and undesired attributes for each specific purpose. 
\end{itemize}

\section{Conclusion}
\label{sec:conlcusion}

In conclusion, our framework's distinct advantage lies in its comprehensive and adaptable methodology, which has a potential to consider wide array of quality indicators and metrics, both positive (desired properties) and negative (undesired properties) attributes across various scenarios. By integrating blockchain technology, we ensure robustness across security threats and offer transparency, accountability, and auditability. The experiments conducted validate the effectiveness of our framework in its ability to rank synthetic data generators in different context-specific requirements. This framework provides a tool that not only assists decision makers in selecting the most appropriate synthetic data generators, but also upholds the principles of compliance.

\vskip 0.2in

\printbibliography[title={References}]

@article{arnold2020really,
  title={Really Useful Synthetic Data--A Framework to Evaluate the Quality of Differentially Private Synthetic Data},
  author={Arnold, Christian and Neunhoeffer, Marcel},
  journal={arXiv preprint arXiv:2004.07740},
  year={2020}
}

@article{chundawat2022tabsyndex,
  title={TabSynDex: A Universal Metric for Robust Evaluation of Synthetic Tabular Data},
  author={Chundawat, Vikram S and Tarun, Ayush K and Mandal, Murari and Lahoti, Mukund and Narang, Pratik},
  journal={arXiv preprint arXiv:2207.05295},
  year={2022}
}

@article{goncalves2020generation,
  title={Generation and evaluation of synthetic patient data},
  author={Goncalves, Andre and Ray, Priyadip and Soper, Braden and Stevens, Jennifer and Coyle, Linda and Sales, Ana Paula},
  journal={BMC medical research methodology},
  volume={20},
  number={1},
  pages={1--40},
  year={2020},
  publisher={BioMed Central}
}

@inproceedings{norcliffe2023survivalgan,
  title={SurvivalGAN: Generating Time-to-Event Data for Survival Analysis},
  author={Norcliffe, Alexander and Cebere, Bogdan and Imrie, Fergus and Lio, Pietro and van der Schaar, Mihaela},
  booktitle={International Conference on Artificial Intelligence and Statistics},
  pages={10279--10304},
  year={2023},
  organization={PMLR}
}

@inproceedings{liu2022goggle,
  title={GOGGLE: Generative modelling for tabular data by learning relational structure},
  author={Liu, Tennison and Qian, Zhaozhi and Berrevoets, Jeroen and van der Schaar, Mihaela},
  booktitle={The Eleventh International Conference on Learning Representations},
  year={2022}
}

@article{yoon2019time,
  title={Time-series generative adversarial networks},
  author={Yoon, Jinsung and Jarrett, Daniel and Van der Schaar, Mihaela},
  journal={Advances in neural information processing systems},
  volume={32},
  year={2019}
}

@inproceedings{veeraragavan2023securing,
  title={Securing Federated GANs: Enabling Synthetic Data Generation for Health Registry Consortiums},
  author={Veeraragavan, Narasimha Raghavan and Nyg{\aa}rd, Jan Franz},
  booktitle={Proceedings of the 18th International Conference on Availability, Reliability and Security},
  pages={1--9},
  year={2023}
}

@article{dankar2022multi,
  title={A multi-dimensional evaluation of synthetic data generators},
  author={Dankar, Fida K and Ibrahim, Mahmoud K and Ismail, Leila},
  journal={IEEE Access},
  volume={10},
  pages={11147--11158},
  year={2022},
  publisher={IEEE}
}

@article{yan2022multifaceted,
  title={A multifaceted benchmarking of synthetic electronic health record generation models},
  author={Yan, Chao and Yan, Yao and Wan, Zhiyu and Zhang, Ziqi and Omberg, Larsson and Guinney, Justin and Mooney, Sean D and Malin, Bradley A},
  journal={Nature communications},
  volume={13},
  number={1},
  pages={7609},
  year={2022},
  publisher={Nature Publishing Group UK London}
}

@inproceedings{alaa2022faithful,
  title={How faithful is your synthetic data? sample-level metrics for evaluating and auditing generative models},
  author={Alaa, Ahmed and Van Breugel, Boris and Saveliev, Evgeny S and van der Schaar, Mihaela},
  booktitle={International Conference on Machine Learning},
  pages={290--306},
  year={2022},
  organization={PMLR}
}

@article{jordon2018measuring,
  title={Measuring the quality of synthetic data for use in competitions},
  author={Jordon, James and Yoon, Jinsung and van der Schaar, Mihaela},
  journal={arXiv preprint arXiv:1806.11345},
  year={2018}
}

@inproceedings{jordon2018pate,
  title={PATE-GAN: Generating synthetic data with differential privacy guarantees},
  author={Jordon, James and Yoon, Jinsung and Van Der Schaar, Mihaela},
  booktitle={International conference on learning representations},
  year={2018}
}

@article{van2021decaf,
  title={Decaf: Generating fair synthetic data using causally-aware generative networks},
  author={van Breugel, Boris and Kyono, Trent and Berrevoets, Jeroen and van der Schaar, Mihaela},
  journal={Advances in Neural Information Processing Systems},
  volume={34},
  pages={22221--22233},
  year={2021}
}

@article{tabatabaei2023understanding,
  title={Understanding blockchain: Definitions, architecture, design, and system comparison},
  author={Tabatabaei, Mohammad Hossein and Vitenberg, Roman and Veeraragavan, Narasimha Raghavan},
  journal={Computer Science Review},
  volume={50},
  pages={100575},
  year={2023},
  publisher={Elsevier}
}

@article{gonzales2023,
  title={Synthetic data in health care: a narrative review},
  author={Gonzales, Aldren and Guruswamy, Guruprabha and Smith, Scott R},
  journal={PLOS Digital Health},
  volume={2},
  number={1},
  pages={e0000082},
  year={2023},
  publisher={Public Library of Science San Francisco, CA USA}
}

@electronic{GDPR,
  author    = {European Commission},
  title     = {Data Protection in the EU},
  url       = {https://commission.europa.eu/data-protection-eu},
  urldate   = {2023-12-06}
}

@article{hernandez2022synthetic,
  title={Synthetic data generation for tabular health records: A systematic review},
  author={Hernandez, Mikel and Epelde, Gorka and Alberdi, Ane and Cilla, Rodrigo and Rankin, Debbie},
  journal={Neurocomputing},
  volume={493},
  pages={28--45},
  year={2022},
  publisher={Elsevier}
}

@article{murtaza2023synthetic,
  title={Synthetic data generation: State of the art in health care domain},
  author={Murtaza, Hajra and Ahmed, Musharif and Khan, Naurin Farooq and Murtaza, Ghulam and Zafar, Saad and Bano, Ambreen},
  journal={Computer Science Review},
  volume={48},
  pages={100546},
  year={2023},
  publisher={Elsevier}
}

@article{fonseca2023tabular,
  title={Tabular and latent space synthetic data generation: a literature review},
  author={Fonseca, Joao and Bacao, Fernando},
  journal={Journal of Big Data},
  volume={10},
  number={1},
  pages={115},
  year={2023},
  publisher={Springer}
}

@misc{Sawtooth,
    title = "Sawtooth",
    url={https://sawtooth.splinter.dev/}, 
    urldate={2024-02-07} 
}

@misc{Sawtooth_key,
    title = "Permissioning Design",
    url={https://sawtooth.splinter.dev/docs/1.2/architecture/permissioning_requirement.html}, 
    urldate={2024-02-15} 
}

@article{bakos2021permissioned,
  title={Permissioned vs Permissionless Blockchain Platforms: Tradeoffs in Trust and Performance},
  author={Bakos, Yannis and Halaburda, Hanna},
  journal={NYU Stern School of Business working paper},
  year={2021}
}

@article{pathare2023comparison,
  title={Comparison of tabular synthetic data generation techniques using propensity and cluster log metric},
  author={Pathare, Aryan and Mangrulkar, Ramchandra and Suvarna, Kartik and Parekh, Aryan and Thakur, Govind and Gawade, Aruna},
  journal={International Journal of Information Management Data Insights},
  volume={3},
  number={2},
  pages={100177},
  year={2023},
  publisher={Elsevier}
}

@article{hernadez2023synthetic,
  title={Synthetic tabular data evaluation in the health domain covering resemblance, utility, and privacy dimensions},
  author={Hernadez, Mikel and Epelde, Gorka and Alberdi, Ane and Cilla, Rodrigo and Rankin, Debbie},
  journal={Methods of Information in Medicine},
  year={2023},
  publisher={Georg Thieme Verlag KG R{\"u}digerstra{\ss}e 14, 70469 Stuttgart, Germany}
}

@article{qian2023synthcity,
  title={Synthcity: facilitating innovative use cases of synthetic data in different data modalities},
  author={Qian, Zhaozhi and Cebere, Bogdan-Constantin and van der Schaar, Mihaela},
  journal={arXiv preprint arXiv:2301.07573},
  year={2023}
}

@inproceedings{castro1999practical,
  title={Practical byzantine fault tolerance},
  author={Castro, Miguel and Liskov, Barbara and others},
  booktitle={OsDI},
  volume={99},
  number={1999},
  pages={173--186},
  year={1999}
}

@INPROCEEDINGS{sawtoothBenchmark,
  author={Ampel, Benjamin and Patton, Mark and Chen, Hsinchun},
  booktitle={2019 IEEE International Conference on Intelligence and Security Informatics (ISI)}, 
  title={Performance Modeling of Hyperledger Sawtooth Blockchain}, 
  year={2019},
  volume={},
  number={},
  pages={59-61},
  }

@inproceedings{bonneau2015sok,
  title={Sok: Research perspectives and challenges for bitcoin and cryptocurrencies},
  author={Bonneau, Joseph and Miller, Andrew and Clark, Jeremy and Narayanan, Arvind and Kroll, Joshua A and Felten, Edward W},
  booktitle={2015 IEEE symposium on security and privacy},
  pages={104--121},
  year={2015},
  organization={IEEE}
}

@article{tschorsch2016bitcoin,
  title={Bitcoin and beyond: A technical survey on decentralized digital currencies},
  author={Tschorsch, Florian and Scheuermann, Bj{\"o}rn},
  journal={IEEE Communications Surveys \& Tutorials},
  volume={18},
  number={3},
  pages={2084--2123},
  year={2016},
  publisher={IEEE}
}

@article{arbabi2022survey,
  title={A survey on blockchain for healthcare: Challenges, benefits, and future directions},
  author={Arbabi, Mohammad Salar and Lal, Chhagan and Veeraragavan, Narasimha Raghavan and Marijan, Dusica and Nyg{\aa}rd, Jan F and Vitenberg, Roman},
  journal={IEEE communications surveys \& tutorials},
  volume={25},
  number={1},
  pages={386--424},
  year={2022},
  publisher={IEEE}
}

@article{lao2020survey,
  title={A survey of IoT applications in blockchain systems: Architecture, consensus, and traffic modeling},
  author={Lao, Laphou and Li, Zecheng and Hou, Songlin and Xiao, Bin and Guo, Songtao and Yang, Yuanyuan},
  journal={ACM Computing Surveys (CSUR)},
  volume={53},
  number={1},
  pages={1--32},
  year={2020},
  publisher={ACM New York, NY, USA}
}

@article{ali2018applications,
  title={Applications of blockchains in the Internet of Things: A comprehensive survey},
  author={Ali, Muhammad Salek and Vecchio, Massimo and Pincheira, Miguel and Dolui, Koustabh and Antonelli, Fabio and Rehmani, Mubashir Husain},
  journal={IEEE Communications Surveys \& Tutorials},
  volume={21},
  number={2},
  pages={1676--1717},
  year={2018},
  publisher={IEEE}
}

@article{de2020survey,
  title={A survey of blockchain-based strategies for healthcare},
  author={De Aguiar, Erikson J{\'u}lio and Fai{\c{c}}al, Bruno S and Krishnamachari, Bhaskar and Ueyama, J{\'o}},
  journal={ACM Computing Surveys (CSUR)},
  volume={53},
  number={2},
  pages={1--27},
  year={2020},
  publisher={ACM New York, NY, USA}
}

@article{xie2019survey,
  title={A survey of blockchain technology applied to smart cities: Research issues and challenges},
  author={Xie, Junfeng and Tang, Helen and Huang, Tao and Yu, F Richard and Xie, Renchao and Liu, Jiang and Liu, Yunjie},
  journal={IEEE communications surveys \& tutorials},
  volume={21},
  number={3},
  pages={2794--2830},
  year={2019},
  publisher={IEEE}
}

@article{ullah2023blockchain,
  title={Blockchain applications in sustainable smart cities},
  author={Ullah, Zaib and Naeem, Muddasar and Coronato, Antonio and Ribino, Patrizia and De Pietro, Giuseppe},
  journal={Sustainable Cities and Society},
  pages={104697},
  year={2023},
  publisher={Elsevier}
}

@article{gai2020blockchain,
  title={Blockchain meets cloud computing: A survey},
  author={Gai, Keke and Guo, Jinnan and Zhu, Liehuang and Yu, Shui},
  journal={IEEE Communications Surveys \& Tutorials},
  volume={22},
  number={3},
  pages={2009--2030},
  year={2020},
  publisher={IEEE}
}

@article{nguyen2020integration,
  title={Integration of blockchain and cloud of things: Architecture, applications and challenges},
  author={Nguyen, Dinh C and Pathirana, Pubudu N and Ding, Ming and Seneviratne, Aruna},
  journal={IEEE Communications surveys \& tutorials},
  volume={22},
  number={4},
  pages={2521--2549},
  year={2020},
  publisher={IEEE}
}

@article{sharma2020blockchain,
  title={Blockchain technology for cloud storage: A systematic literature review},
  author={Sharma, Pratima and Jindal, Rajni and Borah, Malaya Dutta},
  journal={ACM Computing Surveys (CSUR)},
  volume={53},
  number={4},
  pages={1--32},
  year={2020},
  publisher={ACM New York, NY, USA}
}

@article{yang2019integrated,
  title={Integrated blockchain and edge computing systems: A survey, some research issues and challenges},
  author={Yang, Ruizhe and Yu, F Richard and Si, Pengbo and Yang, Zhaoxin and Zhang, Yanhua},
  journal={IEEE Communications Surveys \& Tutorials},
  volume={21},
  number={2},
  pages={1508--1532},
  year={2019},
  publisher={IEEE}
}

@inproceedings{vo2018research,
  title={Research Directions in Blockchain Data Management and Analytics.},
  author={Vo, Hoang Tam and Kundu, Ashish and Mohania, Mukesh K},
  booktitle={EDBT},
  pages={445--448},
  year={2018}
}

@inproceedings{eberhardt2017or,
  title={On or off the blockchain? Insights on off-chaining computation and data},
  author={Eberhardt, Jacob and Tai, Stefan},
  booktitle={Service-Oriented and Cloud Computing: 6th IFIP WG 2.14 European Conference, ESOCC 2017, Oslo, Norway, September 27-29, 2017, Proceedings 6},
  pages={3--15},
  year={2017},
  organization={Springer}
}

@article{tang2020ethics,
  title={Ethics of blockchain: A framework of technology, applications, impacts, and research directions},
  author={Tang, Yong and Xiong, Jason and Becerril-Arreola, Rafael and Iyer, Lakshmi},
  journal={Information Technology \& People},
  volume={33},
  number={2},
  pages={602--632},
  year={2020},
  publisher={Emerald Publishing Limited}
}

@article{upadhyay2021blockchain,
  title={Blockchain technology and the circular economy: Implications for sustainability and social responsibility},
  author={Upadhyay, Arvind and Mukhuty, Sumona and Kumar, Vikas and Kazancoglu, Yigit},
  journal={Journal of cleaner production},
  volume={293},
  pages={126130},
  year={2021},
  publisher={Elsevier}
}

@article{rahimzadeh2018ethics,
  title={Ethics Governance Outside the Box: Reimagining Blockchain as a Policy Tool to Facilitate Single Ethics Review and Data Sharing for the'omics' Sciences},
  author={Rahimzadeh, Vaso},
  journal={Blockchain in Healthcare Today},
  year={2018}
}

@article{sharif2022ethics,
  title={The ethics of blockchain in organizations},
  author={Sharif, Monica M and Ghodoosi, Farshad},
  journal={Journal of Business Ethics},
  volume={178},
  number={4},
  pages={1009--1025},
  year={2022},
  publisher={Springer}
}

@article{haque2021gdpr,
  title={GDPR compliant blockchains--a systematic literature review},
  author={Haque, AKM Bahalul and Islam, AKM Najmul and Hyrynsalmi, Sami and Naqvi, Bilal and Smolander, Kari},
  journal={IEEE Access},
  volume={9},
  pages={50593--50606},
  year={2021},
  publisher={IEEE}
}

@inproceedings{hyrynsalmi2020blockchain,
  title={Blockchain ethics: A systematic literature review of blockchain research},
  author={Hyrynsalmi, Sami and Hyrynsalmi, Sonja M and Kimppa, Kai K},
  booktitle={Well-Being in the Information Society. Fruits of Respect: 8th International Conference, WIS 2020, Turku, Finland, August 26--27, 2020, Proceedings 8},
  pages={145--155},
  year={2020},
  organization={Springer}
}

@article{dierksmeier2020blockchain,
  title={Blockchain and business ethics},
  author={Dierksmeier, Claus and Seele, Peter},
  journal={Business Ethics: A European Review},
  volume={29},
  number={2},
  pages={348--359},
  year={2020},
  publisher={Wiley Online Library}
}

@inproceedings{agerskov2023ethical,
    author = {Agerskov, Signe and Pedersen, Asger Balle and Beck, Roman},
    title = {Ethical Guidelines for Blockchain Systems},
    booktitle = {ECIS 2023 Research Papers},
    pages={275},    
    year = {2023}
}

@article{hussein2023evolution,
  title={Evolution of blockchain consensus algorithms: a review on the latest milestones of blockchain consensus algorithms},
  author={Hussein, Ziad and Salama, May A and El-Rahman, Sahar A},
  journal={Cybersecurity},
  volume={6},
  number={1},
  pages={30},
  year={2023},
  publisher={Springer}
}

@inproceedings{androulaki2018hyperledger,
  title={Hyperledger fabric: a distributed operating system for permissioned blockchains},
  author={Androulaki, Elli and Barger, Artem and Bortnikov, Vita and Cachin, Christian and Christidis, Konstantinos and De Caro, Angelo and Enyeart, David and Ferris, Christopher and Laventman, Gennady and Manevich, Yacov and others},
  booktitle={Proceedings of the thirteenth EuroSys conference},
  pages={1--15},
  year={2018}
}

@phdthesis{soelman2021permissioned,
  title={Permissioned Blockchains: A Comparative Study},
  author={Soelman, Mark},
  year={2021}
}

@article{VALLEVIK2024105413,
title = {Can I trust my fake data – A comprehensive quality assessment framework for synthetic tabular data in healthcare},
journal = {International Journal of Medical Informatics},
volume = {185},
year = {2024},
issn = {1386-5056},
author = {Vibeke Binz Vallevik and Aleksandar Babic and Serena E. Marshall and Severin Elvatun and Helga M.B. Brøgger and Sharmini Alagaratnam and Bjørn Edwin and Narasimha R. Veeraragavan and Anne Kjersti Befring and Jan F. Nygård},
}

@article{act2021proposal,
  title={Proposal for a regulation of the European Parliament and the Council laying down harmonised rules on Artificial Intelligence (Artificial Intelligence Act) and amending certain Union legislative acts},
  author={Act, Artificial Intelligence},
  journal={EUR-Lex-52021PC0206},
  year={2021}
}

@misc{seer2023,
  author = {{National Cancer Institute}},
  title = {Surveillance, Epidemiology, and End Results (SEER) Program},
  year = {2023},
  url = {https://seer.cancer.gov/data/},
  note = {Accessed: 2024-03-18}
}

@article{polge2021permissioned,
  title={Permissioned blockchain frameworks in the industry: A comparison},
  author={Polge, Julien and Robert, J{\'e}r{\'e}my and Le Traon, Yves},
  journal={Ict Express},
  volume={7},
  number={2},
  pages={229--233},
  year={2021},
  publisher={Elsevier}
}

@article{capocasale2023comparative,
  title={Comparative analysis of permissioned blockchain frameworks for industrial applications},
  author={Capocasale, Vittorio and Gotta, Danilo and Perboli, Guido},
  journal={Blockchain: Research and Applications},
  volume={4},
  number={1},
  pages={100113},
  year={2023},
  publisher={Elsevier}
}

@inproceedings{min2016permissioned,
  title={A permissioned blockchain framework for supporting instant transaction and dynamic block size},
  author={Min, Xinping and Li, Qingzhong and Liu, Lei and Cui, Lizhen},
  booktitle={2016 IEEE Trustcom/BigDataSE/ISPA},
  pages={90--96},
  year={2016},
  organization={IEEE}
}

@inproceedings{monrat2020performance,
  title={Performance evaluation of permissioned blockchain platforms},
  author={Monrat, Ahmed Afif and Schel{\'e}n, Olov and Andersson, Karl},
  booktitle={2020 IEEE Asia-Pacific Conference on Computer Science and Data Engineering (CSDE)},
  pages={1--8},
  year={2020},
  organization={IEEE}
}

@inproceedings{qi2021bidl,
  title={Bidl: A high-throughput, low-latency permissioned blockchain framework for datacenter networks},
  author={Qi, Ji and Chen, Xusheng and Jiang, Yunpeng and Jiang, Jianyu and Shen, Tianxiang and Zhao, Shixiong and Wang, Sen and Zhang, Gong and Chen, Li and Au, Man Ho and others},
  booktitle={Proceedings of the ACM SIGOPS 28th Symposium on Operating Systems Principles},
  pages={18--34},
  year={2021}
}

@article{peng2022neuchain,
  title={Neuchain: a fast permissioned blockchain system with deterministic ordering},
  author={Peng, Zeshun and Zhang, Yanfeng and Xu, Qian and Liu, Haixu and Gao, Yuxiao and Li, Xiaohua and Yu, Ge},
  journal={Proceedings of the VLDB Endowment},
  volume={15},
  number={11},
  pages={2585--2598},
  year={2022},
  publisher={VLDB Endowment}
}

@inproceedings{li2023fisco,
  title={FISCO-BCOS: An Enterprise-grade Permissioned Blockchain System with High-performance},
  author={Li, Huizhong and Chen, Yujie and Shi, Xiang and Bai, Xingqiang and Mo, Nan and Li, Wenlin and Guo, Rui and Wang, Zhang and Sun, Yi},
  booktitle={Proceedings of the International Conference for High Performance Computing, Networking, Storage and Analysis},
  pages={1--17},
  year={2023}
}

@article{al2019blockchain,
  title={Blockchain in industries: A survey},
  author={Al-Jaroodi, Jameela and Mohamed, Nader},
  journal={IEEE access},
  volume={7},
  pages={36500--36515},
  year={2019},
  publisher={IEEE}
}

@article{spearman1961proof,
  title={The proof and measurement of association between two things.},
  author={Spearman, Charles},
  year={1961},
  publisher={Appleton-Century-Crofts}
}

@article{kendall1938new,
  title={A new measure of rank correlation},
  author={Kendall, Maurice G},
  journal={Biometrika},
  volume={30},
  number={1/2},
  pages={81--93},
  year={1938},
  publisher={JSTOR}
}

\end{document}